\documentclass[a4paper,12pt]{article}

\usepackage[margin=2.10cm]{geometry}
\usepackage[T1]{fontenc}
\usepackage{lmodern}
\usepackage{microtype}

\usepackage{amsmath, amssymb, amsthm, mathtools, bm, dsfont, bbm}
\usepackage{mathrsfs}

\usepackage{enumitem}
\usepackage{booktabs}
\usepackage{threeparttable}
\usepackage{multirow}
\usepackage{siunitx}

\usepackage{graphicx}
\usepackage{subcaption}
\usepackage{float}

\usepackage{natbib}
\usepackage[nottoc,numbib]{tocbibind}
\settocbibname{References}

\usepackage{hyperref}
\usepackage{enumitem}
\usepackage{cleveref}
\usepackage{bookmark}

\crefname{assumptioni}{Assumption}{Assumptions}
\Crefname{assumptioni}{Assumption}{Assumptions}

\hypersetup{
	colorlinks = true,
	citecolor  = blue,
	linkcolor  = blue,
	urlcolor   = blue,
	pdfauthor  = {Jonas Skjold Raaschou-Pedersen},
	pdftitle   = {Doubly Robust Instrumented Difference-in-Differences}
}

\theoremstyle{plain}

\newtheorem{proposition}{Proposition}
\newtheorem{lemma}{Lemma}

\theoremstyle{definition}

\newtheorem{assumption}{Assumption}

\theoremstyle{remark}
\newtheorem{remark}{Remark}

\setlist[itemize]{noitemsep, topsep=3pt}
\setlist[enumerate]{noitemsep, topsep=3pt}

\crefname{assumption}{assumption}{assumptions}
\Crefname{assumption}{Assumption}{Assumptions}
\crefname{theorem}{theorem}{theorems}
\Crefname{theorem}{Theorem}{Theorems}
\crefname{lemma}{lemma}{lemmas}
\Crefname{lemma}{Lemma}{Lemmas}

\usepackage{xcolor}

\newif\ifdraft
\drafttrue        

\ifdraft
	\newcommand{\draftnote}[1]{\textcolor{red}{[#1]}}
\else
	\newcommand{\draftnote}[1]{}
\fi

\usepackage{xcolor}
\usepackage{marginnote}

\newif\iffinal
\finalfalse   

\iffinal
	\newcommand{\jonas}[1]{}
	\newcommand{\todo}[1]{}
	\newcommand{\mynote}[1]{}
\else
	\newcommand{\jonas}[1]{\textcolor{blue}{[Jonas: #1]}}
	\newcommand{\todo}[1]{\textcolor{red}{[TODO: #1]}}
	\newcommand{\mynote}[1]{\marginnote{\scriptsize\textcolor{purple}{#1}}}
\fi

\crefname{enumi}{Assumption}{Assumptions}
\Crefname{enumi}{Assumption}{Assumptions}

\usepackage{xparse}

\NewDocumentEnvironment{proofref}{m}
{\noindent\textbf{Proof of \Cref{#1}. }\pushQED{\qed}\ignorespaces}
{\popQED}

\newlist{tricks}{enumerate}{2}

\setlist[tricks,1]{
	label=\textbf{Trick \arabic*.},
	ref=\arabic*,
	leftmargin=*,
	align=left,
	wide
}

\setlist[tricks,2]{
	label=\textbf{Trick \arabic{tricksi}\alph*.},
	ref=\arabic{tricksi}\alph*,
	leftmargin=2em,
	labelindent=2em,
	align=left,
	wide
}

\crefname{tricksi}{Trick}{Tricks}
\Crefname{tricksi}{Trick}{Tricks}
\crefname{tricksii}{Trick}{Tricks}
\Crefname{tricksii}{Trick}{Tricks}

\usepackage{tikz}
\usetikzlibrary{arrows.meta,positioning}

\usepackage{makecell}

\begin{document}

\title{
	Doubly Robust Instrumented Difference-in-Differences
}
\date{\today}
\author{
	Jonas Skjold Raaschou-Pedersen\thanks{
		Center for Social Data Science, University of Copenhagen.
		Email: \texttt{jsr-p@sodas.ku.dk}.
	} \thanks{
		I thank Andreas Bjerre-Nielsen and Anders Munch for helpful comments.
	}
}
\maketitle

\begin{abstract}
	We study estimation of the local average treatment effect on the treated
	($LATT$) in instrumented difference-in-differences (IDiD) designs with
	covariates and staggered instrument exposure. We derive the efficient influence function
	(EIF) of the target parameter in both panel and repeated cross-sections
	settings, allowing for two classes of control groups: never-exposed and
	not-yet-exposed.
	Building on the EIF, we construct doubly robust estimands and corresponding
	estimators from first principles. The resulting procedures are the IDiD
	analogues of the difference-in-differences (DiD) procedures in \cite{csa},
	targeting $LATT$ rather than $ATT$.
	We further establish a Bloom-type result under one-sided compliance and
	absorbing treatment, linking $LATT$ to a convex combination of
	exposure-cohort-specific $ATT(g, t)$ parameters,
	making the connection between IDiD and DiD explicit.
	Asymptotic properties are established under conditions on the remainder
	term and either Donsker conditions or via cross-fitting.
	We also construct double machine learning (DML) estimators for the
	$LATT$ in both data settings and show their equivalence to
	cross-fitted estimators.
	Simulations assess the double robustness and finite-sample performance of the
	proposed methods.
	An implementation is available in the Python package \texttt{idid}\footnote{
		Replication code is available at \url{https://github.com/jsr-p/idid}.
	}.
\end{abstract}


\newpage

\section{Introduction}
Difference-in-differences (DiD) is a central tool in applied econometrics for
estimating causal effects in non-experimental settings.
Recent work, often referred to as
the ``DiD renaissance'', has exposed important limitations of classical DiD
methods, particularly under treatment effect heterogeneity across cohorts and
over time. In such settings, coefficients from standard regression
specifications can lack a clear causal interpretation.

Two broad responses have emerged. One is to refine the underlying regression
model. The other is to begin with a well-defined, interpretable statistical
target, an \textit{estimand}, rather than starting from a regression model.
Under suitable identifying assumptions, this
estimand can be linked to observable quantities and estimated in a way that
admits a transparent causal interpretation.
A central object in this approach is the efficient influence function
(EIF), which is used to construct semiparametrically efficient estimators
and valid inference.
This perspective is adopted, for example, in the work of
\citet{sazhao} and \citet{csa}.

However, even modern DiD approaches can fail when the treatment is endogenous.
A classical remedy is the use of instrumental variables (IV). The integration
of IV with DiD remained underdeveloped until \citet{miyaji}, who introduced the
instrumented difference-in-differences (IDiD) framework with staggered
exposure to the instrument.

Adopting an estimand-based perspective, \citet{miyaji} defines the causal
parameter of interest, the cohort-specific time-varying local average treatment
effect on the treated ($LATT(e, t)$), and links it to an estimable quantity under
identifying assumptions.
However, the framework does not incorporate covariates, which are often
essential in empirical applications to support the plausibility of the
identifying assumptions.

In related work, \citet{csax} identify the $LATT$ parameter of
\citet{miyaji} with the inclusion of covariates as a special case of their
sequential conditional moment restriction framework for DiD.
They derive the efficient influence function of the parameter in the
panel data setting and construct a corresponding doubly robust estimator.
However, their analysis focuses on the case with a single exposure date to
the instrument.
Moreover, they leave the case of repeated cross-sections for future work.

This paper studies the general case with covariates and staggered exposure to
the instrument in both panel data and repeated cross-sections.
We derive the efficient influence function for the $LATT$ parameter in
both data settings.
The repeated cross-sections case is empirically relevant in settings where
balanced panels are unavailable, for example in ``trimmed panel'' data or
repeated survey samples.
A novel contribution is that the EIF is derived explicitly using the
approach developed by \citet{kennedy}.
Moreover, as done in \cite{miyaji}, the framework of \cite{csa} is extended to
the IDiD case allowing the use of not-yet-exposed units as controls, analogous
to the not-yet-treated control group in the staggered DiD setting.
Both control variables are handled generally by invoking either of the
corresponding identifying assumptions.

Using the derived EIFs, we construct doubly robust estimands and
corresponding estimators for the $LATT$ parameter in both data settings.
The construction follows the estimating equation approach:
the estimand solves the population moment condition implied by the EIF,
while the estimator solves its empirical counterpart.
This construction also underlies the doubly robust estimand and
estimators in \citet{sazhao} (and, by extension, \citet{csa}), although
implicitly.
The resulting doubly robust estimators have a structure closely related to
those proposed by \citet{sazhao} and \citet{csa}.
In particular, the estimator takes the form of a ratio of two
doubly robust estimators for ATT-type parameters.
The numerator corresponds to the outcome of interest and the denominator
to the treatment variable, with the instrument exposure variable playing
the role of the treatment indicator.
This structure arises naturally because the identified $LATT$ parameter
itself can be written as a ratio of two ATT-type parameters.
The estimators reduce to those of \cite{miyaji} in the case of no covariates.

The asymptotic behavior of the estimator is established via a decomposition
into an influence function term, an empirical process term, and a remainder
term (cf. \cite{kennedy}).
The remainder term is handled directly in the proofs.
The empirical process term is controlled either via Donsker class
assumptions or by employing cross-fitting, which permits the use of
flexible machine learning methods. We derive DML estimators, which we show are
equivalent to the cross-fitted estimator in \cite{kennedy}, for both data
settings.

A simulation study illustrates the finite-sample properties of the
proposed estimators.
A freely available implementation is provided in the Python package
\texttt{idid}.

\paragraph{Contributions}
This paper makes three contributions.
First, we derive the efficient influence function for the $LATT$ parameter
in the IDiD framework with covariates and
staggered instrument exposure.
The derivation covers both panel data and repeated cross-sections.
Second, using the EIF, we construct doubly robust estimands and estimators
for the $LATT$ parameter.
Third, the paper illustrates how the influence function derivation
strategy of \citet{kennedy} can be applied to econometric target
parameters.
This approach provides a practical alternative to classical tangent
space calculations \citep{newey1990,tsiatis} and ensures that the resulting
estimators target precisely the interpretable estimand specified
by the researcher.
We also construct
Neyman-orthogonal scores for DML estimation using the EIFs and relate IDiD with
staggered instrument exposure to DiD with staggered treatment via a
Bloom-type result.

\paragraph{Related literature}
This paper relates to several strands of the literature.

First, it contributes to the literature on IDiD.
\citet{miyaji} introduce the IDiD framework and define the $LATT$
parameter.
\citet{csax} extend this framework to allow for covariates and derive
a doubly robust estimator in the panel setting with a single exposure
date.
The present paper studies the general case with covariates and staggered
exposure to the instrument in both panel data and repeated cross-sections,
and derives corresponding doubly robust estimands and estimators.

Second, the paper relates to the literature on semiparametric
DiD estimation.
\citet{sazhao} and \citet{csa} develop doubly robust estimators for
ATT-type parameters in two-period and staggered adoption settings.
The estimators proposed here have a similar structure but tailored to the IDiD
framework.
Moreover, in deriving the influence function, doubly robust estimand, and
estimators for the LATT parameter, we simultaneously derive the corresponding
results for the ATT parameter, thereby recovering the results of \citet{sazhao}
and \citet{csa} from first principles.

Third, the paper connects to recent work at the intersection of modern
semiparametric methods and DiD/IDiD, as well as to recent work on doubly
robust IV estimators.
For example,
\cite{chang2020} develop DML estimators for DiD in the case of panel data and
repeated cross-sections,
\citet{deng} develop a TMLE estimator for the two-period
DiD parameter, while \citet{metalearner} propose a meta-learner
algorithm for the $LATT$ in two-period IDiD with panel data.
\citet{sloczynskiDoublyRobustEstimation2022} develop doubly robust
estimators for LATE and LATT in cross-sectional IV settings, whose
structure resembles the panel data estimator derived here.

More broadly, the paper relates to recent work emphasizing clearly defined
target parameters in causal inference and IV analysis,
including \citet{mogstad2024instrumental}, who distinguish between forward
and reverse engineering approaches to IV parameters.
In this paper, we take the forward engineering approach.
It also connects to the targeted learning literature
\cite{vdlrose,van2015statistics}, which emphasizes aligning the estimand
with the underlying scientific question and constructing estimators under
minimal modeling assumptions, rather than relying on potentially
misspecified parametric models. Finally, it relates to the double machine
learning literature \cite{chernozhukovDoubleDebiasedMachine2018a}, as we
derive two DML estimators.

\paragraph{Organization of the paper}
\Cref{sec:identification} introduces IDiD, the causal estimands, and the main
identification results.
\Cref{sec:weighting} studies aggregated effects.
\Cref{sec:estimation-and-inference} develops the doubly robust and DML estimators and
establishes their asymptotic properties. \Cref{sec:simulation} reports the
simulation results. Proofs and derivations are collected in the appendix.
\Cref{sec:influence-functions} provides a detailed account of how the influence
function calculations fit into the development of the paper.

\paragraph{Notation}
For a measurable function $f$, let $\Vert f \Vert_{q,P} = (\int |f|^q
	dP)^{1/q}$ denote its $L^q(P)$ norm.
Write $P_n f = n^{-1}\sum_{i=1}^n f(X_i)$ for the empirical average and $Pf =
	\int f\, dP$ for the expectation under $P$.
We also write $E[f(X)]$ for the expectation (under the relevant distribution) and
use the notation interchangeably where convenient.
If $f$ depends on a parameter $\tau$ and some nuisance $\eta$, we write $P_n
	f(\cdot;\tau, \eta) =
	n^{-1}\sum_{i=1}^n f(X_i;\tau, \eta)$ and $P f(\cdot;\tau, \eta) = \int
	f(\cdot;\tau, \eta)\, dP$.
The empirical process is written $\sqrt{n}(P_n - P)[f]$.
Calligraphic letters denote supports of random variables, e.g.\ $\mathcal{E}$, $\mathcal{D}$, and $\mathcal{Z}$ for $E$, $D$, and $Z$.
For an event $A$, $\mathbf{1}\{A\}$ denotes its indicator.
We write $O$ for a generic tuple of random variables that is context-dependent.
See also \Cref{sec:eif-derivation} for details on the notation for the
influence function operator, $\mathbb{IF}$.

\section{Identification}
\label{sec:identification}

\subsection{Setup}
We first introduce the notation used throughout the article, building on \cite{miyaji}.
We consider the general case of $\mathcal{T}$ periods.
Let $D_{t} \in \{0, 1\}$ denote treatment status and
$Z_{t} \in \{0, 1\}$ the instrument status.
Moreover, let
$D = (D_{1}, D_{2}, \ldots, D_{\mathcal{T}})$
and
$Z = (Z_{1}, Z_{2}, \ldots, Z_{\mathcal{T}})$
denote the treatment and instrument paths.

We make the following assumption on the instrument:

\begin{assumption}[Staggered adoption]
	\label{ass:staggered-instrument}
	$Z_{1} = 0$\footnote{If not already stated explicitly, equalities involving random variables are understood to hold almost
		surely.};
	and for $t = 2, \ldots, \mathcal{T}$:
	\begin{equation}
		\label{eq:staggered-instrument}
		Z_{t-1} = 1 \implies  Z_{t} = 1.
	\end{equation}
	\label{i:ms0}
\end{assumption}
\Cref{ass:staggered-instrument} enforces that no units are exposed to the
instrument in the first period and all units that are exposed in some period
stay exposed\footnote{
	This is analogous to Assumption 1 on the treatment variable assumed by
	\cite{csa}.
}.
\Cref{ass:staggered-instrument}
implies that the time period where the instrument switches on characterizes
the instrument path $Z$ completely. Because of this we define the
\textit{cohort exposure variable} $E := \min \{t \mid Z_{t} = 1\}$ and $E_{e}
	:=  \mathbf{1}\{E = e\}$.

To identify the target parameter, valid control groups are needed.
In the DiD literature, commonly used control groups are never-treated and
not-yet-treated units.
Here, we construct similar control groups but based on the instrument instead
of the treatment variable.
Hence, we define
\begin{equation}
	\label{eq:control-variables}
	C^{nev} := \mathbf{1}\{E = \infty\},
	\quad
	C^{nye}_{e, s} := \mathbf{1}\{E_{e} = 0, Z_{s} = 0\},
\end{equation}
for the \textit{never-exposed} and \textit{not-yet-exposed} control groups,
respectively\footnote{
	The control variable $C^{nye}_{e,s}$ is analogous to $(1-D_s)(1-G_g)$
	in \cite{csa}, but here corresponds to units not in the cohort exposed
	at $E = e$ and not yet exposed at time $s$, i.e., $Z_s = 0$.
	In the appendix, we provide a table comparing the different objects in
	the two-period DiD, staggered adoption DiD and our case of staggered
	IDiD; see \Cref{table:comparison}.
}.
Let $\bar{e} := \max_{i}E_{i}$;
in the case of never-exposed units, $\bar{e} = \infty$,
and in the case of only not-yet-exposed units, $\bar{e} < \infty$.
Denote the support of the exposure variable excluding $\bar{e}$ as
$\mathcal{E} := \mathrm{supp}(E) \setminus \bar{e} \subseteq \{2, 3, \ldots,
	\mathcal{T}\}$\footnote{
	Analogous to \citet{csa}, when there is a never-exposed cohort, $E = \infty$,
	$\mathcal{E}$ only excludes this.
	In the case of not-yet-exposed control groups only, we exclude the
	last exposed cohort because there are no available control groups for
	this cohort.
}.
Denote the \textit{generalized propensity scores} corresponding to the control
variable in use as
\begin{align}
	\label{eq:generalized-propensity-scores}
	p_{e}(X) := P(E_{e} = 1 \mid X, E_{e} + C^{nev} = 1),
	\quad
	p_{e, s}(X) := P(E_{e} = 1 \mid X, E_{e} + C^{nye}_{e, s} = 1).
\end{align}
The conditioning on $E_{e} + C = 1$ restricts attention to the relevant
$2 \times 2$ comparison: units are either in the exposed cohort, $E_{e} = 1$,
or in the control group, $C = 1$. Within each such slice, the objects behave
as in the corresponding two-period setup.

\paragraph{Potential treatment and outcomes}
Let $D_{t}(\infty)$ denote a unit's unexposed potential treatment at time
$t$ if they remain untreated through time period $\mathcal{T}$,
i.e., if they were not to be exposed to the instrument across all available time
periods.
For $e = 2, 3, \ldots, \mathcal{T}$, let $D_{t}(e)$
denote the potential treatment
that a given unit would experience at time $t$ if they were to first become
exposed to the instrument in time period $e$.
The observed and potential treatment for a given unit are related through
\begin{align}
	\label{eq:observed-treatment-status}
	D_{t} = D_{t}(\infty) + \sum_{2 \leq e \leq \mathcal{T}}
	[D_{t}(e) - D_{t}(\infty)] \cdot \mathbf{1}\{E = e\},
\end{align}
i.e., we only observe one potential treatment path for each unit.

Let $Y_{t}(d, z)$ be a given unit's potential outcome in period $t$ had they
been given treatment path $D = d$ and instrument path $Z = z$.
The analogy between the treatment and instrument for the DiD and IDiD
frameworks differs in the sense that the instrument does not affect the outcome
directly but only through the treatment. Specifically, the instrument creates
exogenous variation in the treatment that allows us to identify the effect of
the treatment on the outcome in the presence of hidden confounders.
This relation between the instrument and the potential outcomes is enforced in the
following assumption:
\begin{assumption}[No carryover and exclusion restriction in multiple time periods]
	\label{ass:no-carryover}
	\begin{align}
		\label{eq:no-carryover}
		\forall z \in \mathcal{Z},
		\forall d \in \mathcal{D},
		\forall t \in \{1, \ldots, \mathcal{T}\},
		\,
		Y_{t}(d, z) = Y_{t}(d_{t}),
	\end{align}
	\label{i:ms1}
\end{assumption}
\Cref{ass:no-carryover} means that the potential outcomes
at time $t$ only depend on the treatment variable at time $t$,
and that they do not depend directly on the instrument.
The latter is analogous to the exclusion restriction in the simple
cross-sectional IV design.

To arrive at an expression for the observed outcome,
let $Y_{t}(0)$ denote a given unit's potential outcome
at time $t$ if they are untreated at time period $t$,
and $Y_{t}(1)$ if they are treated at time period $t$.
With \eqref{eq:observed-treatment-status} we can write the observed
outcome in terms of the potential outcomes and treatments as:
\begin{align}
	\label{eq:observed-outcome}
	Y_{t} & = Y_{t}(0) + [Y_{t}(1) - Y_{t}(0)]D_{t}
	\\&=
	\label{eq:exposed-stag}
	Y_{t}(0) + [Y_{t}(1) - Y_{t}(0)][
	D_{t}(\infty) + \sum_{2 \leq e \leq \mathcal{T}}
	[D_{t}(e) - D_{t}(\infty)] \cdot \mathbf{1}\{E = e\}
	].
\end{align}
Using \eqref{eq:exposed-stag}, we define the \textit{exposed/unexposed
	outcomes} as:
\begin{align}
	\label{eq:exposed-unexposed-outcome}
	Y_{t}(D_{t}(E)) := \begin{cases}
		                   Y_{t}(D_{t}(e))
		                   := Y_{t}(0) + [Y_{t}(1) - Y_{t}(0)]D_{t}(e),      &
		                   E = e,                                              \\
		                   Y_{t}(D_{t}(\infty))
		                   := Y_{t}(0) + [Y_{t}(1) - Y_{t}(0)]D_{t}(\infty), &
		                   E = \infty.                                         \\
	                   \end{cases}
\end{align}
The exposed/unexposed outcomes\footnote{
As noted by \cite{miyaji}, the concept of exposed and unexposed outcomes
is not new. In the standard cross-sectional binary IV setup with potential
outcomes $Y(0), Y(1)$, potential treatments $D(0), D(1)$, and instrument $Z$,
the observed treatment can be written as
$D = D(0) + [D(1) - D(0)]Z$.
Substituting this into the observed outcome equation
$Y = Y(0) + [Y(1) - Y(0)]D$
yields the exposed and unexposed outcomes
$Y(D(Z)) := Y(0) + [Y(1) - Y(0)]D(Z)$ for $Z \in \{0,1\}$.
This is analogous to the staggered exposure setting in
\cref{eq:exposed-stag}, which leads to the exposed and unexposed outcomes in
\cref{eq:exposed-unexposed-outcome}.
} play a key role in the identification results,
\Cref{prop:identification-panel,prop:identification-rc}.

\paragraph{Sampling assumption}
Our results apply to both panel data and repeated cross-sections,
which are covered by the following assumption.
Let $T \in \{1,\ldots,\mathcal{T}\}$ denote the period a unit is
observed in the repeated cross-sections case.
\begin{assumption}[Random Sampling or Repeated Cross-Sections]
	\label{ass:sample}
	Assume either:
	\begin{enumerate}[label=(\Alph*), ref=\theassumption(\Alph*)]
		\item \label{ass:sample:panel}
		      The data are i.i.d. from the distribution of
		      $(Y_{1},\ldots,Y_{\mathcal{T}},D_{1},\ldots,D_{\mathcal{T}},
			      E_{2},\ldots,E_{\mathcal{T}},C^{nev}, X)$.
		\item \label{ass:sample:mixture}
		      Conditional on $T=t$, the data are i.i.d. from the distribution of
		      $(Y_t,D_t,E_2,\ldots,E_{\mathcal{T}},C^{nev},X)$
		      with $(E_{2}, E_{3}, \ldots, E_{\mathcal{T}}, C^{nev}, X)$ being
		      invariant to $T$.
	\end{enumerate}
\end{assumption}

\Cref{ass:sample:panel} implies that the observed data consist of panel
data, whereas \cref{ass:sample:mixture} implies that the observed
data are i.i.d. draws from the mixture distribution
\begin{align}
	\label{eq:mixture-csa}
	F_{M}(y, d, e_{2}, \ldots, e_{\mathcal{T}}, c, t, x)
	= \sum_{t=1}^{\mathcal{T}}
	\lambda_{t} \cdot
	F_{Y, D, E_{2}, E_{3}, \ldots, E_{\mathcal{T}}, C^{nev}, X \mid T}
	(y, d, e_{2}, \ldots, e_{\mathcal{T}}, c, x \mid t),
\end{align}
where $\lambda_{t} := P(T_{t} = 1)$ and $T_{t} := \mathbf{1}\{T = t\}$.
This mixture distribution setup is analogous to
\cite{abadieSemiparametricDifferenceinDifferencesEstimators2005a,sazhao,csa}
but here for the IDiD design with staggered exposure.

\begin{remark}
	All expectations in the repeated cross-sections case are taken
	with respect to the mixture measure $P_{M}$, although we suppress this
	subscript in the notation.

	Moreover, in the repeated cross-section setting, for each unit in the pooled
	sample we observe
	$(Y, D, E_{2},\ldots,E_{\mathcal{T}}, C^{nev}, T, X)$.
\end{remark}

\subsection{The LATT parameter}
The target parameter in this paper is the cohort-specific time-varying local
average treatment effect on the treated of \cite{miyaji}
\begin{align}
	\label{eq:latt}
	LATT(e, t)
	 & := E[Y_{t}(1) - Y_{t}(0) \mid E_{e} = 1, D_{t}(e) > D_{t}(\infty)].
\end{align}
\Cref{eq:latt} is the treatment effect, $Y_{t}(1) - Y_{t}(0)$,
averaged over the subpopulation of compliers, $D_{t}(e) > D_{t}(\infty)$,
and the units exposed to the instrument in period $e$, $E_{e} = 1$.
The parameter varies across cohorts $E$ and time $t$;
hence it allows us to answer questions related to the heterogeneity across
cohorts and time.
In \Cref{sec:weighting} we show how to aggregate the parameters into
aggregated effects; similar to how \cite{csa} aggregates their $ATT(g, t)$
parameters.

\subsection{Identifying assumptions}

\begin{assumption}[Monotonicity assumption in multiple time periods]
	\label{ass:monotonicity}
	For all $e \in \mathcal{E}$ and $t \geq e$:
	\begin{align*}
		P(D_{t}(e) \geq D_{t}(\infty) \mid X) = 1
		\ \text{a.s.}
	\end{align*}
\end{assumption}
\Cref{ass:monotonicity} is analogous to the monotonicity assumption in
cross-sectional IV and requires that the instrument, here the exposure
cohort variable $E$, affects the treatment in only one direction.
\begin{assumption}[No anticipation in the first stage]
	\label{ass:no-anticipation}
	For all
	$e \in \mathcal{E}$
	and $t < e$:
	\begin{align}
		\label{eq:no-anticipation}
		E[D_{t}(e) \mid X, E_{e} = 1]
		= E[D_{t}(\infty) \mid X, E_{e} = 1]
	\end{align}
\end{assumption}
\Cref{ass:no-anticipation} is analogous to the standard no-anticipation
assumption in DiD, with the exposure variable replacing the treatment,
and requires that exposure does not affect the treatment prior to the
exposure date.

\medskip\noindent
The following two assumptions use the never-exposed control group $C^{nev}$.
\begin{assumption}[Conditional parallel trends in the treatment based on "Never-Exposed" Group]
	\label{ass:cpt-treat-ne}
	For all $e, t$:
	\begin{align}
		\label{eq:cpt-treat-ne}
		E[D_{t}(\infty)-  D_{t-1}(\infty) \mid X, E_{e} = 1]
		= E[D_{t}(\infty)-  D_{t-1}(\infty) \mid X, C^{nev} = 1]
		\; \text{a.s.}
	\end{align}
\end{assumption}
\begin{assumption}[Conditional parallel trends in the unexposed outcome based
		on "Never-Exposed" Groups]
	\label{ass:cpt-unexpout-ne}
	For all $e, t$:
	\begin{align}
		\label{eq:cpt-unexpout-ne}
		E[Y_{t}(D_{t}(\infty))-  Y_{t-1}(D_{t-1}(\infty)) \mid X, E_{e} = 1]
		= E[Y_{t}(D_{t}(\infty))-  Y_{t-1}(D_{t-1}(\infty)) \mid X, C^{nev} = 1]
		\; \text{a.s.}
	\end{align}
\end{assumption}

\medskip\noindent
The following two assumptions are analogous to the previous two, but use
the not-yet-exposed control group $C^{nye}_{e, s}$.
\begin{assumption}[Conditional parallel trends in the treatment based on "Not-Yet-Exposed" Groups]
	\label{ass:cpt-treat-nye}
	For all $e, t$:
	\begin{align}
		\label{eq:ptstagtreatmentX}
		E[D_{t}(\infty)-  D_{t-1}(\infty) \mid X, E_{e} = 1]
		= E[D_{t}(\infty)-  D_{t-1}(\infty) \mid X, C^{nye}_{e, s} = 1]
		\; \text{a.s.}
	\end{align}
\end{assumption}
\begin{assumption}[Conditional parallel trends in the unexposed outcome based on "Not-Yet-Exposed" Groups]
	\label{ass:cpt-unexpout-nye}
	For all $e, t$:
	\begin{align}
		\label{eq:ptstagoutcomeX}
		E[Y_{t}(D_{t}(\infty))-  Y_{t-1}(D_{t-1}(\infty)) \mid X, E_{e} = 1]
		= E[Y_{t}(D_{t}(\infty))-  Y_{t-1}(D_{t-1}(\infty)) \mid X, C^{nye}_{e, s} = 1]
		\; \text{a.s.}
	\end{align}
\end{assumption}
In the identification proofs we invoke either pair of the four assumptions above depending on
the control variable used.
If using the never-exposed control group $C^{nev}$, we assume
\Cref{ass:cpt-treat-ne,ass:cpt-unexpout-ne},
and if using the not-yet-exposed control group $ C^{nye}_{e,t}$, we assume
\Cref{ass:cpt-treat-nye,ass:cpt-unexpout-nye}.
Writing $C$ for a generic control variable allows us to encompass both control
variables in the identification arguments.

\begin{assumption}[Overlap]
	\label{ass:overlap}
	For each $t \in \{2, 3, \ldots, \mathcal{T}\}$, $e \in \mathcal{E}$,
	there exist some $c > 0$ such that $P(E_{e} = 1) > c$
	and $P(E_{e} = 1 \mid X) < 1 - c$ a.s.
\end{assumption}
\Cref{ass:overlap} is a standard overlap condition. It requires that, for
each exposure cohort, the probability of exposure is bounded away from
zero, and from one conditional on covariates. This ensures that all relevant
subpopulations have a non-negligible probability of being both exposed
and unexposed, which is necessary for identification and stable
estimation.

\subsection{Nonparametric identification of the group-time LATT parameter}
\label{sec:nonparametric-identification}

\subsubsection{Notation}

To derive the main identification result, we introduce the central objects and
notation below. Let $V_t$ be a generic random variable and define
\[
	\Delta_{t-e+1} V_t := V_t - V_{e-1},
\]
i.e.\ the change of $V_{t}$ between period $t$ and the pre-exposure period
$e-1$.

\medskip
\noindent
\textit{Panel data.}
For panel data, define the outcome regression functions for treated,
never-exposed, and not-yet-exposed units as
\begin{align}
	\label{eq:mean-functions-panel}
	m_{e, t}^{trt, p}(X) & := E[\Delta_{t-e+1} Y_t \mid X, E_e = 1],
	\\
	\nonumber
	m_{e, t}^{nev, p}(X) & := E[\Delta_{t-e+1} Y_t \mid X, C^{nev} = 1],
	\quad
	m_{e, t}^{nye, p}(X) := E[\Delta_{t-e+1} Y_t \mid X, C^{nye}_{e,t} = 1].
\end{align}
Analogously, define $g_{e, t}^{trt, p}(X)$, $g_{e, t}^{nev, p}(X)$, and
$g_{e, t}^{nye, p}(X)$ with outcome $\Delta_{t-e+1} D_t$.

\medskip
\noindent
\textit{Repeated cross-sections.}
For repeated cross-sections, the outcome regression functions are defined as
\begin{align}
	\label{eq:mean-functions-rc}
	m_{e, t}^{trt, rc}(X) & := E[Y \mid X, E_e = 1, T = t],
	\\
	\nonumber
	m_{e, t}^{nev, rc}(X) & := E[Y \mid X, C^{nev} = 1, T = t],
	\quad
	m_{e, s, t}^{nye, rc}(X) := E[Y \mid X, C^{nye}_{e,s} = 1, T = t],
\end{align}
with corresponding definitions
$g_{e, t}^{trt, rc}(X)$, $g_{e, t}^{nev, rc}(X)$, and
$g_{e, s, t}^{nye, rc}(X)$ obtained by replacing $Y$ with $D$.

\paragraph{Encompassing both never exposed and not-yet-exposed}
A generic control variable is written as $C$ and the corresponding
propensity as $p(X)$.
Below we write $m_{e,  t}^{c, p}(X), g_{e, t}^{c, p}(X)$
for a generic control variable outcome regression function in the case of panel
data,
and similarly
$m_{e, s, t}^{c, rc}(X), g_{e, s, t}^{c, rc}(X)$
in the case of repeated cross-sections.

The proofs for the not-yet-exposed and never-exposed cases are identical,
differing only in the control indicator and propensity score.
To avoid repetition, we present a unified argument covering both cases.
The results follow under \Cref{ass:cpt-treat-ne,ass:cpt-unexpout-ne}
for the never-exposed case and under
\Cref{ass:cpt-treat-nye,ass:cpt-unexpout-nye}
for the not-yet-exposed case.
This is done for both the panel-data and repeated-cross-sections settings.

\subsubsection{Panel Data}
\label{sec:identification-panel}
\begin{proposition}[Identification Panel Data]
	\label{prop:identification-panel}
	Fix $e \in \mathcal{E}$ and $t \ge e$.
	Suppose \Cref{ass:sample:panel,ass:staggered-instrument,ass:no-carryover,ass:monotonicity,ass:no-anticipation,ass:overlap} hold.
	Let $C$ denote a generic control-group indicator, where either $C = C^{nev}
		\quad \text{or} \quad C = C^{nye}_{e,t}.$ If $C = C^{nev}$, additionally
	assume
	\Cref{ass:cpt-treat-ne,ass:cpt-unexpout-ne}.
	If $C = C^{nye}_{e,t}$, additionally assume
	\Cref{ass:cpt-treat-nye,ass:cpt-unexpout-nye}.
	Let $m_{e, t}^{c, p}(X)$ and $g_{e, t}^{c, p}(X)$
	be the mean functions \cref{eq:mean-functions-panel} for a generic control
	variable.
	Then the cohort-specific time-varying local average treatment effect on the
	treated \cref{eq:latt} is identified by
	\begin{align}
		\label{eq:identification-panel}
		LATT(e,t)
		=
		\frac{E[
					m_{e, t}^{trt, p}(X) - m_{e, t}^{c, p}(X) \mid E_{e} = 1
				]
		}{E[
					g_{e, t}^{trt, p}(X) - g_{e, t}^{c, p}(X) \mid E_{e} = 1
				]} =: \tau^{p}_{e, t}.
	\end{align}
\end{proposition}

\medskip\noindent
Proof in \Cref{sec:proofs-panel}.

\begin{remark}
	The identification result in \cref{eq:identification-panel} is closely
	related to \cite{csax}, but adapted here to also including the case of
	not-yet-exposed controls. Its
	general structure also mirrors the results of
	\cite{frolichExploitingRegionalTreatment2010} and
	\cite{sloczynskiDoublyRobustEstimation2022}, derived here for panel
	data in the IDiD framework allowing for the different control group variables
	in \cref{eq:control-variables}.
\end{remark}

\paragraph{Influence functions and construction of the DR Estimands}
Below we derive DR estimands and corresponding DR estimators for the $LATT(e,
	t)$ in both data settings.
The steps are as follows:
\begin{enumerate}
	\item Identification yields an initial estimand which is a ratio of two
	      ATT-type parameters, e.g.
	      \cref{eq:identification-panel,eq:identification-rc}.
	\item We derive the efficient influence function (EIF) of the initial estimand\footnote{
		      Kennedy distinguishes between the influence function for a parameter
		      and that for an estimator, referring to the former as an
		      \textit{influence curve}. We do not adopt this distinction in terminology.
	      }.
	      This EIF involves weights that are normalized to Hájek-type weights
	      that are variance stabilizing cf. \citep{sazhao}.
	      \label{i:if-procedure-2}
	\item Solving the population version of the estimating equation using the EIF
	      for the target parameter yields the doubly robust (DR) estimand.
	      \label{i:if-procedure-3}
	\item We derive the EIF of the DR estimand\footnote{
		      The normalization in Step \ref{i:if-procedure-2} requires an adjusted
		      influence function that accounts for the normalization.
	      }.
	\item The estimator is the plug-in estimator of the DR estimand, with
	      inference based on the estimated influence function from step 4.
\end{enumerate}
As the above procedure shows, the EIF of the target
parameter is central to deriving the DR estimands. A key
feature is that the identified LATT parameters are ratios of ATT
components and that the staggered exposure setting reduces to
two-period comparisons. This allows us to start from the canonical
two-period DiD case targeting the ATT parameter and derive the DR
estimands and EIFs in this setting. As part of this derivation, we
recover the DR DiD estimands of \cite{sazhao} from first principles. We
then exploit the ratio structure of the LATT parameters together with
the EIF machinery of \cite{kennedy} developed for our setting in
\Cref{sec:influence-functions} to combine the components. Constructing
the resulting estimators requires additional work and is taken up in
\Cref{sec:estimation-and-inference}.

\paragraph{Weights definition}
As noted in Step~\ref{i:if-procedure-3} of the procedure, the doubly robust
estimands and their corresponding EIFs rely on Hájek-type normalized weights.
The construction of these weights in both sampling settings follows the same
principle as in \cite{sazhao,csa}: the control weights are normalized to sum to one in
sample (in contrast to the unnormalized control weights appearing in the EIF
derived in Step~\ref{i:if-procedure-2}).

The weights in the panel data setting are:
\begin{align}
	\label{eq:weights1-panel-norm}
	w^{trt,p}_{e}
	 & := \frac{E_{e}}{\rho},
	\quad
	\rho := E[E_{e}]
	\\
	\label{eq:weights-nev-panel-norm}
	w^{nev,p}_{e}
	 & :=
	C^{nev}\frac{p_{e}(X)}{1 - p_{e}(X)}
	/
	E\left[C^{nev}\frac{p_{e}(X)}{1 - p_{e}(X)}\right]
	\\
	\label{eq:weights-nye-panel-norm}
	w^{nye,p}_{e, t}
	 & :=
	C^{nye}_{e, t}\frac{p_{e, t}(X)}{1 - p_{e, t}(X)}
	/
	E\left[C^{nye}_{e, t}\frac{p_{e, t}(X)}{1 - p_{e, t}(X)}\right]
\end{align}
The weights in the repeated cross-sections setting are:
\begin{align}
	\label{eq:weights-trt-rc-norm}
	w^{trt,rc}_{e, t}
	 & := \frac{E_{e} \cdot \mathbf{1}\{T = t\}}{ \rho^{trt, rc}_{e, t} },
	\quad
	\rho^{trt, rc}_{e, t}  := E[E_{e} \cdot \mathbf{1}\{T = t\}]
	\\
	\label{eq:weights-nev-rc-norm}
	w^{nev,rc}_{e, t}
	 & :=
	\frac{C^{nev} \cdot \mathbf{1}\{T = t\}p_{e}(X)}{1 - p_{e}(X)}
	/
	E\left[
	\frac{C^{nev} \cdot \mathbf{1}\{T = t\}p_{e}(X)}{1 - p_{e}(X)}
	\right]
	\\
	\label{eq:weights-nye-rc-norm}
	w^{nye,rc}_{e, s, t}
	 & :=
	\frac{C^{nye}_{e, s} \cdot \mathbf{1}\{T = t\}p_{e, s}(X)}{1 - p_{e, s}(X)}
	/
	E\left[
	\frac{C^{nye}_{e, s} \cdot \mathbf{1}\{T = t\}p_{e, s}(X)}{1 - p_{e, s}(X)}
	\right].
\end{align}
In the repeated cross-sections setting, we further define:
\begin{align}
	\label{eq:weights-trt-joint-rc}
	w^{trt,rc}_{e}
	 & :=
	w^{trt,rc}_{e, t} - w^{trt,rc}_{e, e - 1},
	\\
	\label{eq:weights-control-joint-rc}
	w^{c,rc}_{e}
	 & :=
	w^{c,rc}_{e, t, t}
	-
	w^{c,rc}_{e, t, e-1}
\end{align}
for $w^{c,rc}_{e, t, s}$ a generic control-weight.

\paragraph{Estimand and EIF}
\begin{proposition}
	\label{prop:if-latt-panel}
	Fix $e \in \mathcal{E}$ and $t \ge e$.
	Suppose the conditions of \Cref{prop:identification-panel} hold.
	Let $C$ denote a generic control-group indicator, where either
	$C = C^{nev}$ or $C = C^{nye}_{e,t}$, and let $p(X)$ denote the corresponding
	propensity score.
	Write $m_{e, t}^{c, p}(X)$ and $g_{e, t}^{c, p}(X)$
	for the mean functions in \cref{eq:mean-functions-panel}
	corresponding to the chosen control group.
	For \cref{eq:identification-panel}, define
	$\tau^{p}_{e, t} =: \tau^{p, num}_{e, t}/\tau^{p, den}_{e, t}$.
	Then the efficient influence function for $\tau^{p}_{e, t}$ is
	\begin{align}
		\label{eq:if-latt-panel}
		 & \varphi^{p}(O; \tau^{p}_{e, t}, \eta^{p}_{e, t})
		\\
		\nonumber
		 & =
		\frac{1}{\tau^{p, den}_{e, t}}
		[
		\{ w^{trt,p}_{e} - w^{c,p}_{e, t} \}
		\{\Delta_{t-e+1}Y_{t} - m_{e, t}^{c, p}(X)\}
		- \tau^{p}_{e, t}
		\{ w^{trt,p}_{e} - w^{c,p}_{e, t} \}
		\{\Delta_{t-e+1}D_{t} - g_{e, t}^{c, p}(X)\}
		],
	\end{align}
	and the nonparametric efficiency bound
	equals
	$E[\{\varphi^{p}(O; \tau^{p}_{e, t}, \eta^{p}_{e, t})\}^{2}]$,
	where $\eta^{p}_{e, t} = (\rho, p, m_{e, t}^{c, p}, g_{e, t}^{c, p})$
	is the nuisance-function tuple.
\end{proposition}

\medskip\noindent
Proof in \Cref{sec:latt-parameter-if-panel}.

\paragraph{DR estimands and estimator decomposition}
In the following, we denote a generic estimand by $\tau$
and its plug-in estimator by $\hat{\tau}$.
A generic decomposition of this plug-in estimator into a CLT, empirical process and
remainder term,
similar to \cite{kennedy}\footnote{
	We depart from the notation in \cite{kennedy} by indexing the EIF explicitly
	by the target parameter and nuisance functions, $\varphi(\cdot; \tau, \eta)$,
	rather than by the distribution $P$ alone.
},
can be written as:
\begin{align}
	\label{eq:decomposition}
	 & \sqrt{n}(\hat{\tau} - \tau)
	\\
	\nonumber
	 & =
	\sqrt{n} (P_{n} - P) \varphi(\cdot; \tau, \eta)
	+
	\sqrt{n} (P_{n} - P)
	[
		\varphi(\cdot; \tau, \hat{\eta})
		-
		\varphi(\cdot; \tau, \eta)
	]
	+ \sqrt{n} P \varphi(\cdot; \tau, \hat{\eta})
	+
	o_{P}(1)
\end{align}
where $\varphi(\cdot; \tau, \eta)$ is the influence function
and $\eta$ is a tuple of nuisance parameters.
We derive our doubly robust estimands by using the EIF, $\varphi(\cdot; \tau,
	\eta)$, corresponding to $\tau$ and solving
\begin{align}
	\label{eq:dr-estimand-foc}
	P\varphi(\cdot; \tau^{dr}, \eta) = 0
\end{align}
for the new target parameter $\tau^{dr}$, where the superscript "dr" means
doubly robust.
This is also how the estimands of \cite{sazhao} are derived (although not
explicitly shown in their paper).
In \Cref{sec:influence-functions} we show this from first principles, hence
demystifying where the estimands come from, building
up to our DR estimands for the $LATT$ parameter in both data settings.
The resulting estimators are of the "estimating equation form", i.e., for
$\hat{\eta}$ a generic estimator of the nuisance, the estimator, $\hat{\tau}$,
solves:
\begin{align}
	\label{eq:estimating-equation}
	P_{n}\varphi(\cdot; \hat{\tau}, \hat{\eta}) = 0.
\end{align}
Note that the DR estimands, e.g. the $\tau^{dr}$ estimand found through
\cref{eq:dr-estimand-foc}, have their own influence functions,
which differ from those of the original parameter, e.g. $\tau$.
Denote this influence function for the DR estimand as $\varphi^{dr}(\cdot;
	\tau^{dr}, \eta^{dr})$.
We use this influence function to conduct inference for the DR estimators.
This is also what \cite{sazhao} does.
However, \cite{sazhao} also study the estimation effects arising from the
nuisance function estimators of $\eta^{dr}$.
The estimation effects arise from their linearization of their estimators and
entails Taylor expanding the components building up to their influence
function.
In the IDiD setting, the estimand and estimators are ratios and hence the
Taylor expansion becomes significantly more tedious to derive.
Hence, in this paper, we do not pursue this approach.
Instead, we take the more general approach of either assuming Donsker
conditions or using cross-fitting to tame the empirical process term in
\cref{eq:decomposition} cf. \cite{kennedy}.
The remainder term is handled explicitly in the proofs; see
\Cref{prop:remainder-term-panel,prop:remainder-term-rc}.
A drawback of taking the more general approach is that our estimator does not
inherit the DR-for-inference property as in \cite{sazhao}; we leave this for
future work (see also \cite{dukesvansteelandt}).

\paragraph{DR Estimand and EIF}
\begin{proposition}
	\label{prop:dr-estimand-panel}
	Fix $e \in \mathcal{E}$ and $t \ge e$.
	Suppose the conditions of \Cref{prop:identification-panel} hold.
	Then the doubly robust panel-data estimand for $LATT(e,t)$ is
	\begin{align}
		\label{eq:dr-estimand-panel}
		\tau^{dr, p}_{e, t}
		 & = \frac{
		E[
		\{ w^{trt,p}_{e} - w^{c,p}_{e, t} \}
		\{\Delta_{t-e+1}Y_{t} - m_{e, t}^{c, p}(X)\}
		]
		}{
		E[
		\{ w^{trt,p}_{e} - w^{c,p}_{e, t} \}
		\{\Delta_{t-e+1}D_{t} - g_{e, t}^{c, p}(X)\}
		]
		}.
	\end{align}
\end{proposition}

\medskip\noindent
Proof in \Cref{sec:dr-estimand-panel}.

\begin{remark}
	\leavevmode
	The estimand in \cref{eq:dr-estimand-panel} resembles a ratio of two
	$ATT_{dr}(g,t;0)$ estimands from \cite{csa}, with the treatment replaced by
	$E_e$ and $D_t$ used as the outcome in the denominator.

	It is also related to the
	cross-sectional LATT estimand in
	\cite{sloczynskiDoublyRobustEstimation2022}. As in that work, the choice
	of nuisance models follows the outcome structure: for count outcomes, a
	Poisson model for $m^{c,p}(X)$ is natural, while for the denominator, a
	binomial mean model for $g^{c,p}(X)$ is appropriate.

	The resulting plug-in estimator, discussed in
	\Cref{sec:estimation-and-inference}, is related to the estimator
	of LATE in \cite{tan2006} in the cross-sectional IV case.
\end{remark}

\begin{proposition}
	\label{prop:dr-estimand-weighted-panel}
	Fix $e \in \mathcal{E}$ and $t \ge e$.
	Let $C$ denote a generic control-group indicator, where either
	$C = C^{nev}$ or $C = C^{nye}_{e,t}$.
	Write $m_{e, t}^{c, p}(X)$ and $g_{e, t}^{c, p}(X)$
	for the corresponding control-group mean functions.
	Define the residuals
	\begin{align*}
		\varepsilon^{Y, p}_{e, t} := \Delta_{t-e+1}Y_{t} - m_{e, t}^{c, p}(X),
		\quad
		\varepsilon^{D, p}_{e, t} := \Delta_{t-e+1}D_{t} - g_{e, t}^{c, p}(X),
	\end{align*}
	and
	$\tau^{dr, p}_{e, t}
		=: \tau^{dr, p, num}_{e, t}/\tau^{dr, p, den}_{e, t}$.
	Let $w^{c,p}_{e, t}$ be the control-weight for the generic control variable.
	Then the influence function of the doubly robust estimand
	\cref{eq:dr-estimand-panel} is
	\begin{align}
		\label{eq:if-dr-estimand-panel}
		\varphi^{dr, p}(O; \tau^{dr, p}_{e, t}, \eta^{dr, p}_{e, t})
		 & =
		\frac{1}{\tau^{dr, p, den}_{e, t}}
		\Big[
			w^{trt,p}_{e}
			\{
			\varepsilon^{Y, p}_{e, t}
			-
			E[ w^{trt,p}_{e} \varepsilon^{Y, p}_{e, t} ]
			\}
			-
			w^{c,p}_{e, t}
			\{
			\varepsilon^{Y, p}_{e, t}
			-
			E[ w^{c,p}_{e, t} \varepsilon^{Y, p}_{e, t} ]
			\}
		\\ & \quad
			\nonumber
			- \tau^{dr, p}_{e, t}
			(w^{trt,p}_{e}
			\{
			\varepsilon^{D, p}_{e, t}
			-
			E[ w^{trt,p}_{e} \varepsilon^{D, p}_{e, t} ]
			\}
			-
			w^{c,p}_{e, t}
			\{
			\varepsilon^{D, p}_{e, t}
			-
			E[ w^{c,p}_{e, t} \varepsilon^{D, p}_{e, t} ]
			\})
			\Big]
	\end{align}
	where $\eta^{dr, p}_{e, t}
		=
		(\rho, p, m_{e, t}^{c, p}, g_{e, t}^{c, p})$
	is the tuple of nuisance functions.
\end{proposition}
\medskip\noindent
Proof in \Cref{sec:eif-dr-estimand-panel}.

\subsubsection{Repeated Cross-Sections}

\begin{proposition}[Identification in repeated cross-sections]
	\label{prop:identification-rc}
	Fix $e \in \mathcal{E}$ and $t \ge e$.
	Suppose
	\Cref{ass:sample:mixture,ass:staggered-instrument,ass:no-carryover,ass:monotonicity,ass:no-anticipation,ass:overlap}
	hold.
	Let $C$ denote a generic control-group indicator, where either $C = C^{nev}
		\quad \text{or} \quad
		C = C^{nye}_{e,t}$, and let $p(X)$ denote the corresponding
	propensity score.
	Write $m_{e, s, t}^{c, rc}(X)$ and $g_{e, s, t}^{c, rc}(X)$
	for the mean functions in \cref{eq:mean-functions-rc} corresponding to the
	chosen control group (in the never-exposed case, the middle subscript $s$ is
	redundant).
	If $C = C^{nev}$, additionally assume
	\Cref{ass:cpt-treat-ne,ass:cpt-unexpout-ne}.
	If $C = C^{nye}_{e,t}$, additionally assume
	\Cref{ass:cpt-treat-nye,ass:cpt-unexpout-nye}.
	Then the cohort-time local average treatment effect on the treated is identified by
	\begin{align}
		\label{eq:identification-rc}
		LATT(e,t)
		=
		\frac{
		E[
		m_{e,t}^{trt,rc}(X)
		- m_{e,e-1}^{trt,rc}(X)
		-
		[
		m_{e, t, t}^{c,rc}(X)
		- m_{e, t, e-1}^{c,rc}(X)
		]
		\mid E_e = 1 ]
		}{
		E\left[
		g_{e,t}^{trt,rc}(X)
		- g_{e,e-1}^{trt,rc}(X)
		-
		[
		g_{e, t, t}^{c,rc}(X)
		- g_{e, t, e-1}^{c,rc}(X)
		]
		\mid  E_e = 1
		\right]
		} =: \tau^{rc}_{e, t}.
	\end{align}
\end{proposition}

\medskip\noindent
Proof in \Cref{sec:proofs-rc}.

\paragraph{Estimand and EIF}
In the repeated cross-sections case, we denote generic mean functions
for both control groups in \cref{eq:mean-functions-rc} by
$m^{c,rc}_{e,s,T}(X)$ and $g^{c,rc}_{e,s,T}(X)$ (again, in the never-exposed
case, the index $s$ is not used). This notation allows us to treat both
cases jointly.
Using this notation, we define the control mean functions unified across
periods $e-1$ and $t$ as:
\begin{align}
	\label{eq:mean-function-m-period-rc-control}
	m_{e, Y}^{c,rc}(X)
	 & :=
	\mathbf{1}\{T = e-1\} \cdot m_{e, t, e-1}^{c, rc}(X)
	+ \mathbf{1}\{T = t\} \cdot m_{e, t, t}^{c, rc}(X),
	\\
	\label{eq:mean-function-g-period-rc-control}
	g_{e, D}^{c,rc}(X)
	 & :=
	\mathbf{1}\{T = e-1\} \cdot g_{e, t, e-1}^{c, rc}(X)
	+ \mathbf{1}\{T = t\} \cdot g_{e, t, t}^{c, rc}(X),
\end{align}
and
\begin{align}
	\label{eq:mean-function-both-period-rc}
	m_{e,Y}^{trt,rc}(X)
	 & :=
	\mathbf{1}\{T = e-1\} \cdot m_{e, e - 1}^{trt, rc}(X)
	+ \mathbf{1}\{T = t\} \cdot m_{e, t}^{trt, rc}(X),
	\\
	g_{e,D}^{trt,rc}(X)
	 & :=
	\mathbf{1}\{T = e-1\} \cdot g_{e, e - 1}^{trt, rc}(X)
	+ \mathbf{1}\{T = t\} \cdot g_{e, t}^{trt, rc}(X),
\end{align}
where implicitly, because we only consider $2 \times 2$ comparisons, $\mathbf{1}\{T = e
	- 1\} + \mathbf{1}\{T = t\} = 1$.

\begin{proposition}
	\label{prop:if-latt-rc}
	Let $C$ be either of the control variables and $p(X)$ the corresponding
	propensity.
	Suppose the conditions of \Cref{prop:identification-rc} hold. For
	\cref{eq:identification-rc},
	define $\tau^{rc}_{e, t} =: \tau^{rc, num}_{e, t}/\tau^{rc, den}_{e, t}$.
	Then the efficient influence function
	for the LATT when repeated cross-sections are available
	equals:
	\begin{align}
		\label{eq:if-latt-rc}
		\varphi^{rc}(O; \tau^{rc}_{e, t}, \eta^{rc}_{e, t})
		          & =
		\frac{1}{\tau^{rc, den}_{e, t}}
		[
		w^{trt,rc}_{e} \{Y - m_{e, Y}^{trt, rc}(X)\}
		-
		w^{c,rc}_{e}
		\{Y - m_{e, Y}^{c,rc}(X)\}
		\\
		\nonumber
		\nonumber & \quad
		+ \frac{E_{e}}{\rho}
		\{
		m_{e,t}^{trt, rc}(X) - m_{e,e-1}^{trt, rc}(X)
		-
		[
		m_{e, t, t}^{c, rc}(X) - m_{e, t, e-1}^{c, rc}(X)
		]
		\}
		\\
		\nonumber & \quad
		\nonumber
		- \tau^{rc}_{e, t}
		(
		w^{trt,rc}_{e} \{D - g_{e, Y}^{trt, rc}(X)\}
		-
		w^{c,rc}_{e}
		\{D - g_{e, D}^{c,rc}(X)\}
		\\
		\nonumber
		\nonumber & \quad
		+ \frac{E_{e}}{\rho}
		\{
		g_{e,t}^{trt, rc}(X) - g_{e,e-1}^{trt, rc}(X)
		-
		[ g_{e, t, t}^{c, rc}(X) - g_{e, t, e-1}^{c, rc}(X) ]
		\}
		)
		]
	\end{align}
	where the nonparametric efficiency bound equals
	$E[\{\varphi^{rc}(O; \tau^{rc}_{e, t}, \eta^{rc}_{e, t})\}^{2}]$,
	and where $\eta^{rc}_{e, t} = (
		\rho,
		\lambda_{e - 1},
		\lambda_{t},
		p,
		m_{e, t}^{trt, rc},
		m_{e, e - 1}^{trt, rc},
		g_{e, t}^{trt, rc},
		g_{e, e - 1}^{trt, rc}
		m_{e, t, t}^{c, rc},
		m_{e, t, e-1}^{c, rc},
		g_{e, t, t}^{c, rc},
		g_{e, t, e-1}^{c, rc}
		)$
	is the tuple of nuisance functions.
\end{proposition}

\medskip\noindent
Proof in \Cref{sec:latt-parameter-if-rc}.

\paragraph{DR Estimand and EIF}
\begin{proposition}
	\label{prop:dr-estimand-rc}
	Fix $e \in \mathcal{E}$ and $t \ge e$.
	Suppose the conditions of \Cref{prop:identification-rc} hold.
	Then the doubly robust repeated-cross-section estimand for $LATT(e,t)$ is
	\begin{align}
		\label{eq:dr-estimand-rc}
		\tau^{dr, rc}_{e, t}
		 & =
		\frac{
			E [
					\{w^{trt,rc}_{e} - w^{c,rc}_{e}\}
					\{Y - m_{e, Y}^{c,rc}(X)\}
				] + \kappa_{e, t}^{Y, rc}
		}{
			E [
					\{w^{trt,rc}_{e} - w^{c,rc}_{e}\}
					\{D - g_{e, Y}^{c,rc}(X)\}
				] + \kappa_{e, t}^{D, rc}
		},
	\end{align}
	where
	\begin{align}
		\label{eq:kappa-Y-rc}
		 & \kappa_{e, t}^{Y, rc}
		\\
		\nonumber
		 & :=
		E [ m_{e, t}^{trt, rc}(X)
				-  m_{e, t, t}^{c, rc}(X) \mid E_{e} = 1 ]
		-
		E[ m_{e, t}^{trt, rc}(X)
				-  m_{e, t, t}^{c, rc}(X) \mid E_{e} = 1, T = t ]
		\\ & \quad
		\nonumber
		-
		\left(
		E [ m_{e, e-1}^{trt, rc}(X)
					-  m_{e, t, e-1}^{c, rc}(X) \mid E_{e} = 1 ]
		-
		E[m_{e, e-1}^{trt, rc}(X)
					-  m_{e, t, e-1}^{c, rc}(X) \mid E_{e} = 1, T = e - 1]
		\right),
		\\
		\label{eq:kappa-D-rc}
		 & \kappa_{e, t}^{D, rc}
		\\
		\nonumber
		 & :=
		E [ g_{e, t}^{trt, rc}(X)
				-  g_{e, t, t}^{c, rc}(X) \mid E_{e} = 1 ]
		-
		E[ g_{e, t}^{trt, rc}(X)
				-  g_{e, t, t}^{c, rc}(X) \mid E_{e} = 1, T = t ]
		\\ & \quad
		\nonumber
		-
		\left(
		E[ g_{e, e-1}^{trt, rc}(X)
					-  g_{e, t, e-1}^{c, rc}(X) \mid E_{e} = 1 ]
		-
		E[g_{e, e-1}^{trt, rc}(X)
					-  g_{e, t, e-1}^{c, rc}(X) \mid E_{e} = 1, T = e - 1]
		\right).
	\end{align}
\end{proposition}

\medskip\noindent
Proof in \Cref{sec:dr-estimand-rc}.

\begin{remark}
	\leavevmode
	The estimand in \cref{eq:dr-estimand-rc} again looks familiar; it
	resembles a ratio of two $ATT_{dr,rc}(g,t;0)$ estimands from
	\cite{csa}, with the treatment
	replaced by the exposure variable $E_e$ and $D$ used as the outcome in
	the denominator.
	Also, \cref{eq:dr-estimand-rc} has a similar form as
	\cref{eq:dr-estimand-panel},
	a weighted residual including the control units mean function,
	but in the repeated cross-sections setting also including the terms
	$\kappa_{e, t}^{Y, rc},\kappa_{e, t}^{D, rc}$,
	consisting of differences between the treated and control mean functions in
	both periods.
\end{remark}

\begin{proposition}
	\label{prop:dr-estimand-weighted-rc}
	For the doubly robust estimand \cref{eq:dr-estimand-rc}, define
	\begin{align}
		\label{eq:resid-rc-case-individual}
		\varepsilon^{Y,trt,rc}_{e,t} & := Y - m_{e, t}^{trt, rc}(X),
		\quad
		\varepsilon^{Y,c,rc}_{e,s,t} := Y - m_{e, s, t}^{c,rc}(X),
		\\
		\nonumber
		\varepsilon^{D,trt,rc}_{e,t} & := D - g_{e, t}^{trt, rc}(X),
		\quad
		\varepsilon^{D,c,rc}_{e,s,t} := D - g_{e, s, t}^{c,rc}(X),
	\end{align}
	$\tau^{dr, rc}_{e, t}
		=: \tau^{dr, rc, num}_{e, t}/\tau^{dr, rc, den}_{e, t}$,
	and:
	\begin{align}
		\label{eq:mean-functions-rc-diffs}
		m_{e, \Delta}^{trt, rc}(X)
		:=
		m_{e, t}^{trt, rc}(X) - m_{e, e - 1}^{trt, rc}(X),
		\quad
		m_{e, t, \Delta}^{c, rc}(X)
		:= m_{e, t, t}^{c, rc}(X) - m_{e, t, e - 1}^{c, rc}(X).
	\end{align}
	Then the efficient influence function of \cref{eq:dr-estimand-rc} is
	\begin{align}
		\label{eq:if-dr-estimand-rc}
		 & \varphi^{dr,rc}(O; \tau^{dr, rc}_{e, t}; \eta^{dr, rc}_{e, t})
		\\
		\nonumber
		 & =
		\frac{1}{\tau^{dr, rc, den}_{e, t}}
		\Big[
			\varphi^{dr,rc}_{num}(O; \tau^{dr, rc, num}_{e, t}; \eta^{dr, rc, num}_{e, t})
			- \tau^{dr, rc}_{e, t}
			\varphi^{dr,rc}_{den}(O; \tau^{dr, rc, den}_{e, t}; \eta^{dr, rc, den}_{e, t})
			\Big],
	\end{align}
	where
	\begin{align}
		\label{eq:if-dr-num-estimand-rc}
		 & \varphi^{dr,rc}_{num}(O; \tau^{dr, rc, num}_{e, t}; \eta^{dr, rc, num}_{e, t})
		\\
		\nonumber
		 & :=
		w^{trt,rc}_{e, t}
		\left\{ \varepsilon^{Y,trt,rc}_{e,t}
		- E[w^{trt,rc}_{e, t} \varepsilon^{Y,trt,rc}_{e,t} ] \right\}
		+ w^{trt,p}_{e} \left\{ m_{e, \Delta}^{trt, rc}(X)
		- E[w^{trt,p}_{e} m_{e, \Delta}^{trt, rc}(X) ] \right\}
		\\ & \quad
		\nonumber
		-
		\left[w^{trt,rc}_{e, e-1}
			\left\{ \varepsilon^{Y,trt,rc}_{e,e-1}
			- E[w^{trt,rc}_{e, e-1} \varepsilon^{Y,trt,rc}_{e,e-1} ] \right\}
			+ w^{trt,p}_{e} \left\{ m_{e, t, \Delta}^{c, rc}(X)
			- E[w^{trt,p}_{e} m_{e, t, \Delta}^{c, rc}(X) ] \right\}\right]
		\\
		\nonumber
		 & \quad
		- [w^{c,rc}_{e, t, t}
		\left\{ \varepsilon^{Y,c,rc}_{e,t,t}
		- E[w^{c,rc}_{e, t, t} \varepsilon^{Y,c,rc}_{e,t,t} ] \right\}
		- w^{c,rc}_{e, t, e - 1}
		\left\{ \varepsilon^{Y,c,rc}_{e,t,e-1}
		- E[w^{c,rc}_{e, t, e - 1} \varepsilon^{Y,c,rc}_{e,t,e-1} ] \right\}]
		\\
		\label{eq:if-dr-den-estimand-rc}
		 &
		\varphi^{dr,rc}_{den}(O; \tau^{dr, rc, den}_{e, t}; \eta^{dr, rc, den}_{e, t})
		\\
		\nonumber
		 & :=
		w^{trt,rc}_{e, t}
		\left\{ \varepsilon^{D,trt,rc}_{e,t}
		- E[w^{trt,rc}_{e, t} \varepsilon^{D,trt,rc}_{e,t} ] \right\}
		+ w^{trt,p}_{e} \left\{ g_{e, \Delta}^{trt, rc}(X)
		- E[w^{trt,p}_{e} g_{e, \Delta}^{trt, rc}(X) ] \right\}
		\\ & \quad
		\nonumber
		-
		\left[w^{trt,rc}_{e, e-1}
			\left\{ \varepsilon^{D,trt,rc}_{e,e-1}
			- E[w^{trt,rc}_{e, e-1} \varepsilon^{D,trt,rc}_{e,e-1} ] \right\}
			+ w^{trt,p}_{e} \left\{ g_{e, t, \Delta}^{c, rc}(X)
			- E[w^{trt,p}_{e} g_{e, t, \Delta}^{c, rc}(X) ] \right\}\right]
		\\
		\nonumber
		 & \quad
		- [w^{c,rc}_{e, t, t}
		\left\{ \varepsilon^{D,c,rc}_{e,t,t}
		- E[w^{c,rc}_{e, t, t} \varepsilon^{D,c,rc}_{e,t,t} ] \right\}
		- w^{c,rc}_{e, t, e - 1}
		\left\{ \varepsilon^{D,c,rc}_{e,t,e-1}
		- E[w^{c,rc}_{e, t, e - 1} \varepsilon^{D,c,rc}_{e,t,e-1} ] \right\}],
	\end{align}
	and
	$\eta^{dr, rc}_{e, t} = (
		\rho,
		\lambda_{e - 1},
		\lambda_{t},
		p,
		m_{e, t}^{trt, rc},
		m_{e, e - 1}^{trt, rc},
		g_{e, t}^{trt, rc},
		g_{e, e - 1}^{trt, rc},
		m_{e, t, t}^{c, rc},
		m_{e, t, e-1}^{c, rc},
		g_{e, t, t}^{c, rc},
		g_{e, t, e-1}^{c, rc}
		)$
	is the nuisance-function tuple.
\end{proposition}

\medskip\noindent
Proof in \Cref{sec:eif-dr-estimand-rc}.

\subsubsection{Double Robustness}

\begin{proposition}
	\label{prop:dr-both}
	Fix $e \in \mathcal{E}$ and $t \ge e$, and let $C$ denote a generic
	control-group indicator, where either $C = C^{nev}$ or
	$C = C^{nye}_{e,t}$ and $p(X)$ the corresponding propensity.
	Let $\pi(X)$ be a working model for the propensity $p(X)$.
	In the panel data setting, let
	$\mu^{c,p}_{e,t}(X)$ and $\eta^{c,p}_{e,t}(X)$ be working models of
	$m^{c,p}_{e,t}(X)$ and $g^{c,p}_{e,t}(X)$, respectively.
	Similarly, in the repeated cross-sections setting, let $\mu^{c,rc}_{e,Y}(X)$
	and $\eta^{c,rc}_{e,Y}(X)$ be the working models
	of $m_{e, Y}^{c,rc}(X)$ and $g_{e, D}^{c,rc}(X)$, respectively.
	Assume the conditions of
	\Cref{prop:identification-panel,prop:identification-rc} hold.

	Then the following statements hold:
	\begin{enumerate}[label=(\alph*), nosep]
		\item In the panel-data setting,
		      \[
			      \tau^{dr,p}_{e,t} = LATT(e,t)
		      \]
		      if either (but not necessarily both) $\pi(X) = p(X)$ almost surely, or
		      \[
			      \mu^{c,p}_{e,t}(X) = m^{c,p}_{e,t}(X)
			      \quad \text{and} \quad
			      \eta^{c,p}_{e,t}(X) = g^{c,p}_{e,t}(X)
		      \]
		      almost surely.
		\item In the repeated cross-section setting,
		      \[
			      \tau^{dr,rc}_{e,t} = LATT(e,t)
		      \]
		      if either (but not necessarily both) $\pi(X) = p(X)$ almost surely, or
		      \[
			      \mu_{e,Y}^{c,rc}(X) = m_{e,Y}^{c,rc}(X)
			      \quad \text{and} \quad
			      \eta_{e,Y}^{c,rc}(X) = g_{e,D}^{c,rc}(X)
		      \]
		      almost surely.
	\end{enumerate}
\end{proposition}

\medskip\noindent
Proof in \Cref{sec:proof-other}.

\subsubsection{Bloom-type result}
In the special case of absorbing treatment (defined below) and
one-sided compliance (i.e., no units unexposed to the instrument are
treated), the $LATT(e,t)$ parameter can be related to the
instrument-exposure-cohort-specific $ATT(g,t)$ parameters of
\cite{csa}.
The following result shows this:
\begin{proposition}[Bloom-type result]
	\label{prop:bloom}
	Assume that panel data is available.
	Let $C$ denote a generic control-group indicator, where either $C = C^{nev}$
	or $C = C^{nye}_{e,t}$.
	Assume that treatment is an absorbing state, i.e. $D_{1} = 0$
	and
	\begin{align}
		\label{eq:treatment-absorbed}
		D_{t-1} = 1 \implies D_{t} = 1,
		\quad
		t \in \{2, 3, \ldots, \mathcal{T}\},
	\end{align}
	and one-sided compliance:
	\begin{align}
		\label{eq:assumption-bloom}
		C = 1 \implies D_{t} = 0.
	\end{align}
	Then:
	\begin{align}
		\label{eq:bloom}
		LATT(e, t)
		=
		\sum_{g \leq t}  ATT(g, t, e)P(G_{g} = 1 \mid D_{t} = 1, E_{e} = 1)
	\end{align}
	where
	\begin{equation}
		\label{eq:csa-objects}
		ATT(g, t, e) := E[Y_t(g)-Y_t(0)\mid G_g=1, E_{e} = 1],
	\end{equation}
	is an exposure-cohort-specific $ATT(g, t)$-estimand of \cite{csa} and
	$G := \min \{t \mid D_{t} = 1\}; \, G_{g} := \mathbf{1}\{G = g\}$
	the corresponding treatment cohort dummies.
	This result shows that the local average treatment effect on the treated in
	period $t$ for exposure cohort $e$, $LATT(e, t)$, under absorbing treatment
	and one-sided compliance, identifies an exposure-cohort-specific weighted
	average of $ATT(g, t)$'s for all those treated by time $t$.
\end{proposition}
\medskip\noindent
Proof in \Cref{sec:proof-other}.
\begin{remark}
	\leavevmode
	\begin{enumerate}[nosep]
		\item \Cref{prop:bloom}
		      is a Bloom-type result
		      \citep{bloom1984estimating,ap}
		      for the case of IDiD.
		      Moreover, \cref{eq:bloom} is a ``sensible causal estimand''
		      \citep{mogstad2024instrumental} in the sense that it is a convex
		      combination of
		      the causal $ATT(g, t, e)$ estimands with weights $P(G_{g} = 1 \mid
			      D_{t} = 1, E_{e} = 1)$, $g \leq t$.
		\item \Cref{prop:bloom} is useful in randomized experiments with
		      one-sided \textit{non}-compliance, where treatment received $D_t$ may differ
		      from treatment assignment $Z_t$.
		      In this case, the random assignment $Z_t$ can be used as an instrument
		      for the received treatment $D_t$, allowing researchers to approximate the
		      treatment effect in the presence of hidden confounding, extending the
		      cross-sectional Bloom result, as described in \cite{ap}, to the IDiD
		      framework.
		      In the two-period case ($\mathcal{T}=2$), there is only a single exposed and
		      treatment cohort. By \cref{eq:bloom}, $LATT(2,2)=ATT(2,2,2)$, which
		      coincides with the DiD estimand $ATT(2,2)$, implying that IDiD exactly
		      recovers the average treatment effect on the treated.
		\item The result is also related to the point that \textit{"IDiD is a natural
			      generalization of DiD, and the differences between the two designs
			      arise because IDiD allows the non-compliance of the assigned
			      treatment"} as stated in \cite{miyaji2025}.
		      Here, because of the mechanical one-sided compliance,
		      we can write the $LATT(e, t)$ as a convex average of
		      exposure-cohort-specific staggered DiD estimands, making the link
		      between IDiD with staggered exposure and DiD with staggered treatment
		      explicit.
		\item An analogous result obtains under the one-sided compliance condition,
		      $E_{e} = 1 \implies D_{t} = 1$. In this case, $LATT(e,t)$ admits a
		      representation as a convex average of treatment effects for the
		      non-treated (cf. \cite{ap}); this result is omitted for brevity.
		\item The result extends to repeated cross-sections with analogous
		      arguments; details are again omitted for brevity.
	\end{enumerate}
\end{remark}

\subsubsection{Group-level difference estimand}
Researchers often compare the target parameter across subgroups. This
can be handled explicitly by expressing subgroup parameters as functionals of the
full population distribution. Differences are then represented by
embedding the corresponding influence functions in a common
full-sample framework, weighted by the group indicator and its
probability.
To see this, let $B \in \{m, f\}$ denote a group indicator.
We consider the difference in subgroup-specific LATTs,
$LATT^{\Delta}(e,t) := LATT^{m}(e,t) - LATT^{f}(e,t)$,
for example based on the DR estimand in the panel data setting:
\begin{align}
	\label{eq:dr-estimand-group}
	\tau^{dr, p, \Delta}_{e, t}
	:= \tau^{dr, p, m}_{e, t} - \tau^{dr, p, f}_{e, t},
\end{align}
where each component admits an influence function under the corresponding
subgroup distribution,
$\varphi^{dr, p, m}(O; \tau^{dr, p, m}_{e, t}, \eta^{m})$
and
$\varphi^{dr, p, f}(O; \tau^{dr, p, f}_{e, t}, \eta^{f})$,
respectively.
These subgroup influence functions can be embedded into the full population as:
\begin{align}
	\label{eq:dr-if-group}
	 & \varphi^{dr, p}(O; \tau^{dr, p, \Delta}_{e, t}, \eta^{dr, p, \Delta}_{e, t})
	\\
	\nonumber
	 & =  \frac{\mathbf{1}\{B = m\}}{P(B = m)}
	\varphi^{dr, p, m}(O; \tau^{dr, p, m}_{e, t}, \eta^{dr, p, m}_{e, t})
	- \frac{\mathbf{1}\{B = f\}}{P(B = f)}
	\varphi^{dr, p, f}(O; \tau^{dr, p, f}_{e, t}, \eta^{dr, p, f}_{e, t})
\end{align}
where $\eta^{dr, p, \Delta}_{e, t} = (\eta^{dr, p, m}_{e, t}, \eta^{dr, p,
		f}_{e, t})$ collects the subgroup-specific nuisance functions.
Thus, inference on \cref{eq:dr-estimand-group} can be based on
\cref{eq:dr-if-group}.
The group-level difference estimand
in the case of repeated cross-sections is constructed in the exact same way and
is omitted for brevity.

\subsection{Aggregated effects}
\label{sec:weighting}
As shown in \cite{miyaji}, the $LATT(e,t)$ parameters can be aggregated into
estimands of interest, in a manner entirely analogous to the
aggregation of $ATT(g,t)$ parameters in \cite{csa}.
In this section, we apply the EIF machinery developed in
\Cref{sec:influence-functions} to derive the influence function for a
general weighted estimand $\theta$, defined below.

Consider an estimand that aggregates the $LATT(e,t)$ parameters:
\begin{align}
	\label{eq:thetaiv}
	\theta
	:=
	\sum_{e \in \mathcal{E}} \sum_{t = 2}^{\mathcal{T}}
	w(e,t)\, LATT(e,t),
\end{align}
where $w(e,t)$ denotes a (possibly data-dependent) weighting scheme.
To ease notation, define $\tau_{e, t} := LATT(e, t)$.
Let $\mathcal{I} :=
	\{(e, t) \in \mathcal{E} \times \{2, 3, \ldots, \mathcal{T}\} \mid
	A_{e, t} \}$
be a set of indices that picks out the correct cohorts and time indices
related to the summary parameter $\theta$.
By the product rule \cref{eq:if-product-rule},
the EIF of the weighted estimand \cref{eq:thetaiv} equals:
\begin{align}
	\label{eq:if-agg-estimand}
	\varphi^{\theta}(O; \{\tau_{i}\}_{i \in \mathcal{I}};
	\{\eta_{i}\}_{i \in \mathcal{I}})
	:=
	\sum_{(e, t) \in \mathcal{I}}
	\left[
	\mathbb{IF}(w(e, t))\tau_{e, t}+ w(e, t)\mathbb{IF}(\tau_{e, t})
	\right].
\end{align}

\paragraph{Weighted estimands and weights}
Define the cohort-specific average exposed effect on the treated in the first
stage, i.e., the share of the compliers in cohort $e$ in period $t$:
\begin{align}
	\label{eq:aet}
	AET(e, t)
	:= P(D_{t}(e) > D_{t}(\infty) \mid E = 1).
\end{align}
\Cref{tab:weighted-estimands} shows the weighted estimands and their
corresponding weights. For comparison, we also include the same aggregation
procedure as in \cite{csa}, but without the cohort-specific time-varying
complier weights. This allows researchers to choose whether to weight by
complier shares or not, although the complier-weighted version is the most
natural when the target parameter is the LATT.

The weights are themselves estimands, so the tools developed in
\Cref{sec:eif-derivation} can be applied to derive influence functions
for each weight in \Cref{tab:weighted-estimands}. With these in hand, we
apply \cref{eq:if-agg-estimand} to obtain the influence function of the
weighted estimand.

As is evident from \Cref{tab:weighted-estimands}, most of the weights
consist of sums of cohort probabilities conditional on different events, and the
time-varying complier share $AET(e, t)$.
For instance, for the weighted estimand $\theta^{o,IV}_{W}$, the main component
of the weight has influence function
(using the quotient rule \cref{eq:if-quotient-rule-estimand})
\begin{align}
	\label{eq:if-component-event}
	 & \mathbb{IF}(P(E = e \mid E \leq \mathcal{T}))
	\\
	\nonumber
	 & =
	\frac{1}{P(E \leq \mathcal{T})}
	\left[
		\mathbb{IF}\left(P(E = e)\right)
		- P(E = e \mid E \leq \mathcal{T})
		\mathbb{IF}\left(P(E \leq \mathcal{T})\right)
		\right]
\end{align}
with
$\mathbb{IF}\left(P(E \leq \mathcal{T})\right) = \sum_{e \leq \mathcal{T}}
	\mathbb{IF}\left(P(E = e)\right)$
and $\mathbb{IF}\left(P(E = e)\right) = \mathbf{1}\{E = e\} - P(E = e)$
cf. \cref{eq:if-block-prob}.
The influence function of the weight can now be derived using
\cref{eq:if-component-event} and the quotient-rule
\cref{eq:if-quotient-rule-estimand} summing the correct terms.

\medskip
\begin{threeparttable}

	\centering
	\small
	\begin{tabular}{l l}
	\hline
	Weighted Estimand & $w(e,t)$                                               \\
	\hline
	\multicolumn{2}{l}{\textbf{IDiD}}                                          \\

	$\theta^{IV}_{es}(l)$
	                  & $\mathbf{1}\{e+l \le \mathcal{T}\}\mathbf{1}\{t=e+l\}
		P(E=e \mid E+l \le \mathcal{T})
		\frac{AET(e,t)}{
			\sum_{e \in \mathcal{E}} P(E=e \mid E+l \le \mathcal{T})AET(e,t)
	}$                                                                         \\

	$\theta^{IV}_{bal,es}(l,l')$
	                  & $\mathbf{1}\{e+l' \le \mathcal{T}\}\mathbf{1}\{t=e+l\}
		P(E=e \mid E+l' \le \mathcal{T})
		\frac{AET(e,t)}{
			\sum_{e \in \mathcal{E}} P(E=e \mid E+l' \le \mathcal{T})AET(e,t)
	}$                                                                         \\

	$\theta^{IV}_{sel}(\tilde{e})$
	                  & $\mathbf{1}\{t \ge e\}\mathbf{1}\{e=\tilde{e}\}
	\frac{AET(e,t)}{\sum_{t=\tilde{e}}^{\mathcal{T}} AET(e,t)}$                \\

	$\theta^{IV}_{c}(\tilde{t})$
	                  & $\mathbf{1}\{t \ge e\}\mathbf{1}\{t=\tilde{t}\}
		P(E=e \mid E \le t)
		\frac{AET(e,t)}{
			\sum_{e \in \mathcal{E}} P(E=e \mid E \le t)AET(e,t)
	}$                                                                         \\

	$\theta^{cumm,IV}_{c}(\tilde{t})$
	                  & $\mathbf{1}\{t \ge e\}\mathbf{1}\{t \le \tilde{t}\}
		P(E=e \mid E \le t)
		\frac{AET(e,t)}{
			\sum_{e \in \mathcal{E}} P(E=e \mid E \le t)AET(e,t)
	}$                                                                         \\

	$\theta^{o,IV}_{W}$
	                  & $\mathbf{1}\{t \ge e\} P(E=e \mid E \le \mathcal{T})
		\Big/
		\sum_{e \in \mathcal{E}} \sum_{t=2}^{\mathcal{T}}
	\mathbf{1}\{t \ge e\} P(E=e \mid E \le \mathcal{T})$                       \\

	$\theta^{o,IV}_{sel}$
	                  & $\mathbf{1}\{t \ge e\} P(E=e \mid E \le \mathcal{T})
	\frac{AET(e,t)}{\sum_{t=e}^{\mathcal{T}} AET(e,t)}$                        \\

	\hline
	\multicolumn{2}{l}{\textbf{CSA}}                                           \\

	$\theta_{es}(l)$
	                  & $\mathbf{1}\{e+l \le \mathcal{T}\}\mathbf{1}\{t=e+l\}
	P(E = e \mid E+l \le \mathcal{T})$                                         \\

	$\theta_{bal,es}(l,l')$
	                  & $\mathbf{1}\{e+l' \le \mathcal{T}\}\mathbf{1}\{t=e+l\}
	P(E = e \mid E+l' \le \mathcal{T})$                                        \\

	$\theta_{sel}(\tilde{e})$
	                  & $\mathbf{1}\{t \ge e\}
	\mathbf{1}\{e=\tilde{e}\}/(\mathcal{T}-e+1)$                               \\

	$\theta_{c}(\tilde{t})$
	                  & $\mathbf{1}\{t \ge e\}\mathbf{1}\{t=\tilde{t}\}
	P(E = e \mid E \le t)$                                                     \\

	$\theta_{c}^{cumm}(\tilde{t})$
	                  & $\mathbf{1}\{t \ge e\}\mathbf{1}\{t \le \tilde{t}\}
	P(E = e \mid E \le t)$                                                     \\


	$\theta^{o}_{sel}$
	                  & $\mathbf{1}\{t \ge e\} P(E = e \mid E \le \mathcal{T})
	/(\mathcal{T}-e+1)$                                                        \\

	\hline
\end{tabular}

	\caption{Weighted estimands and their weights}

	\begin{tablenotes}
		\footnotesize
		\item The table shows different versions of \cref{eq:thetaiv}
		and the corresponding weights $w(e, t)$.
		The upper IDiD part of the table is a modification of the table in
		\cite{miyaji}.
		The lower part is a modification of the one in \cite{csa} tailored to the
		IDiD setting.
		Note that the
		$\theta^{IV}_{es}(l)$, $\theta^{bal,IV}_{es}(l,l')$,
		$\theta^{IV}_{c}(\tilde{t})$,
		have the group probability multiplied in the denominator compared to \cite{miyaji}
		in order for the weights to sum to $1$ (and likewise for
		$\theta^{cumm,IV}_{c}(\tilde{t})$ for it to sum to $\mathcal{T} - 1$).
		The $\theta^{o,IV}_{W}$ parameter equals $\theta^{o}_{W}$, so the
		latter is omitted. The indicators for cohort $e$ and time $t$ are
		also implicitly included in the denominator notation.
	\end{tablenotes}
	\label{tab:weighted-estimands}
\end{threeparttable}

\paragraph{Event study at horizon $l$ parameter, $\theta_{es}^{IV}(l)$}
Here, we focus on the event study at horizon $l$ target parameter, which
aggregates the $LATT(e, t)$ parameter at time $t = e + l$, i.e., $l$ periods
after exposure to the instrument.
The target parameter is:
\begin{align}
	\label{eq:esl-estimand}
	\theta_{es}^{IV}(l)
	 & :=
	\sum_{e \in \mathcal{E}} \sum_{t = 2}^{\mathcal{T}}
	w^{IV}_{es(l)}(e, t) \tau_{e, t}
\end{align}
where the weight $w^{IV}_{es(l)}(e, t)$ is defined in the upper part of
\Cref{tab:weighted-estimands}.
\paragraph{Influence function of $\theta_{es}^{IV}(l)$}
To derive the influence function of the $\theta_{es}^{IV}(l)$ estimand,
we apply the product rule \cref{eq:if-product-rule} and derive the influence
function of each component separately.
First, we find the influence function of the weight $w^{IV}_{es(l)}(e, t)$.
We apply the product rule \cref{eq:if-product-rule} to get the
EIF of a generic term in the numerator and denominator, and then use the
quotient rule to get the EIF of the weight.
Towards this end, define
\begin{align*}
	\varphi^{w^{IV}_{es(l)}, num}_{e, t}
	 & :=
	\mathbb{IF}(
	P(E = e \mid E + l \leq \mathcal{T}) AET(e, t)
	)
	\\
	 & =
	\mathbb{IF}(P(E = e \mid E + l \leq \mathcal{T})) AET(e, t)
	+ P(E = e \mid E + l \leq \mathcal{T})\mathbb{IF}(AET(e, t)).
\end{align*}
For the conditional probability we can use the block \cref{eq:if-block-cprob},
and the EIF for $AET(e, t)$ is the IF of the denominator for our (estimable) target
parameter at hand.
Hence, applying the quotient rule gives:
\begin{align}
	\label{eq:if-weight-esl}
	 & \varphi^{w^{IV}_{es(l)}}_{e, t}(O; w^{IV}_{es(l)}(e, t), \eta)
	\\
	\nonumber
	 & :=
	\frac{\mathbf{1}\{e + l \leq \mathcal{T}, t = e + l\}}{
		\sum_{e \in \mathcal{E}}
		P(E = e \mid E + l \leq \mathcal{T})AET(e, t)
	}
	\Big[
		\varphi^{w^{IV}_{es(l)}, num}_{e, t}
		- w^{IV}_{es(l)}(e, t)
		\sum_{e \in \mathcal{E}}
		\varphi^{w^{IV}_{es(l)}, num}_{e, t}
		\Big].
\end{align}
Let $\mathcal{I}_{l} :=
	\{(e, t) \in \mathcal{E} \times \{2, 3, \ldots, \mathcal{T}\} \mid
	e + l \leq \mathcal{T} \wedge t = e + l\}$.
Thus, using \cref{eq:if-weight-esl},
and writing $\varphi_{e, t}(O; \tau_{e, t}, \eta)$
for the influence function of the generic $\tau_{e, t}$,
we get using \cref{eq:if-agg-estimand}:
\begin{align}
	\label{eq:if-esl}
	 & \varphi^{\theta_{es}(l)}(O; \{\tau_{i}\}_{i \in \mathcal{I}_{l}};
	\{\eta_{i}\}_{i \in \mathcal{I}})
	\\
	 & :=
	\nonumber
	\sum_{(e, t) \in \mathcal{I}_{l}}
	\left[
	\varphi^{w^{IV}_{es(l)}}_{e, t}(O; w^{IV}_{es(l)}(e, t), \eta)\tau_{e, t}
	+ w^{IV}_{es(l)}(e, t)\varphi_{e, t}(O; \tau_{e, t}, \eta)
	\right].
\end{align}
\Cref{eq:if-esl} is the influence function we use to conduct inference with on
the horizon $l$ parameters, $\theta^{IV}_{es(l)}$.

Note that we chose $\theta^{IV}_{es(l)}$ over $\theta^{IV}_{bal,es}(l,l')$ for
simplicity.
The point on compositional changes making the interpretation of
$\theta^{IV}_{es(l)}$ harder, as pointed out in \cite{csa}, still applies,
and it is up to the econometrician to make the trade-off of interpretability vs
efficiency (more observations).

\paragraph{Multiplier Bootstrap and Simultaneous Confidence Bands}
Researchers often interpret multiple $\theta^{IV}_{es(l)}$ jointly.
In this case, simultaneous confidence bands are appropriate.
\cite{csa} show how to construct such bands via the multiplier bootstrap
and establish validity for the vector of $ATT(g,t)$ in their Theorem~3.
An analogous result applies to the vector of $LATT(e,t)$, allowing us to
use their Algorithm~1 for inference on our parameters of interest.
In particular, their Corollary~1 implies that the resulting bands
achieve correct asymptotic coverage, so simultaneous inference on
$\{\theta^{IV}_{es}(l) : l \in \{0,1,\ldots,h\}\}$ is valid.
The multiplier bootstrap is applied to the influence function
\cref{eq:if-esl} for each $l$.
We verify the procedure in \Cref{sec:simexp2}.

\section{Estimation and inference}
\label{sec:estimation-and-inference}

We consider plug-in estimators of $\tau^{dr,p}_{e,t}$ and
$\tau^{dr,rc}_{e,t}$ in
\cref{eq:dr-estimand-panel,eq:dr-estimand-rc}. We take two approaches.
First, we estimate the nuisance functions using simple parametric models
(e.g.\ linear or logistic regression) and invoke Donsker conditions to control
the empirical process term. Second, we allow for flexible machine learning
methods, avoid empirical process assumptions, and instead use DML.
We show below that the resulting DML estimators coincide with the cross-fitted estimator of \cite{kennedy}.

\paragraph{Asymptotically linear representation}
The plug-in estimators of $\tau^{dr,p}_{e,t}$ and $\tau^{dr,rc}_{e,t}$
admit the decomposition in \Cref{eq:decomposition}.
\cite{kennedy} shows that the general plug-in estimator has first-order bias
equal to $- \int \varphi(\cdot; \hat{\tau}, \hat{\eta}) \, dP$. This bias can
be removed by imposing the moment condition
$P_n[\varphi(\cdot; \hat{\tau}, \hat{\eta})] = 0$ and solving for the estimator,
yielding an estimating equation estimator\footnote{
	We show in \Cref{sec:dml} that the cross-fitted estimating equation estimator
	coincides with the DML estimator under the \texttt{DML2} algorithm;
	see \Cref{remark:comparing-dml-cross-fit}.
}\footnote{
	Alternative approaches include one-step estimators and targeted maximum
	likelihood estimation (TMLE); see \cite{kennedy}.
}.

As discussed in \Cref{sec:identification-panel}, the estimators based on the DR
estimands in this paper are of this type and therefore eliminate the
first-order bias by construction\footnote{
	Although the plug-in estimator of \cref{eq:dr-estimand-panel} solves the
	estimating equation based on \cref{eq:if-latt-panel},
	$P_n \varphi^{p}(O; \hat{\tau}^{dr,p}_{e,t}, \hat{\eta}^{p}_{e,t}) = 0$,
	rather than the IF
	$\varphi^{dr,p}(\cdot; \tau^{dr,p}_{e,t}, \eta^{dr,p}_{e,t})$ in
	\cref{eq:if-dr-estimand-panel}, the plug-in bias is still zero since
	$P_n \varphi^{dr,p}(O; \hat{\tau}^{dr,p}_{e,t}, \hat{\eta}^{dr,p}_{e,t}) = 0$
	as well. The same holds for repeated cross-sections. This follows by direct
	calculation; see \Cref{eq:plug-in-extra}.
}.
This justifies the decomposition in \cref{eq:decomposition} for our plug-in
estimators. To obtain the asymptotic linear representation
\begin{align}
	\label{eq:decomposition-alin}
	\sqrt{n}(\hat{\tau} - \tau)
	=
	\sqrt{n}(P_n - P)\varphi(\cdot; \tau, \eta)
	+ o_P(1),
\end{align}
it remains to control the empirical process and remainder terms. Specifically,
to go from \cref{eq:decomposition} to \cref{eq:decomposition-alin}
we require
\begin{align}
	\label{eq:conditions-alin}
	\sqrt{n}(P_n - P)
	\big[
		\varphi(\cdot; \tau, \hat{\eta})
		-
		\varphi(\cdot; \tau, \eta)
		\big]
	= o_P(1),
	\quad
	\sqrt{n} P \varphi(\cdot; \tau, \hat{\eta}) = o_P(1).
\end{align}
The first condition is ensured by cross-fitting or Donsker assumptions, while
the second requires case-specific arguments \citep{kennedy}, formalized in
\Cref{prop:remainder-term-panel,prop:remainder-term-rc}.

\subsection{Remainder terms}
The following proposition provides conditions under which the remainder
term is asymptotically negligible, i.e., satisfies the right-hand side of
\cref{eq:conditions-alin}, in the panel data setting.

\begin{proposition}[Remainder term: panel data]
	\label{prop:remainder-term-panel}
	Suppose the identification in \Cref{prop:identification-panel} holds,
	and that the second-order rate condition holds:
	\begin{align*}
		\Big\Vert
		\frac{p}{1 - p}
		-
		\frac{\hat{p}}{1 - \hat{p}}
		\Big\Vert_{2, P}
		\Big\Vert
		m_{e, t}^{c, p} - \hat{m}_{e, t}^{c, p}
		\Big\Vert_{2, P}
		 & =
		o_{P}(n^{-1/2}),
		\\
		\Big\Vert
		\frac{p}{1 - p}
		-
		\frac{\hat{p}}{1 - \hat{p}}
		\Big\Vert_{2, P}
		\Big\Vert
		g_{e, t}^{c, p} - \hat{g}_{e, t}^{c, p}
		\Big\Vert_{2, P}
		 & =
		o_{P}(n^{-1/2}).
	\end{align*}
	Then the remainder term for the influence function
	\cref{eq:if-dr-estimand-panel} satisfies
	\begin{align}
		\label{eq:remainder-term-panel}
		\sqrt{n}\,
		P \varphi^{dr, p}(\cdot; \tau^{dr, p}_{e, t}, \hat{\eta}^{dr, p}_{e, t})
		=
		o_{P}(1).
	\end{align}
\end{proposition}
\medskip\noindent
Proof in \Cref{sec:remainder-panel}.

\medskip\noindent
The following proposition provides conditions under which the remainder
term is asymptotically negligible, i.e., satisfies the right-hand side of
\cref{eq:conditions-alin}, in the repeated cross-sections setting.
\begin{proposition}[Remainder term: repeated cross-sections]
	\label{prop:remainder-term-rc}
	Suppose the identification in \Cref{prop:identification-rc} holds.
	Moreover, assume:
	\begin{enumerate}
		\item (Treated nuisance consistency)
		      \begin{align}
			      \label{eq:treated-nuisance-rc-conditions}
			      \Vert m^{trt, rc}_{e,t} - \hat{m}^{trt, rc}_{e,t} \Vert_{2,P}
			      = o_P(1),
			      \quad
			      \Vert g^{trt, rc}_{e,t} - \hat{g}^{trt, rc}_{e,t} \Vert_{2,P}
			      = o_P(1),
		      \end{align}
		      and similarly for $t = e-1$.

		\item (Weight fluctuations)
		      \begin{align*}
			      \left|
			      \frac{\lambda_t}{\hat{\rho}^{trt,rc}_{e,t}} - \frac{1}{\hat{\rho}}
			      \right|
			      = O_P(n^{-1/2}),
			      \quad
			      \left|
			      \frac{\lambda_{e-1}}{\hat{\rho}^{trt,rc}_{e,e-1}} - \frac{1}{\hat{\rho}}
			      \right|
			      = O_P(n^{-1/2}),
		      \end{align*}
		      where
		      $\hat{\rho}^{trt, rc}_{e, t} :=  P_{n}[E_{e} \cdot \mathbf{1}\{T = t\}]$
		      and $\hat{\rho} := P_{n}[E_{e}]$
		      are estimators of the denominators
		      of the weights in
		      \cref{eq:weights-trt-rc-norm,eq:weights1-panel-norm},
		      respectively.
		\item (Second-order rate condition)
		      \begin{align*}
			      \Big\Vert
			      \frac{p}{1 - p}
			      -
			      \frac{\hat{p}}{1 - \hat{p}}
			      \Big\Vert_{2, P}
			      \Big\Vert
			      m^{c, rc}_{e,t} - \hat{m}^{c, rc}_{e,t}
			      \Big\Vert_{2, P}
			       & = o_P(n^{-1/2})
			      \\
			      \Big\Vert
			      \frac{p}{1 - p}
			      -
			      \frac{\hat{p}}{1 - \hat{p}}
			      \Big\Vert_{2, P}
			      \Big\Vert
			      g^{c, rc}_{e,t} - \hat{g}^{c, rc}_{e,t}
			      \Big\Vert_{2, P}
			       & = o_P(n^{-1/2}),
		      \end{align*}
		      and similarly for $t = e-1$.
	\end{enumerate}
	Then
	\begin{align}
		\label{eq:remainder-term-rc}
		\sqrt{n}\,
		P \varphi^{dr, rc}(\cdot; \tau^{dr, rc}_{e,t}, \hat{\eta}^{dr, rc}_{e,t})
		= o_P(1).
	\end{align}
\end{proposition}
\medskip\noindent
Proof in \Cref{sec:remainder-rc}.

\begin{remark}[Rate double robustness of the remainder term]
	Compared to \Cref{prop:remainder-term-panel}, the repeated
	cross-sections remainder does not exhibit the same rate double robustness,
	in that the outcome and propensity models cannot offset each other’s
	estimation error, as in the panel setting.
	While double robustness for consistency still holds, only the
	control-group nuisance components enter through the usual
	propensity-outcome product structure; see also \cite[p.~26]{kennedy}.
\end{remark}

\subsection{Plug-in estimators}
In the following, we derive asymptotic linear representations for
plug-in estimators of $\tau^{dr,p}_{e,t}$ and
$\tau^{dr,rc}_{e,t}$ under Donsker class conditions, building on the
remainder results in
\Cref{prop:remainder-term-panel,prop:remainder-term-rc}.
Such conditions are often restrictive in high-dimensional settings.
Accordingly, \Cref{sec:dml} introduces a DML estimator that avoids these
assumptions by controlling the empirical process term via cross-fitting.

\paragraph{Panel Data}
\begin{proposition}
	\label{prop:dr-estimator-alin-panel}
	Consider the plug-in estimator of \cref{eq:dr-estimand-panel},
	$\hat{\tau}^{dr,p}_{e, t}$.
	Suppose \Cref{prop:remainder-term-panel} holds and that Donsker conditions
	ensure the empirical process term is $o_{P}(1)$ as in
	\cref{eq:conditions-alin}. Then
	\begin{align}
		\label{eq:dr-estimator-alin-panel}
		\sqrt{n}(\hat{\tau}^{dr,p}_{e, t} - \tau^{dr,p}_{e, t})
		=
		\sqrt{n} (P_{n} - P)
		\varphi^{dr, p}(\cdot; \tau^{dr, p}_{e, t}, \eta^{dr, p}_{e, t})
		+ o_{P}(1),
	\end{align}
	and consequently,
	\begin{align*}
		\sqrt{n}(\hat{\tau}^{dr,p}_{e, t} - \tau^{dr,p}_{e, t})
		\overset{d}{\to}
		N\left(
		0,
		E[
			\{\varphi^{dr, p}(O; \tau^{dr, p}_{e, t}, \eta^{dr, p}_{e, t})\}^{2}
		]
		\right).
	\end{align*}
\end{proposition}

\paragraph{Repeated Cross-Sections}
\begin{proposition}
	\label{prop:dr-estimator-alin-rc}
	Consider the plug-in estimator of \cref{eq:dr-estimand-rc},
	$\hat{\tau}^{dr,rc}_{e, t}$.
	Suppose \Cref{prop:remainder-term-rc} holds and that Donsker conditions
	ensure the empirical process term is $o_{P}(1)$ as in
	\cref{eq:conditions-alin}.
	Then
	\begin{align}
		\label{eq:dr-estimator-alin-rc}
		\sqrt{n}(\hat{\tau}^{dr,rc}_{e, t} - \tau^{dr,rc}_{e, t})
		=
		\sqrt{n} (P_{n} - P)
		\varphi^{dr, rc}(\cdot; \tau^{dr, rc}_{e, t}, \eta^{dr, rc}_{e, t})
		+ o_{P}(1),
	\end{align}
	and consequently,
	\begin{align*}
		\sqrt{n}(\hat{\tau}^{dr, rc}_{e, t} - \tau^{dr, rc}_{e, t})
		\overset{d}{\to}
		N(
		0,
		E[
		\{\varphi^{dr, rc}(O; \tau^{dr, rc}_{e, t}, \eta^{dr, rc}_{e, t})\}^{2}
		]
		).
	\end{align*}
\end{proposition}

\subsection{Double Machine Learning / Cross-Fitting}
\label{sec:dml}
In this section, we present DML estimators for both data settings. In
each setting, the Neyman-orthogonal score is based on the corresponding
EIF in \cref{eq:if-latt-panel,eq:if-latt-rc}.
The resulting estimators coincide with the cross-fitted estimating
equation estimator of \cite{kennedy}. This follows because DML applies
cross-fitting to Neyman-orthogonal scores; when these scores are derived
from the EIFs in \cref{sec:nonparametric-identification}, the procedure
is equivalent to cross-fitting the corresponding estimating equation estimator.

\paragraph{Regularization and overfitting bias}
The Donsker assumptions made in the previous section are often unrealistic in
settings where the nuisances are high-dimensional and complex.
Here, flexible machine learning estimators are more appropriate.
Estimating nuisance functions with machine learning induces two forms of bias,
often referred to as regularization and overfitting bias in the DML literature
\citep{chernozhukovDoubleDebiasedMachine2018a}.
The remedy is to leverage a Neyman-orthogonal score for the DML estimator and
apply cross-fitting \citep{chernozhukovDoubleDebiasedMachine2018a}. This is
what we will do below.

\paragraph{DML Algorithm}
We construct our DML estimators using the \texttt{DML2} procedure of
\citep{chernozhukovDoubleDebiasedMachine2018a,bach2022doubleml}. The key components are:
\begin{enumerate}
	\item a Neyman-orthogonal score identifying the target parameter
	\item machine learning estimators of the nuisance functions
	\item $K$-fold cross-fitting of the nuisance estimators
	\item plugging the cross-fitted nuisances into the score and solving the
	      resulting moment condition yielding the DML estimator
\end{enumerate}
Hence, we first construct the Neyman-orthogonal scores.
For this we can use the influence functions
\cref{eq:if-latt-panel,eq:if-latt-rc}.

\paragraph{Panel Data}
Using the EIF for the panel data setting in
\cref{eq:if-latt-panel}, we obtain:
\begin{align}
	\nonumber
	 & \varphi^{p}(O; \tau^{p}_{e, t}, \eta^{p}_{e, t})
	\\
	\nonumber
	 & \propto
	\{ w^{trt,p}_{e} - w^{c,p}_{e, t} \}
	\{\Delta_{t-e+1}Y_{t} - m_{e, t}^{c, p}(X)\}
	- \tau^{p}_{e, t}
	\{ w^{trt,p}_{e} - w^{c,p}_{e, t} \}
	\{\Delta_{t-e+1}D_{t} - g_{e, t}^{c, p}(X)\}
	\\
	\nonumber
	 & =:
	\psi^{p}_{b}(O; \eta^{p}_{e, t})
	- \tau^{p}_{e, t}
	\psi^{p}_{a}(O; \eta^{p}_{e, t})
	\\
	\label{eq:dml-score-panel}
	 & =: \psi^{p}(O; \tau^{p}_{e, t}, \eta^{p}_{e, t})
\end{align}
where we use DML notation for the score components, i.e. $\psi$, $\psi_{a}$ and
$\psi_{b}$.
Hence, we see that the influence function is proportional to a score (i.e.
function that identifies the target parameter) that is linear in the estimand
$\tau^{p}_{e, t}$.

Let $\{O_i\}_{i=1}^n$ be partitioned into $K$ folds of equal size
$n_k = n/K$, and denote the empirical measure on fold $k$ by $P_n^k$\footnote{
For equal-sized folds:
$K^{-1}P_{n}^{k}[f]
	= K^{-1}n_{k}^{-1}\sum_{i=1}^{n} f(O_{i})
	= n^{-1}\sum_{i=1}^{n}f(O_{i})$.
}.
Let $\hat{\eta}^{p,-k}_{e,t}$ be the nuisance estimator trained on all
observations outside fold $k$ (the superscript $-k$ indicates this).
Hats denote estimated quantities.

The DML estimator solves:
\begin{align*}
	 & K^{-1} \sum_{k=1}^{K}
	P_{n}^{k}
		[\psi^{p}(O; \hat{\tau}^{dml, p}_{e, t}, \hat{\eta}^{p, -k}_{e, t})]
	= 0
	\\
	\iff
	\hat{\tau}^{dml, p}_{e, t}
	 & =
	\frac{
		K^{-1} \sum_{k=1}^{K} P_{n}^{k}
			[\psi^{p}_{b}(O; \hat{\eta}^{p, -k}_{e, t})]
	}{
		K^{-1} \sum_{k=1}^{K} P_{n}^{k}
			[\psi^{p}_{a}(O_; \hat{\eta}^{p, -k}_{e, t})]
	}
	\\
	 & =
	\frac{
	K^{-1}\sum_{k=1}^{K} P_{n}^{k}
	[
	\{ \hat{w}^{trt, p, -k}_{e} - \hat{w}^{c, p, -k}_{e, t} \}
	\{\Delta_{t-e+1}Y_{t} - \hat{m}_{e, t}^{c, p, -k}(X)\}
	]
	}{
	K^{-1}\sum_{k=1}^{K} P_{n}^{k}
	[\{ \hat{w}^{trt, p, -k}_{e} - \hat{w}^{c, p, -k}_{e, t} \}
	\{\Delta_{t-e+1}D_{t} - \hat{g}_{e, t}^{c, p, -k}(X)\}]
	}.
\end{align*}
Hence we have the result:

\begin{proposition}[DML Panel Data]
	\label{prop:dml-panel}
	Suppose the identification in
	\Cref{prop:identification-panel} holds,
	and also the condition on the remainder in
	\cref{prop:remainder-term-panel}\footnote{
		In \Cref{remark:identical-remainders} I show the remainder terms
		using either \cref{eq:if-latt-panel}
		or \cref{eq:if-dr-estimand-panel} are the same
		and likewise for the repeated cross-sections setting
		\cref{eq:if-latt-rc}
		or \cref{eq:if-dr-estimand-rc}.
		Hence, \Cref{prop:remainder-term-panel} can be used even though
		the DML score was constructed based on \cref{eq:if-latt-panel}
		and not \cref{eq:if-dr-estimand-panel}.
		The same statements hold in the repeated cross-sections case for
		\Cref{prop:dml-rc}.
	}.
	Then:
	\begin{align*}
		\sqrt{n}(\hat{\tau}^{dml, p}_{e, t}  - \tau^{p}_{e, t, 0})
		\overset{d}{\to} N(0, \sigma^{2, p}_{e, t}),
	\end{align*}
	where
	\begin{align*}
		\sigma^{2, p}_{e, t} :=
		(J^{p}_{0})^{-2}
		E[\{\psi^{p}(O; \tau^{p}_{e, t, 0}, \eta^{p}_{e, t, 0})\}^{2}],
		\quad
		J^{p}_{0} = E[\psi^{p}_{a}(O; \eta^{p}_{e, t, 0})]
	\end{align*}
	and we use the notation $\tau^{p}_{e, t, 0}, \eta^{p}_{e, t, 0}$
	for the true value of the target parameter and nuisance function, respectively.
	A variance estimator is:
	\begin{align*}
		\hat{\sigma}^{2, p}_{e, t}
		=
		(\hat{J}^{p}_{0})^{-2}
		K^{-1} \sum_{k=1}^{K}
		P_{n}^{k}
		[\{\psi^{p}(O; \hat{\tau}^{dml, p}_{e, t}, \hat{\eta}^{p, -k}_{e, t})\}^{2}],
		\quad
		\hat{J}^{p}_{0}
		=
		K^{-1} \sum_{k=1}^{K}
		P_{n}^{k}
			[\psi^{p}_{a}(O; \hat{\eta}^{p, -k}_{e, t})]
	\end{align*}
	and an approximate $(1 - \alpha)\%$ confidence interval is constructed as:
	\begin{align*}
		[\hat{\tau}^{dml, p}_{e, t}
			\pm \Phi^{-1}(1 - \alpha/2) \hat{\sigma}^{p}_{e, t}/\sqrt{n}],
	\end{align*}
	where $\Phi^{-1}(\cdot)$ is the quantile function of the standard normal
	distribution.
\end{proposition}

\paragraph{Repeated Cross-Sections}
Using the EIF for the repeated cross-sections setting in
\cref{eq:if-latt-rc}, we obtain:
\begin{align}
	\nonumber
	\varphi^{rc}(O; \tau^{rc}_{e, t}, \eta^{rc}_{e, t})
	          & \propto
	w^{trt,rc}_{e} \{Y - m_{e, Y}^{trt, rc}(X)\}
	-
	w^{c,rc}_{e}
	\{Y - m_{e, Y}^{c,rc}(X)\}
	\\
	\nonumber
	\nonumber & \quad
	+ \frac{E_{e}}{\rho}
	\{
	m_{e,t}^{trt, rc}(X) - m_{e,e-1}^{trt, rc}(X)
	-
	[
	m_{e, t, t}^{c, rc}(X) - m_{e, t, e-1}^{c, rc}(X)
	]
	\}
	\\
	\nonumber & \quad
	\nonumber
	- \tau^{rc}_{e, t}
	(
	w^{trt,rc}_{e} \{D - g_{e, Y}^{trt, rc}(X)\}
	-
	w^{c,rc}_{e}
	\{D - g_{e, D}^{c,rc}(X)\}
	\\
	\nonumber
	\nonumber & \quad
	+ \frac{E_{e}}{\rho}
	\{
	g_{e,t}^{trt, rc}(X) - g_{e,e-1}^{trt, rc}(X)
	-
	[ g_{e, t, t}^{c, rc}(X) - g_{e, t, e-1}^{c, rc}(X) ]
	\}
	)
	\\
	\label{eq:dml-score-rc}
	          & =:
	\psi^{rc}_{b}(O; \eta^{rc}_{e, t})
	- \tau^{rc}_{e, t} \psi^{rc}_{a}(O; \eta^{rc}_{e, t})
	\\
	\nonumber
	          & =: \psi^{rc}(O; \tau^{rc}_{e, t}, \eta^{rc}_{e, t}).
\end{align}
Again, the influence function is proportional to a score that is
linear in the estimand $\tau^{rc}_{e, t}$.

Similarly to the panel data case, the DML estimator solves:
\begin{align*}
	 & K^{-1} \sum_{k=1}^{K}
	P_{n}^{k}
		[\psi^{rc}(O; \hat{\tau}^{dml, rc}_{e, t}, \hat{\eta}^{rc, -k}_{e, t})]
	= 0
	\\
	\iff
	\hat{\tau}^{dml, rc}_{e, t}
	 & =
	\frac{
		K^{-1} \sum_{k=1}^{K} P_{n}^{k}
			[\psi^{rc}_{b}(O; \hat{\eta}^{rc, -k}_{e, t})]
	}{
		K^{-1} \sum_{k=1}^{K} P_{n}^{k}
			[\psi^{rc}_{a}(O; \hat{\eta}^{rc, -k}_{e, t})]
	},
\end{align*}
where the score components are defined in \cref{eq:dml-score-rc}.
Hence we have the result:

\begin{proposition}[DML Repeated Cross-Sections]
	\label{prop:dml-rc}
	Suppose the identification in
	\Cref{prop:identification-rc} holds,
	and also the condition on the remainder in \cref{prop:remainder-term-rc}.
	Then:
	\begin{align*}
		\sqrt{n}(\hat{\tau}^{dml, rc}_{e, t}  - \tau^{rc}_{e, t, 0})
		\overset{d}{\to} N(0, \sigma^{2, rc}_{e, t}),
	\end{align*}
	where
	\begin{align*}
		\sigma^{2, rc}_{e, t} :=
		(J^{rc}_{0})^{-2}
		E[\{\psi^{rc}(O; \tau^{rc}_{e, t, 0}, \eta^{rc}_{e, t, 0})\}^{2}],
		\quad
		J^{rc}_{0} = E[\psi^{rc}_{a}(O; \eta^{rc}_{e, t, 0})]
	\end{align*}
	A variance estimator is:
	\begin{align*}
		\hat{\sigma}^{2, rc}_{e, t}
		=
		(\hat{J}^{rc}_{0})^{-2}
		K^{-1} \sum_{k=1}^{K}
		P_{n}^{k}
		[\{\psi^{rc}(O; \hat{\tau}^{dml, rc}_{e, t}, \hat{\eta}^{rc, -k}_{e, t})\}^{2}],
		\quad
		\hat{J}^{rc}_{0}
		=
		K^{-1} \sum_{k=1}^{K}
		P_{n}^{k}
			[\psi^{rc}_{a}(O; \hat{\eta}^{rc, -k}_{e, t})]
	\end{align*}
	and an approximate $(1 - \alpha)\%$ confidence interval is constructed as:
	\begin{align*}
		[\hat{\tau}^{dml, rc}_{e, t}
			\pm \Phi^{-1}(1 - \alpha/2) \hat{\sigma}^{rc}_{e, t}/\sqrt{n}].
	\end{align*}
\end{proposition}

\begin{remark}[Comparing DML to cross-fitting for IDiD]
	\label{remark:comparing-dml-cross-fit}
	\leavevmode
	DML with the Neyman-orthogonal scores
	\cref{eq:dml-score-panel,eq:dml-score-rc}
	is equivalent to a cross-fitted plug-in estimator based on
	the double robust estimands \cref{eq:dr-estimand-panel,eq:dr-estimand-rc},
	and conducting inference using the cross-fitted efficient influence
	functions \cref{eq:if-latt-panel,eq:if-latt-rc}.

	To see this, note that the cross-fitted estimator is obtained by solving
	an equation of the form in \cref{eq:estimating-equation}. When
	\Cref{eq:if-latt-panel,eq:if-latt-rc} are used, the resulting
	estimating-equation estimator coincides with the plug-in estimators of
	\Cref{eq:dr-estimand-panel,eq:dr-estimand-rc}. For instance, the
	panel denominator is estimated by
	\begin{align*}
		P_{n}[
		\{ \hat{w}^{trt,p}_{e} - \hat{w}^{c,p}_{e, t} \}
		\{\Delta_{t-e+1}D_{t} - \hat{g}_{e, t}^{c, p}(X)\}
		]
		= P_{n}[\psi^{p}_{a}(O; \hat{\eta}^{p}_{e, t})]
	\end{align*}
	and in the repeated cross-sections case by
	\begin{align*}
		P_{n}
		[
		\{\hat{w}^{trt,rc}_{e} - \hat{w}^{c,rc}_{e}\}
		\{D - \hat{g}_{e, Y}^{c,rc}(X)\}
		] + \hat{\kappa}_{e, t}^{D, rc}
		= P_{n}[\psi^{rc}_{a}(O; \hat{\eta}^{rc}_{e, t})],
	\end{align*}
	and similarly for the numerators.
	Hence, the cross-fitted estimators are equivalent to the DML estimators.

	Because of this equivalence, the remainder terms and nuisance-rate
	conditions for the cross-fitted estimators are the same as those
	needed for the DML estimators. The theorem for the cross-fitted
	estimator in \cite{kennedy} therefore also applies to the DML estimators.
	The DML estimators are also doubly robust, since they implicitly
	estimate the doubly robust estimand, and cross-fitting handles
	the empirical process term in the general decomposition automatically.

	As a final point, the DML estimators are built from the EIF with normalized
	weights. Hence
	the DML estimators with $K=1$ correspond to the estimating-equation
	estimators based on the normalized influence functions, which match
	the DR plug-in estimators of \cref{eq:dr-estimand-panel,eq:dr-estimand-rc}.
	However, the variance estimators of the DML estimators do not take into
	account the normalization adjustments as done in the influence functions
	\cref{eq:if-dr-estimand-panel,eq:if-dr-estimand-rc}
	when conducting inference.
	This does not seem to matter in the simulation experiments in
	\Cref{sec:simexp1}.
\end{remark}

\paragraph{Estimators of aggregated effects}
For completeness, we state below a proposition for estimators of the weighted
estimands in \Cref{sec:weighting}.
\begin{proposition}
	\label{prop:weighted-estimand-normality}
	Suppose the assumptions of one of
	\Cref{prop:dr-estimator-alin-panel,prop:dr-estimator-alin-rc,prop:dml-panel,prop:dml-rc}
	hold for each cohort-time estimator entering the aggregated estimator
	$\hat{\theta}$. Then
	\begin{align*}
		\sqrt{n}(\hat{\theta} - \theta)
		\overset{d}{\to}
		N(0, E[\{\varphi^{\theta}(O; \{\tau_{i}\}_{i \in \mathcal{I}};
		\{\eta_{i}\}_{i \in \mathcal{I}})\}^{2}])
	\end{align*}
	where the estimand and influence function are defined in
	\cref{eq:thetaiv,eq:if-agg-estimand}, respectively.
\end{proposition}
\Cref{prop:weighted-estimand-normality} is analogous to Corollary 1 in
\cite{miyaji} and Corollary 2 of \cite{csa}, and the proof is omitted.

\section{Simulation Experiments}
\label{sec:simulation}
In this section, we present three simulation experiments that assess the
finite-sample performance and double robustness of the estimators in
\Cref{sec:estimation-and-inference}.

The first experiment verifies double robustness in a simple two-period
DGP using the never-exposed control group. We compare the main estimators,
namely the DR estimators in
\Cref{prop:dr-estimator-alin-panel,prop:dr-estimator-alin-rc}
and the DML estimators in \Cref{prop:dml-panel,prop:dml-rc}, to three
non-doubly robust alternatives: inverse probability weighting (IPW),
standardized IPW (IPWS), and outcome regression (REG), each constructed by
exploiting the ratio-of-ATT-parameters structure in
\cref{eq:identification-panel,eq:identification-rc}. The standardization in
IPWS corresponds to the normalization of control weights as done in
\Cref{sec:nonparametric-identification}. This simulation setup follows the one in
\cite{sazhao}.

The second experiment considers staggered exposure to the instrument for
both control groups in \cref{eq:control-variables}, in both data
settings, and in the presence of an unobserved confounder affecting both
treatment and outcome. We focus on the $LATT(e,t)$ estimates varying
across cohorts and time, and aggregate these into estimates of the
weighted estimand in \cref{eq:esl-estimand}. Finally, we introduce a
group indicator under which $LATT$ differs across the two groups. We estimate
group-specific effects via \cref{eq:dr-estimand-group} and aggregate
these to horizon-specific effects as in \cref{eq:esl-estimand}, capturing
average dynamic differences across groups for each horizon $l$.

The third experiment follows the same setup as in the second experiment
(without the group indicator) but now imposes absorbing treatment
\cref{eq:treatment-absorbed}.
This setting allows us to verify the Bloom-type result in
\Cref{prop:bloom} and compare our estimators to those of \citet{csa}.
As before, a hidden confounder jointly determines treatment and outcome,
so the DiD estimator is biased, whereas the IDiD estimator is not.

\subsection{Simulation Experiment 1: Double Robustness}
\label{sec:simexp1}
We compare the DR and DML estimators to the IPW, IPWS and REG estimators
through a simulation experiment similar to \cite{sazhao}.
To assess the double robustness properties, we simulate four DGPs with varying
degrees of misspecification.
The first DGP has both the outcome regression and propensity score correctly
specified.
The second DGP has the outcome regression correctly specified and the propensity
score misspecified.
The third DGP has the outcome regression misspecified and the propensity
score correctly specified.
The fourth and last DGP has both the outcome regression and propensity score
misspecified.

We focus on a two-period setup with covariates using a never-exposed control
group with panel data.
Let
\begin{align*}
	X = (X_{1}, X_{2})' \sim N(0, I_{2}),
	\quad
	\mathrm{supp}(E) = \{\infty,2\},
\end{align*}
and
\begin{align*}
	h_{\ell}(X) = 0.8 X_{1} - 0.5 X_{2},
	\quad
	h_{n}(X)    = 0.8\sin(X_{1}) - 0.5 \cdot \mathbf{1}\{X_{2}>0\} + 0.3 X_{1}X_{2},
\end{align*}
and define
\begin{align*}
	g_{ps}(X) \in \{h_{\ell}(X), h_{n}(X)\},
	\qquad
	s(X) \in \{h_{\ell}(X), h_{n}(X)\}.
\end{align*}
Here, $h_{\ell}$ corresponds to correct specification and $h_{n}$ to
misspecification.
Exposure occurs in period $2$ with probability
\begin{align*}
	P(E=2 \mid X) = \mathrm{expit}(g_{ps}(X)).
\end{align*}
Treatment propensities are
\begin{align*}
	p_{t}(\infty) = \mathrm{expit}(0.2 + s(X) - 0.4t),
	\quad
	p_{t}(2)
	= \mathrm{expit}(0.2 + s(X) - 0.4t + \kappa),
\end{align*}
and treatment states are generated as
\begin{align*}
	D_{t}(\infty) = \mathbf{1}\{U \le p_{t}(\infty)\},
	\quad
	D_{t}(2) = \mathbf{1}\{U \le p_{t}(2)\},
	\quad
	U \sim \mathrm{Uniform}(0,1).
\end{align*}
Untreated potential outcomes follow
\begin{align*}
	Y_{t}(0) & = [1 + 0.8 \cdot (t - 1)] s(X) + \varepsilon_{t},
\end{align*}
while treated outcomes are
\begin{align*}
	Y_{t}(1) = Y_{t}(0) + \tau,
\end{align*}
where $\varepsilon_{1}, \varepsilon_{2} \sim N(0,0.2^{2})$.
The observed treatment is given by
\cref{eq:observed-treatment-status} and the observed outcome by
\cref{eq:observed-outcome}.
Note that the treatment and outcome evolutions depend on the same specification
$s(X)$.
Also, the true local average treatment effect on the treated for those exposed
in period $2$ equals $LATT(2, 2) = \tau$.

In the experiment, we set $\kappa, \tau = 1$,
$n = 5000$, $\mathcal{T} = 2$, and repeat the experiment $B=4999$ times.
For the repeated cross-sections case, we set $\lambda = 0.5$.
For the DML estimators, we use $K=5$ folds for cross-fitting. Outcome
regressions are estimated by linear regression, and propensity scores,
when required, by logistic regression.
The results are shown in \Cref{tab:sim-exp-1-panel,tab:sim-exp-1-rc} for the panel
and repeated cross-sections settings, respectively.

The panel results confirm the expected double robustness and efficiency
properties of the proposed estimators. When both nuisance components are
correctly specified (DGP1), all estimators are essentially
unbiased with coverage close to the nominal level.
The IPW estimators exhibit substantially higher variance than the other
estimators. When only the
outcome regression is correctly specified (DGP2), the DR, DML and regression estimators remain
approximately unbiased, whereas the IPW estimators become severely
biased and coverage deteriorates.
Conversely, when only the propensity score is correctly specified (DGP3), the
DR and DML estimators again remain unbiased, while the regression estimator is biased,
illustrating the double robustness property. When both nuisance components are
misspecified (DGP4), all estimators are biased, as expected.

The repeated cross-section results display the same qualitative double
robustness pattern but with larger dispersion. In DGP1-DGP3,
the DR and DML estimators remain approximately unbiased with reasonable coverage,
again verifying the double robustness property,
although their variance is higher than in the panel case. In contrast,
the IPW estimators, even in DGP1 under correct specification, exhibit extreme
variability, reflected in very large RMSE and asymptotic variance, driven by
instability in the weighting scheme.
In particular, the non-standardized IPW estimator explodes in variance, as was
also the case for the two-period DiD experiment in \cite{sazhao},
although here further amplified by the ratio structure of the estimators.
When the propensity score is misspecified (DGP2), the DR and DML estimators remain
relatively robust, while the regression estimator is biased. Finally, when both
nuisance components are misspecified (DGP4), all estimators are biased.

Overall, the results highlight the double robustness of the proposed
estimators. While this property holds in both data settings, the
repeated cross-section estimators are less stable in finite samples due
to noisier estimation of the underlying components; larger sample sizes
would mitigate this. The DR and DML estimators perform similarly in this
experiment, although the DML estimator exhibits slightly higher
variance, likely due to cross-fitting. In more complex settings, DML,
combined with more flexible nuisance models, is expected to outperform
the simpler DR estimator.

\begin{table}[H]
	\centering
	\caption{Simulation experiment 1: Panel data}
	\label{tab:sim-exp-1-panel}

	\begin{threeparttable}

		\includegraphics[width=\textwidth]{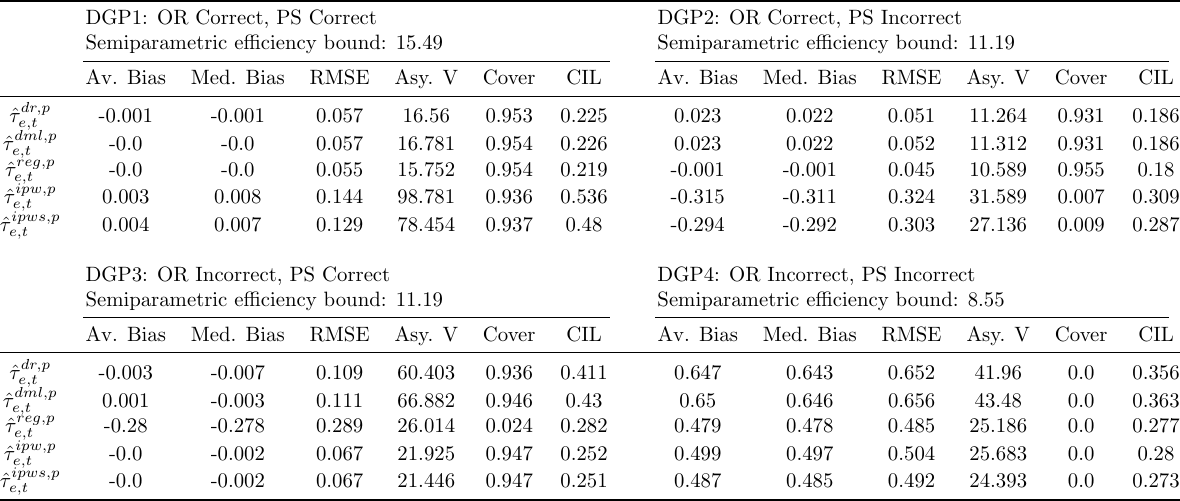}

		\begin{tablenotes}
			\footnotesize
			\item Notes: Simulations are based on $B=4999$ Monte Carlo repetitions
			with sample size $n=5000$. The target parameter is $LATT(2, 2)$;
			all subscripts $(e, t) = (2, 2)$.

			$\hat{\tau}^{reg,s}_{e,t}$ denotes the outcome regression (OR)
			estimator, $\hat{\tau}^{ipw,s}_{e,t}$ the inverse probability
			weighting (IPW) estimator, $\hat{\tau}^{ipws,s}_{e,t}$ the
			standardized IPW estimator, $\hat{\tau}^{dr,s}_{e,t}$ the doubly
			robust (DR) estimator, and $\hat{\tau}^{dml,s}_{e,t}$ the DML estimator,
			where $s \in \{p, rc\}$ indicates the panel and repeated cross-section
			settings, respectively.

			We use a linear working model for the outcome regression and a logistic
			working model for the propensity score, with parameters estimated by OLS
			and maximum likelihood. “Av. Bias”, “Med. Bias”, “RMSE”, “Asy. V”,
			“Cover”, and “CIL” denote the average bias, median bias, root mean
			squared error, average estimated asymptotic variance, 95\% coverage
			probability, and average confidence interval length, respectively.
			The semiparametric efficiency bounds are calculated using the true values
			and the results from \Cref{prop:if-latt-panel,prop:if-latt-rc},
			respectively.
			The table follows the structure of \cite{sazhao}.
		\end{tablenotes}

	\end{threeparttable}
\end{table}

\begin{table}[H]
	\centering
	\caption{Simulation experiment 1: Repeated cross-sections}
	\label{tab:sim-exp-1-rc}

	\begin{threeparttable}

		\includegraphics[width=\textwidth]{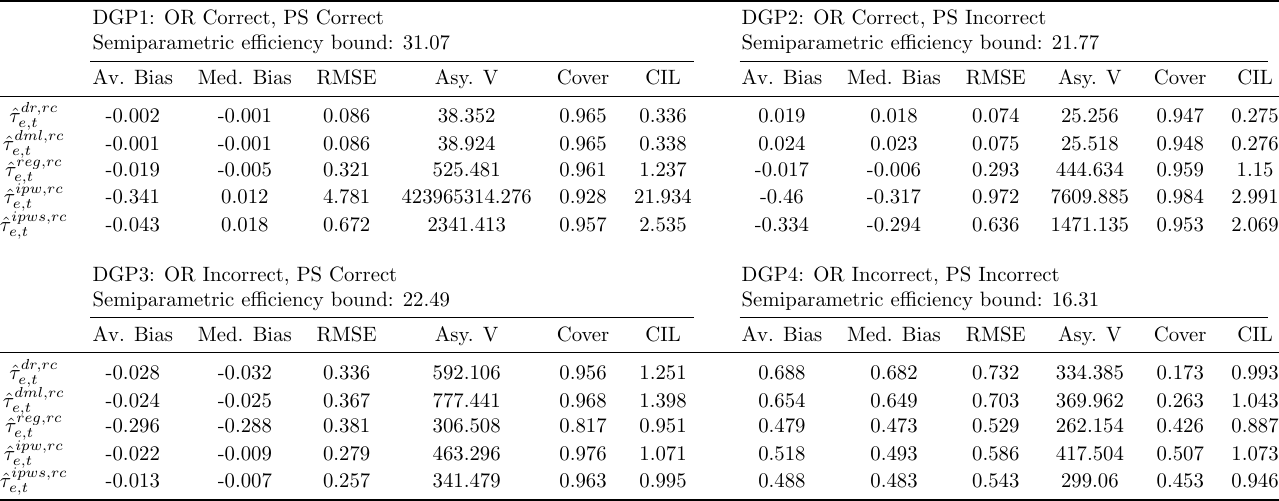}

		\begin{tablenotes}
			\footnotesize
			\item Notes: See \Cref{tab:sim-exp-1-panel} for notes.
		\end{tablenotes}

	\end{threeparttable}
\end{table}

\subsection{Simulation Experiment 2: Staggered exposure and control groups}
\label{sec:simexp2}
In this simulation experiment we simulate a panel with staggered exposure and
introduce the time-varying hidden confounder $H_{t}$
determining both the treatment $D_{t}$ and outcome $Y_{t}$.
The setup for the simulation is as follows.
We simulate a time-invariant covariate and time-varying hidden confounder
as $X, H_{t} \sim N(0, 1),$
and exposure cohorts from
\begin{align}
	\label{eq:propensity-and-support}
	P(E = e \mid X) = \frac{\exp(\beta_{e} X,)}{
		\sum_{e \in \mathcal{E}} \exp(\beta_{e} X,)
	},
	\quad
	\mathrm{supp}(E) = \{\infty, 2, 3, \ldots, n_{E}\},
\end{align}
with $\beta_{e} \in \{0.1, \ldots, 0.3\}$, $e \in \mathrm{supp}(E)$,
equally spaced between $0.1$ and $0.3$\footnote{
	For $P(E = e \mid X)$ simulated this way,
	the generalized propensity \cref{eq:generalized-propensity-scores}
	in the case of the never-exposed control group
	equals
	$p := P(E = e \mid X, E_{e} + C^{nev} = 1)
		= \mathrm{expit}([\beta_{e} - \beta_{\infty}]X),
	$
	where $\beta_{\infty}$ is the parameter corresponding to the never-exposed control
	group;
	hence, $p / (1 - p) = \exp([\beta_{e} - \beta_{\infty}]X)
		\sim \mathrm{LogNormal}(0, [\beta_{e} - \beta_{\infty}]^{2}).
	$
	Thus, in order to stabilize the propensity in the simulations,
	we constrain the parameters $\beta_{e}$ to $\{0.1, \ldots, 0.3\}$ so the
	maximum difference is $0.2$.
}.
The potential treatments and outcomes are simulated as:
\begin{align}
	\label{eq:meat-treatment}
	L_{t}         & = \nu + 0.5 X + H_{t} + \frac{1}{8} t/\mathcal{T},
	\quad
	\nu  \sim \mathrm{Uniform}(-1, -0.2),
	\\
	\nonumber
	p_{t}(\infty) & = \mathrm{expit}(L_{t}),
	\quad
	p_{t}(e) = \mathrm{expit}(L_{t} + \delta),
	\\
	\nonumber
	D_{t}(\infty) & = \mathbf{1}\{U_{t} \leq p_{t}(\infty)\},
	\quad
	D_{t}(e) = \mathbf{1}\{U_{t} \leq p_{t}(e)\},
	\quad
	U_{t}  \sim \mathrm{Uniform}(0, 1),
	\\
	\label{eq:pos-simulation}
	              & Y_{t}(0) = \eta + X + H_{t} + t/\mathcal{T} + \varepsilon_{t},
	\quad
	Y_{t}(1) = Y_{t}(0) + \tau_{t} + \nu_{t},
\end{align}
where $\eta, \varepsilon_{t}, \nu_{t} \sim N(0, 1)$.
Note that $H_{t}$ determines both the treatment and outcome cf.
\cref{eq:meat-treatment,eq:pos-simulation}.
Again, the observed treatment is given by
\cref{eq:observed-treatment-status} and the observed outcome by
\cref{eq:observed-outcome}.

For simplicity, we focus on the DR estimators
\cref{eq:dr-estimator-alin-panel,eq:dr-estimator-alin-rc} and evaluate
their finite-sample performance under panel and repeated cross-section
sampling, respectively, using never-exposed and not-yet-exposed control
groups. We set $n = 10{,}000$, $\mathcal{T} = 5$, $\tau_t = 1$ for all $t$, $\delta = 1$,
$n_{E} = 5$
and perform $B = 1499$ simulation draws.
In the not-yet-exposed experiment the $E = \infty$ is removed from
\cref{eq:propensity-and-support}.
The average evolution of the
treatment $D_t$ and outcome $Y_t$, conditional on the exposure cohort
$E$, for a simulated dataset is shown in
\Cref{fig:evolution}.

\medskip
\noindent\textit{Cohort-time effects.}
\Cref{fig:simexp2-a} reports the distribution of $\hat{\tau}^{dr, p}_{e,t}$ and
$\hat{\tau}^{dr, rc}_{e,t}$ across all $(e, t)$ pairs. The estimates are
centered around the true value in all designs, showing the robustness of the
IDiD procedure to hidden confounding. The panel
case is most tightly concentrated, while the repeated cross-sections case exhibits
higher dispersion; both show occasional large realizations due to the ratio
structure.
In both designs, the not-yet-exposed estimates, compared to the never-exposed,
become increasingly dispersed at longer horizons, reflecting the corresponding
reduction in the size of the control group.

\begin{figure}[!htbp]
	\centering

	\begin{subfigure}{\textwidth}
		\centering
		\includegraphics[width=0.70\textwidth]{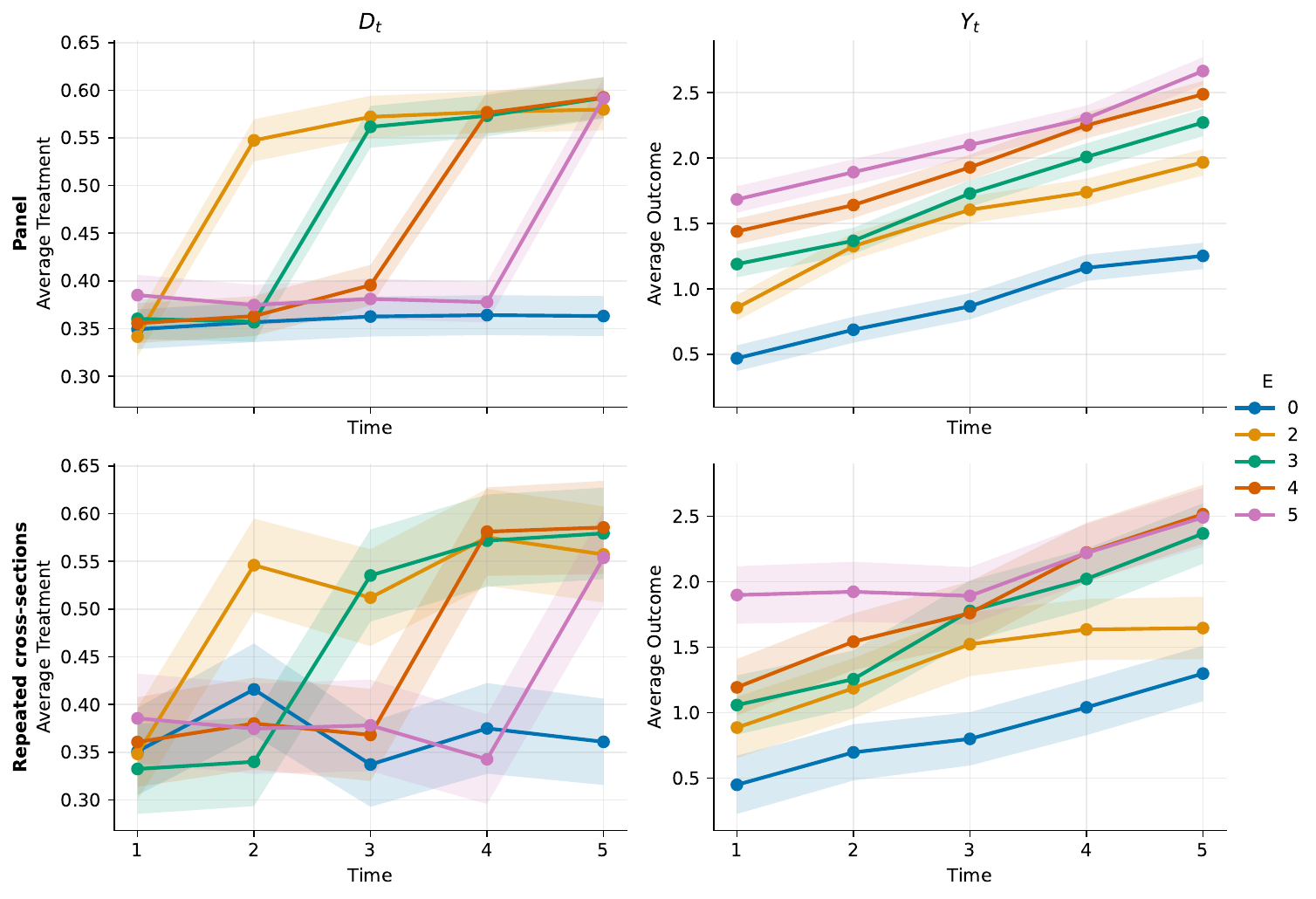}
		\caption{Average evolution of treatment $D_{t}$ and outcome $Y_{t}$
			conditional on exposure cohort $E$ from the DGP of
			Simulation Experiment 2. The upper row shows panel data and
			the lower repeated cross-sections.}
		\label{fig:evolution}
	\end{subfigure}

	\vspace{0.5em}

	\begin{subfigure}{\textwidth}
		\centering
		\includegraphics[width=0.7\textwidth]{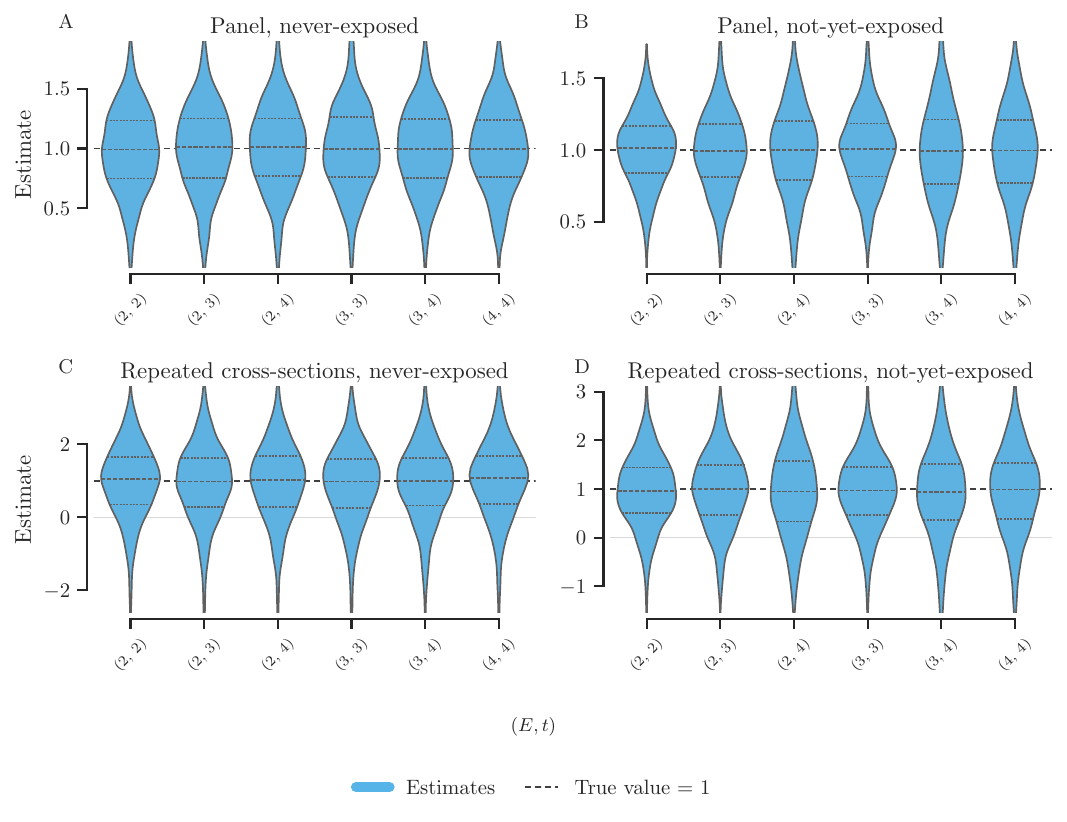}
		\caption{Simulation experiment 2: distributions of the cohort-time
			estimators $\hat{\tau}^{dr,p}_{e,t}$ and $\hat{\tau}^{dr,rc}_{e,t}$
			across $(e,t)$ pairs, for panel data and repeated cross-sections under
			never-exposed and not-yet-exposed controls.
			The estimates for $(e, t) = (5, 5)$ using never-exposed controls are not
			shown.
		}
		\label{fig:simexp2-a}
	\end{subfigure}

	\caption{Simulation experiment 2: design and estimator distributions.}
	\label{fig:simexp2}
\end{figure}

\medskip
\noindent\textit{Event study aggregation.}
\Cref{fig:simexp2-b} shows the aggregated effects $\hat{\theta}^{IV}_{es(l)}$. The
same pattern emerges: estimates remain centered, but variability increases in
repeated cross-sections. Again, for the not-yet-exposed controls, dispersion
increases with the horizon due to the shrinking control group.

We conduct a multiplier bootstrap and construct simultaneous confidence
bands for $\{\theta^{IV}_{es}(l) \mid l = 0,1,2\}$ in
\cref{eq:esl-estimand}.
Let $\hat{C}^{IV}_{es}(l)$ denote the interval for the
horizon-$l$ effect. Then, as in \cite{csa}, the intervals are simultaneous in
the sense that\footnote{
	The general simultaneous confidence band for
	$LATT(e, t)$, $t \geq e$,
	is simultaneous in the sense that
	$
		P (LATT(e, t) \in  \hat{C}(e, t) \,
		\forall (e,t) \in \mathcal{E} \times \{2, 3, \ldots,
		\mathcal{T}\}
		: t \geq e) \to  1 - \alpha
	$; see \cite{csa} for details.
}
\begin{equation}
	\label{eq:cov}
	P (\theta^{IV}_{es}(l) \in  \hat{C}^{IV}_{es}(l) \,
	\forall l \in  \{0,1,2\} )
	\to  1 - \alpha.
\end{equation}
We verify this in \Cref{tab:sim-exp-2-tests} with $\alpha = 0.05$.
The block ``Pooled'' reports coverage rates \cref{eq:cov} computed from
simultaneous confidence bands (``Simultaneous'') obtained via the
multiplier bootstrap and from pointwise confidence intervals
(``Pointwise'').

As expected, the simultaneous coverage is close to $0.95$, whereas the
pointwise coverage is too low. The repeated cross-section estimates are
noisier, which propagates to the simultaneous bands, yielding slightly
conservative (too high) coverage.

\begin{figure}[H]
	\centering
	\includegraphics[width=0.7\textwidth]{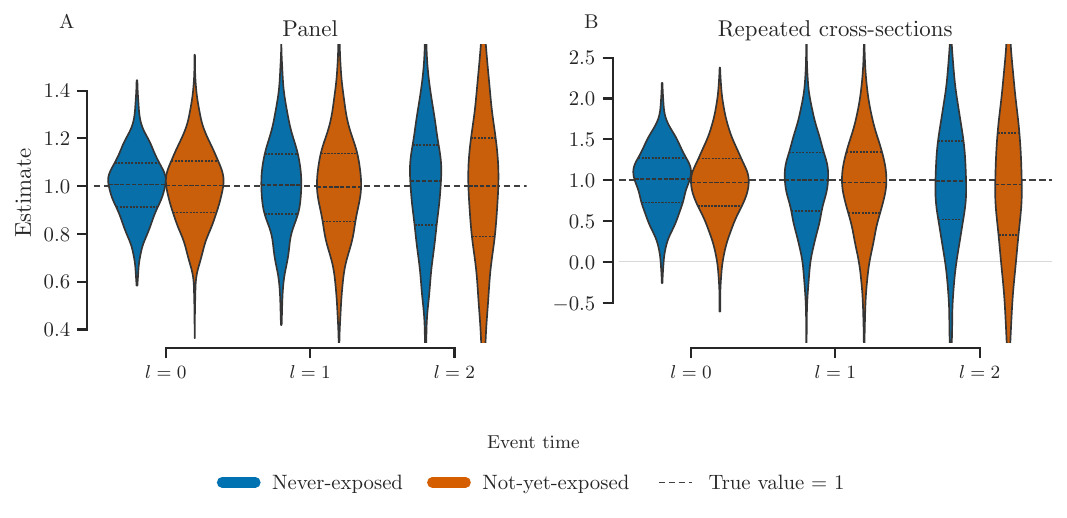}
	\caption{Simulation experiment 2: aggregated effects using never-exposed
		and not-yet-exposed controls, for panel data and repeated cross-sections.
		The largest realizations are omitted for readability.
	}
	\label{fig:simexp2-b}
\end{figure}

\begin{table}[H]
	\centering

	\caption{Simulation experiment 2: aggregated effects coverage}
	\label{tab:sim-exp-2-tests}

	\begin{threeparttable}

		\includegraphics[width=0.8\textwidth]{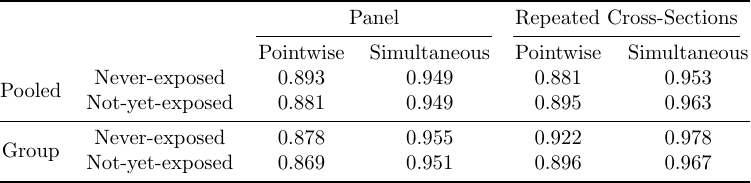}

		\begin{tablenotes}
			\footnotesize
			\item Notes: The first two rows report joint coverage for the
			horizon-aggregated effects in \Cref{fig:simexp2-b}. The last two rows
			report joint coverage for the differenced group-specific
			horizon-aggregated effects
			in \Cref{fig:simexp2-c}.
		\end{tablenotes}

	\end{threeparttable}
\end{table}

\medskip
\noindent
\textit{Group-specific aggregation.}
Finally, we simulate a new dataset with a binary group indicator $F \sim
	\mathrm{Bern}(0.5)$ splitting the units into two groups.
The groups have different $LATT$s with $F = 1$ being affected less.
We set $n=20,000$ so there are two groups of approximately $10,000$.
Our aim is to estimate the group-wise estimand $\tau^{dr,p,\Delta}_{e,t}$ as
defined in \cref{eq:dr-estimand-group}
and aggregate the effects into dynamic horizon-$l$ effects as in
\cref{eq:esl-estimand}.
The treated state in \cref{eq:pos-simulation} now becomes
\begin{align*}
	Y_{t}(1) = Y_{t}(0) + \tau_{t}(1 - 0.5 F) + \nu_{t},
\end{align*}
where those with $F = 1$ have $LATT(e, t) = 0.5 \tau_{t}$.

\Cref{fig:simexp2-c} reports the aggregated group-specific horizon
effects and their difference, i.e., the aggregated estimates of
$\tau^{dr,p,\Delta}_{e,t}$.
The estimator tracks the targets
well in the panel setting, whereas repeated cross-sections exhibit
higher variability, particularly for the difference estimand.
Again, the not-yet-exposed case has higher variability because of the
shrinking control group.

We again verify the validity of the multiplier bootstrap to construct
simultaneous confidence bands for the aggregated effects, here for the
aggregated horizon-$l$ effects of $\tau^{dr,p,\Delta}_{e,t}$.
The results are reported in block ``Group'' of \Cref{tab:sim-exp-2-tests}.
For panel data, the
simultaneous coverage is close to $0.95$ for both control groups. In repeated
cross-sections, coverage is again slightly conservative (above $0.95$),
reflecting the higher variability of this design, compounded by the additional
noise from differencing. Pointwise coverage is too low across all simulations,
as expected.

\begin{figure}[H]
	\centering
	\includegraphics[width=0.7\textwidth]{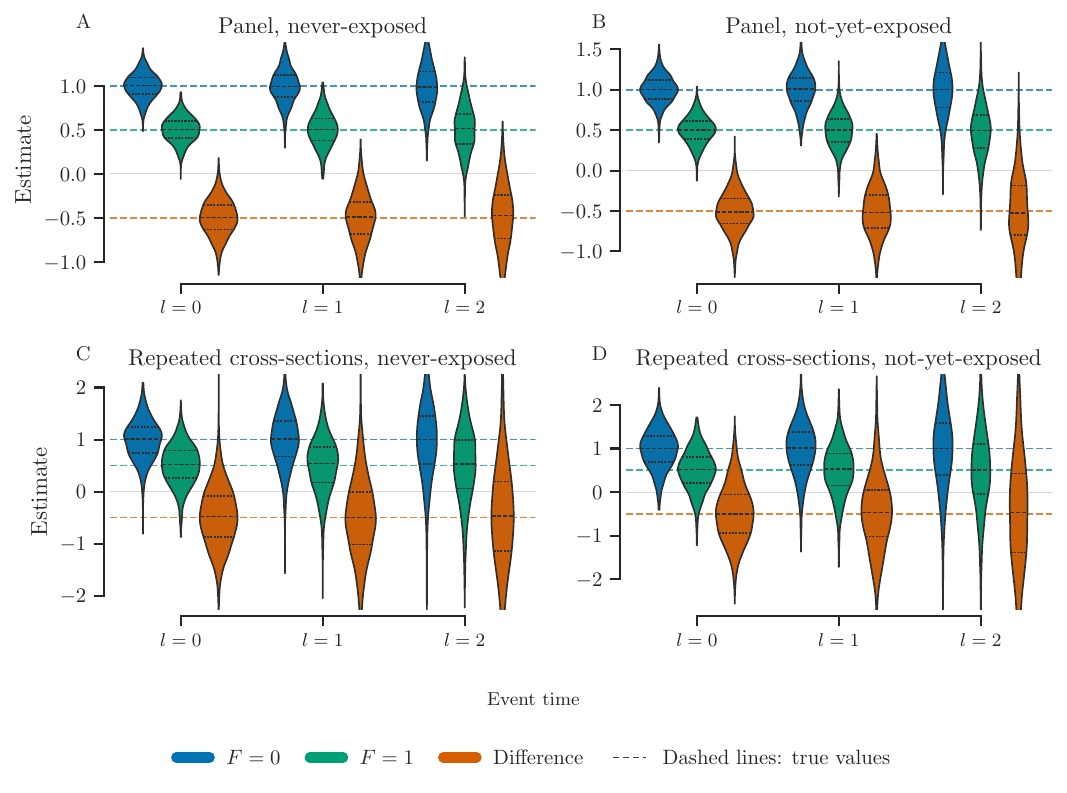}
	\caption{
		Simulation experiment 2: group-specific aggregated effects and their difference
		under never-exposed and not-yet-exposed controls.
		The largest realizations are omitted for readability.
	}
	\label{fig:simexp2-c}
\end{figure}

\subsection{Simulation Experiment 3: One-sided compliance and staggered exposure}
We simulate a setting with staggered exposure to the instrument and
one-sided compliance in the presence of an unobserved confounder
affecting both treatment and outcomes to verify \Cref{prop:bloom}. The
setup follows Simulation Experiment 2, except that treatment is an
absorbing state, i.e. \cref{eq:treatment-absorbed}, and the treatment effects
are now heterogeneous across treatment cohorts and time.
The potential outcomes of \cref{eq:pos-simulation}
in the treated state are now given by
\begin{align}
	\label{eq:att-effect-sim2}
	Y_{t}(g) = Y_{t}(0) + \tau_{g, t} + \nu_{t},
	\quad
	\tau_{g,t}
	=
	\frac{t-g+1}{2},
\end{align}
where the dependence on $g$ reflects the absorbing treatment;
the $ATT(g,t,e)$ parameters increases linearly over post-treatment periods as
$1/2,\,1,\,3/2,\ldots$ and are invariant across exposure cohorts $e \in
	\mathcal{E}$.
The potential outcomes in the untreated states are as in
\cref{eq:pos-simulation}.

In the simulation experiment, we set $n = 10,000$, $\mathcal{T}=5$,
$\mathrm{supp}(E)=\{\infty,2,3,4,5\}$, use as control group the
never-exposed units, $C^{nev}$,
and repeat the experiment $B=1499$ times.
The implied $ATT(g, t)$ parameters are
$
	\{\tau_{2,2}, \tau_{2,3}, \tau_{2,4}, \tau_{2,5},
	\tau_{3,3}, \tau_{3,4}, \tau_{3,5},
	\tau_{4,4}, \tau_{4,5},
	\tau_{5,5}\},
$
with true values given by \cref{eq:att-effect-sim2}.

For each draw, we estimate $LATT(e,t)$ using our DR
estimators in each data setting. We also construct the treatment cohort
variable $G = \min\{t \mid D_t = 1\}$, enabling two quantities:
(i) the true $ATT(g,t,e)$, computed from the simulated (unobserved)
potential outcomes, and (ii) $ATT(g,t)$ estimated via the
procedure of \cite{csa}.

By \Cref{prop:bloom}, for each $(e,t) \in \mathcal{E} \times \{2, 3, \ldots,
	\mathcal{T}\}$ with $t \ge e$,
$\hat{LATT}(e,t)$ should approximately equal $\sum_{g \le t}
	\hat{ATT}(g,t,e)\,\hat{P}(G = g \mid D_t = 1, E_e = 1),$ cf. \cref{eq:bloom}.
We compute this quantity using Oracle estimates for $\hat{ATT}(g,t,e)$ (the CSA
estimators won't work because of the hidden confounder) and empirical
conditional probabilities for the weights.

Results are reported in \Cref{tab:sim-exp-2}.
For horizon $l=t-e=0$, the $LATT(e, t)$ estimates coincide exactly with the
exposure-cohort-specific $ATT(g, t)$, so the aggregated $l=0$ effect equals the
cohort-weighted $ATT(g, t)$, which equals $1/2$ for each pair $(e, t) \in
	\{(2, 2), (3, 3), \ldots, (5, 5)\}$.
For $l>0$, the estimates of $LATT(e,t)$ are a convex combination of the
underlying $ATT(g, t, e)$ parameters.
For instance,
\begin{align*}
	\hat{LATT}(2, 4) = 1.25
	 & \approx
	\hat{ATT}(2, 4, 2)\hat{P}(G_{2} = 1 \mid D_{4} = 1, E_{2} = 1)
	\\ & \quad
	+ \hat{ATT}(3, 4, 2)\hat{P}(G_{3} = 1 \mid D_{4} = 1, E_{2} = 1)
	\\ & \quad
	+ \hat{ATT}(4, 4, 2)\hat{P}(G_{4} = 1 \mid D_{4} = 1, E_{2} = 1)
	\\
	& = 1.501 \cdot 0.615 + 1.001 \cdot 0.266 + 0.493 \cdot 0.119
\\
& = 1.248
,
\end{align*}
use the first two entries in row $(e,t) = (2,4)$ of
\Cref{tab:sim-exp-2} (the oracle and weight estimates are not reported
in the table).

The weights $\hat{P}(G_g=1 \mid D_4=1, E_2=1)$, $g = 2, 3, 4$, are decreasing, reflecting that
treatment is absorbing. Among units first exposed at $t=2$, a large share is
treated at $t=2$, a smaller share at $t=3$, and only a small
fraction remains untreated until $t=4$. The heterogeneity in $\hat{ATT}(g,4,2)$
captures the different treatment horizons.

Comparing the Oracle and $LATT(e,t)$ columns across both data settings,
we see that the columns are approximately equal, verifying the result in
\Cref{prop:bloom}. In contrast, the DiD estimates of $ATT(g,t)$ in the CSA
column are biased because of the unobserved confounder $H_t$ affecting
both the treatment and outcome, cf. \cref{eq:meat-treatment,eq:pos-simulation}.
The instrumental variable component of the IDiD estimator addresses this source
of bias and remains unbiased, illustrating the robustness of the IDiD design to
hidden confounding,
and how the IDiD estimators can be leveraged in settings where DiD estimators
fail.

\begin{table}[H]
	\centering

	\caption{Simulation experiment 3}
	\label{tab:sim-exp-2}

	\begin{threeparttable}

		\includegraphics[width=1\textwidth]{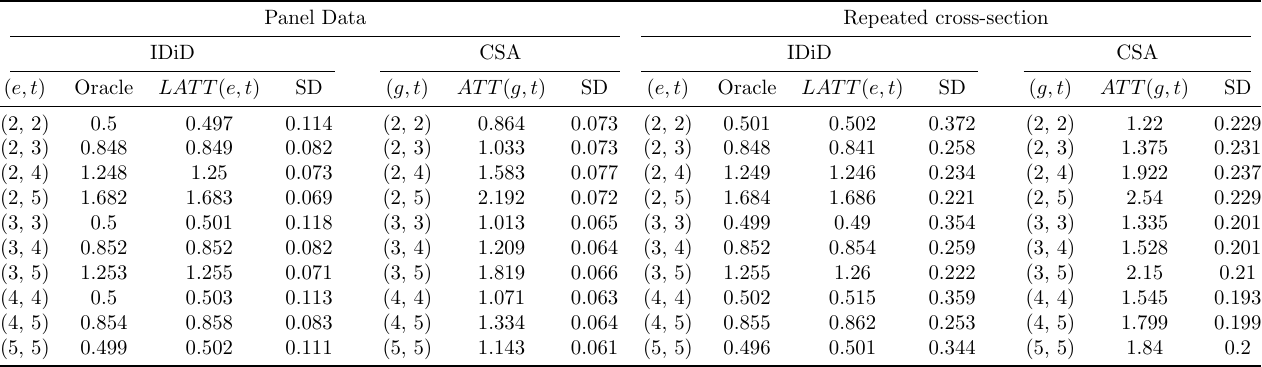}

		\begin{tablenotes}
			\footnotesize
			\item Notes: The table shows the simulation results of simulation
			experiment 3.
			Each entry is an average over $B=1499$ replications.
			The ``$LATT(e,t)$'' columns report the IDiD estimates, the
			``Oracle'' columns report the right-hand side of \cref{eq:bloom} using
			the simulated potential outcomes and estimated group probabilities, and
			the ``$ATT(g,t)$'' columns report the DiD estimates based on
			the estimator of \cite{csa}. The comparison between the ``$LATT(e,t)$''
			and ``Oracle'' columns verifies \Cref{prop:bloom}.
			The ''SD`` column reports the standard deviation of the $LATT(e, t)$
			and $ATT(g, t)$ over all simulation draws.
		\end{tablenotes}

	\end{threeparttable}
\end{table}

\section{Conclusion}
This paper develops doubly robust estimands for the $LATT(e, t)$ parameter in
the IDiD setting with covariates, covering both panel data and repeated
cross-sections, and allowing for never-exposed and not-yet-exposed control
groups. We also construct corresponding DR and DML estimators.

Our approach is estimand-based: using the influence function machinery of
\cite{kennedy}, adapted to our setting in
\Cref{sec:influence-functions}, we derive the DR estimands from first
principles. The simulation results confirm the validity
of the corresponding estimators across both data settings and control groups, and verify
the group-difference estimands and the Bloom result under absorbing
treatment, linking IDiD with staggered instrument exposure to DiD with
staggered treatment.
At the same time, the analysis highlights an important distinction between the
two data settings: while both admit doubly robust identification and
estimation, the repeated-cross-sections case does not inherit the same
remainder-term behavior as the panel case because of the additional nuisance
components.

More broadly, the paper underscores the value of an estimand-based approach: it
replaces reverse-engineering of regression parameters with the specification of
interpretable estimands, and enables modular reuse of components across
settings, in particular when deriving the corresponding EIFs.
It also provides another example of constructing Neyman-orthogonal
scores from EIFs (cf.~\cite[Section 2.2.5]{chernozhukovDoubleDebiasedMachine2018a};
see also \cite{chen2026equivalence}) and of two cases in which the
cross-fitted estimator coincides with the DML estimator.
Practically, the paper delivers both the DR and DML estimators, together with
a software implementation, making the proposed methods directly usable in
applied work.
We hope that the approach and components developed here, together with the work
on which they build, serve as a useful guide and toolbox for future work.

\newpage
\bibliography{refs}

@article{csa,
        title = {Difference-in-{{Differences}} with Multiple Time Periods},
        author = {Callaway, Brantly and Sant'Anna, Pedro H. C.},
        year = 2021,
        month = dec,
        journal = {Journal of Econometrics},
        series = {Themed {{Issue}}: {{Treatment Effect}} 1},
        volume = {225},
        number = {2},
        pages = {200--230},
        issn = {0304-4076},
        doi = {10.1016/j.jeconom.2020.12.001},
        url = {
               https://www.sciencedirect.com/science/article/pii/S0304407620303948
               },
        urldate = {2024-02-06},
}

@article{sazhao,
        title = {Doubly Robust Difference-in-Differences Estimators},
        author = {Sant'Anna, Pedro H.C. and Zhao, Jun},
        year = 2020,
        month = nov,
        journal = {Journal of Econometrics},
        volume = {219},
        number = {1},
        pages = {101--122},
        issn = {03044076},
        doi = {10.1016/j.jeconom.2020.06.003},
        url = {https://linkinghub.elsevier.com/retrieve/pii/S0304407620301901},
        urldate = {2024-08-29},
        langid = {english},
}

@misc{kennedy,
        title = {Semiparametric Doubly Robust Targeted Double Machine Learning:
                 A Review},
        shorttitle = {Semiparametric Doubly Robust Targeted Double Machine
                      Learning},
        author = {Kennedy, Edward H.},
        year = 2023,
        month = jan,
        number = {arXiv:2203.06469},
        eprint = {2203.06469},
        primaryclass = {stat},
        publisher = {arXiv},
        url = {http://arxiv.org/abs/2203.06469},
        urldate = {2024-02-15},
        archiveprefix = {arXiv},
}

@misc{miyaji2025,
        title = {Instrumented {{Difference-in-Differences}} with {{Heterogeneous
                 Treatment Effects}}},
        author = {Miyaji, Sho},
        year = 2025,
        month = jan,
        number = {arXiv:2405.12083},
        eprint = {2405.12083},
        primaryclass = {econ},
        publisher = {arXiv},
        doi = {10.48550/arXiv.2405.12083},
        url = {http://arxiv.org/abs/2405.12083},
        urldate = {2025-04-07},
        archiveprefix = {arXiv},
}

@misc{miyaji,
        title = {Instrumented {{Difference-in-Differences}} with {{Heterogeneous
                 Treatment Effects}}},
        author = {Miyaji, Sho},
        year = 2026,
        month = feb,
        journal = {arXiv.org},
        number = {arXiv:2405.12083},
        eprint = {2405.12083},
        primaryclass = {econ},
        publisher = {arXiv},
        doi = {10.48550/arXiv.2405.12083},
        url = {https://arxiv.org/abs/2405.12083v6},
        urldate = {2026-05-04},
        langid = {english},
        archiveprefix = {arXiv},
}

@article{abadieSemiparametricDifferenceinDifferencesEstimators2005a,
        title = {Semiparametric {{Difference-in-Differences Estimators}}},
        author = {Abadie, Alberto},
        year = 2005,
        month = jan,
        journal = {The Review of Economic Studies},
        volume = {72},
        number = {1},
        pages = {1--19},
        issn = {1467-937X, 0034-6527},
        doi = {10.1111/0034-6527.00321},
        url = {
               https://academic.oup.com/restud/article-lookup/doi/10.1111/0034-6527.00321
               },
        urldate = {2024-08-29},
        langid = {english},
}

@article{frolichExploitingRegionalTreatment2010,
        title = {Exploiting {{Regional Treatment Intensity}} for the {{
                 Evaluation}} of {{Labor Market Policies}}},
        author = {Fr{\"o}lich, Markus and Lechner, Michael},
        year = 2010,
        month = sep,
        journal = {Journal of the American Statistical Association},
        volume = {105},
        number = {491},
        pages = {1014--1029},
        publisher = {Taylor \& Francis},
        issn = {0162-1459},
        doi = {10.1198/jasa.2010.ap08148},
        url = {https://doi.org/10.1198/jasa.2010.ap08148},
        urldate = {2026-01-30},
}

@article{sloczynskiDoublyRobustEstimation2022,
        title = {Doubly {{Robust Estimation}} of {{Local Average Treatment
                 Effects Using Inverse Probability Weighted Regression Adjustment
                 }}},
        author = {S{\l}oczy{\'n}ski, Tymon and Uysal, S Derya and Wooldridge,
                  Jeffrey M},
        year = 2022,
        langid = {english},
}

@misc{csax,
        title = {Efficient {{Difference-in-Differences}} and {{Event Study
                 Estimators}}},
        author = {Chen, Xiaohong and Sant'Anna, Pedro H. C. and Xie, Haitian},
        year = 2025,
        month = jun,
        number = {arXiv:2506.17729},
        eprint = {2506.17729},
        primaryclass = {econ},
        publisher = {arXiv},
        doi = {10.48550/arXiv.2506.17729},
        url = {http://arxiv.org/abs/2506.17729},
        urldate = {2025-07-10},
        archiveprefix = {arXiv},
}

@article{chernozhukovDoubleDebiasedMachine2018a,
        title = {Double/Debiased Machine Learning for Treatment and Structural
                 Parameters},
        author = {Chernozhukov, Victor and Chetverikov, Denis and Demirer, Mert
                  and Duflo, Esther and Hansen, Christian and Newey, Whitney and
                  Robins, James},
        year = 2018,
        month = feb,
        journal = {The Econometrics Journal},
        volume = {21},
        number = {1},
        pages = {C1-C68},
        issn = {1368-4221, 1368-423X},
        doi = {10.1111/ectj.12097},
        url = {https://academic.oup.com/ectj/article/21/1/C1/5056401},
        urldate = {2025-02-13},
        copyright = {http://doi.wiley.com/10.1002/tdm\_license\_1.1},
        langid = {english},
}

@article{chang2020,
        title = {Double/Debiased Machine Learning for Difference-in-Differences
                 Models},
        author = {Chang, Neng-Chieh},
        year = 2020,
        month = may,
        journal = {The Econometrics Journal},
        volume = {23},
        number = {2},
        pages = {177--191},
        issn = {1368-4221, 1368-423X},
        doi = {10.1093/ectj/utaa001},
        url = {https://academic.oup.com/ectj/article/23/2/177/5722119},
        urldate = {2025-08-26},
        copyright = {https://academic.oup.com/journals/pages/open\_access/funder
                     \_policies/chorus/standard\_publication\_model},
        langid = {english},
}

@misc{metalearner,
        title = {A {{Meta-learner}} for {{Heterogeneous Effects}} in {{
                 Difference-in-Differences}}},
        author = {Lan, Hui and Chang, Haoge and Dillon, Eleanor and Syrgkanis,
                  Vasilis},
        year = 2025,
        month = apr,
        number = {arXiv:2502.04699},
        eprint = {2502.04699},
        primaryclass = {stat},
        publisher = {arXiv},
        doi = {10.48550/arXiv.2502.04699},
        url = {http://arxiv.org/abs/2502.04699},
        urldate = {2026-01-21},
        archiveprefix = {arXiv},
        langid = {english},
}

@article{deng,
        title = {Improved Two-Period Difference-in-Differences by Targeted
                 Estimation},
        author = {Deng, Yuhao and Zhang, Tao and Peng, Xiang and Liu, Qinqing},
        year = 2025,
        month = dec,
        journal = {Economics Letters},
        volume = {257},
        pages = {112600},
        issn = {01651765},
        doi = {10.1016/j.econlet.2025.112600},
        url = {https://linkinghub.elsevier.com/retrieve/pii/S0165176525004379},
        urldate = {2026-03-06},
        langid = {english},
}

@article{dukesvansteelandt,
        title = {On doubly robust inference for double machine learning in
                 semiparametric regression},
        author = {Dukes, Oliver and Vansteelandt, Stijn and Whitney, David},
        journal = {Journal of Machine Learning Research},
        volume = {25},
        number = {279},
        pages = {1--46},
        year = {2024},
}

@incollection{mogstad2024instrumental,
        title = {Instrumental variables with unobserved heterogeneity in
                 treatment effects},
        author = {Mogstad, Magne and Torgovitsky, Alexander},
        booktitle = {Handbook of labor economics},
        volume = {5},
        pages = {1--114},
        year = {2024},
        publisher = {Elsevier},
}

@article{newey1990,
        title = {Semiparametric Efficiency Bounds},
        author = {Newey, Whitney K.},
        year = 1990,
        journal = {Journal of Applied Econometrics},
        volume = {5},
        number = {2},
        pages = {99--135},
        issn = {1099-1255},
        doi = {10.1002/jae.3950050202},
        url = {https://onlinelibrary.wiley.com/doi/abs/10.1002/jae.3950050202},
        urldate = {2023-12-21},
        copyright = {Copyright \copyright{} 1990 John Wiley \& Sons, Ltd.},
        langid = {english},
}

@book{ap,
        title = {Mostly {{Harmless Econometrics}}: {{An Empiricist}}'s {{
                 Companion}}},
        shorttitle = {Mostly {{Harmless Econometrics}}},
        author = {Angrist, Joshua D. and Pischke, J{\"o}rn-Steffen},
        year = 2009,
        eprint = {j.ctvcm4j72},
        eprinttype = {jstor},
        publisher = {Princeton University Press},
        doi = {10.2307/j.ctvcm4j72},
        url = {https://www.jstor.org/stable/j.ctvcm4j72},
        urldate = {2024-09-07},
        isbn = {978-0-691-12034-8},
}

@article{bloom1984estimating,
        title = {Estimating the effect of job-training programs, using
                 longitudinal data: Ashenfelter's findings reconsidered},
        author = {Bloom, Howard S},
        journal = {Journal of Human Resources},
        pages = {544--556},
        year = {1984},
        publisher = {JSTOR},
}

@article{dml,
        title = {Double/Debiased Machine Learning for Treatment and Structural
                 Parameters},
        author = {Chernozhukov, Victor and Chetverikov, Denis and Demirer, Mert
                  and Duflo, Esther and Hansen, Christian and Newey, Whitney and
                  Robins, James},
        year = 2018,
        month = feb,
        journal = {The Econometrics Journal},
        volume = {21},
        number = {1},
        pages = {C1-C68},
        issn = {1368-4221, 1368-423X},
        doi = {10.1111/ectj.12097},
        url = {https://academic.oup.com/ectj/article/21/1/C1/5056401},
        urldate = {2025-02-13},
        copyright = {http://doi.wiley.com/10.1002/tdm\_license\_1.1},
        langid = {english},
}

@article{tan2006,
        title = {Regression and {{Weighting Methods}} for {{Causal Inference
                 Using Instrumental Variables}}},
        author = {Tan, Zhiqiang},
        year = 2006,
        month = dec,
        journal = {Journal of the American Statistical Association},
        volume = {101},
        number = {476},
        pages = {1607--1618},
        issn = {0162-1459, 1537-274X},
        doi = {10.1198/016214505000001366},
        url = {https://www.tandfonline.com/doi/full/10.1198/016214505000001366},
        urldate = {2026-01-30},
        langid = {english},
}

@book{tsiatis,
        title = {Semiparametric {{Theory}} and {{Missing Data}}},
        author = {Tsiatis, Anastasios},
        year = 2006,
        series = {Springer {{Series}} in {{Statistics}}},
        publisher = {Springer New York},
        address = {New York, NY},
        doi = {10.1007/0-387-37345-4},
        url = {http://link.springer.com/10.1007/0-387-37345-4},
        urldate = {2024-02-11},
        isbn = {978-0-387-32448-7},
        langid = {english},
}

@article{van2015statistics,
        title = {Statistics as a science, not an art: the way to survive in data
                 science},
        author = {van der Laan, Mark},
        journal = {Amstat News},
        volume = {1},
        pages = {292},
        year = {2015},
}

@book{vdlrose,
        title = {Targeted {{Learning}}: {{Causal Inference}} for {{Observational
                 }} and {{Experimental Data}}},
        shorttitle = {Targeted {{Learning}}},
        author = {Van Der Laan, Mark J. and Rose, Sherri},
        year = 2011,
        series = {Springer {{Series}} in {{Statistics}}},
        publisher = {Springer},
        address = {New York, NY},
        doi = {10.1007/978-1-4419-9782-1},
        url = {https://link.springer.com/10.1007/978-1-4419-9782-1},
        urldate = {2024-01-30},
        isbn = {978-1-4419-9781-4 978-1-4419-9782-1},
        langid = {english},
}

@article{bach2022doubleml,
        title = {DoubleML-an object-oriented implementation of double machine
                 learning in python},
        author = {Bach, Philipp and Chernozhukov, Victor and Kurz, Malte S and
                  Spindler, Martin},
        journal = {Journal of Machine Learning Research},
        volume = {23},
        number = {53},
        pages = {1--6},
        year = {2022},
}

@article{chen2026equivalence,
        title = {On the Equivalence between Neyman Orthogonality and Pathwise
                 Differentiability},
        author = {Chen, Yuxi and Kennedy, Edward H and Balakrishnan, Sivaraman},
        journal = {arXiv preprint arXiv:2603.15817},
        year = {2026},
}

\newpage
\appendix
\section{Proofs}
\label{sec:proofs}
\subsection{Panel Data Case}
\label{sec:proofs-panel}

Define the LATT conditional on $X$:
\begin{align}
	\label{eq:lattx}
	LATT(e, t, X) :=
	E [ Y_{t}(1) - Y_{t}(0)
			\mid X, E_{e} = 1, D_{t}(e) > D_{t}(\infty)
		]
\end{align}
Also:
\begin{align}
	LATT^{num}(e, t, X)
	 & :=
	LATT(e, t, X)
	P(D_{t}(e) > D_{t}(\infty)\mid X, E_{e} = 1) \\
	LATT^{den}(e, t, X)
	 & :=
	P(D_{t}(e) > D_{t}(\infty)\mid X, E_{e} = 1)
\end{align}
such that by definition
$LATT(e, t, X) = LATT^{num}(e, t, X)/LATT^{den}(e, t, X)$.

\begin{lemma}
	\label{lemma:id-panel-num}
	\begin{align}
		\label{eq:latt-num-lie}
		E [ LATT^{num}(e, t, X) \mid E_{e} = 1 ]
		=
		LATT(e, t) P(D_{t}(e) > D_{t}(\infty)\mid E_{e} = 1)
	\end{align}
\end{lemma}
\begin{proof}
	First, notice:
	\begin{align*}
		 & E \left(
		[Y_{t}(1) - Y_{t}(0)]\mathbf{1}\{D_{t}(e) > D_{t}(\infty)\}
		\mid X, E_{e} = 1
		\right)
		\\&=
		E \left(
		Y_{t}(1) - Y_{t}(0)
		\mid X, E_{e} = 1, D_{t}(e) > D_{t}(\infty)
		\right)
		P(D_{t}(e) > D_{t}(\infty)\mid X, E_{e} = 1)
		\\
		 & =
		LATT^{num}(e, t, X)
	\end{align*}
	An application of the LIE yields:
	\begin{align*}
		 & E \left\{
		LATT^{num}(e, t, X)
		\mid E_{e} = 1
		\right\}
		\\&=
		E \left\{
		E \left(
		[Y_{t}(1) - Y_{t}(0)]\mathbf{1}\{D_{t}(e) > D_{t}(\infty)\}
		\mid X, E_{e} = 1
		\right)
		\mid E_{e} = 1
		\right\}
		\\&=
		E \left\{
		[Y_{t}(1) - Y_{t}(0)]\mathbf{1}\{D_{t}(e) > D_{t}(\infty)\}
		\mid E_{e} = 1
		\right\}
		\\&=
		E \left\{
		Y_{t}(1) - Y_{t}(0)
		\mid E_{e} = 1, D_{t}(e) > D_{t}(\infty)
		\right\}
		P(D_{t}(e) > D_{t}(\infty) \mid E_{e} = 1)
		\\
		 & =
		LATT(e, t) P(D_{t}(e) > D_{t}(\infty)\mid E_{e} = 1)
	\end{align*}
	where the second last equal sign used the LIE together with Monotonicity
	\Cref{ass:monotonicity}.
\end{proof}

\begin{proposition}
	\label{prop:latt-lie}
	\begin{align}
		\label{eq:latt-lie}
		LATT(e, t) =
		\frac{
			E[LATT^{num}(e, t, X) \mid E_{e} = 1]
		}{
			E[LATT^{den}(e, t, X) \mid E_{e} = 1]
		}
	\end{align}
\end{proposition}
\begin{proof}
	We already have an expression for the numerator in \cref{eq:latt-num-lie}.
	For the denominator, an application of the LIE yields:
	\begin{align}
		\nonumber
		E[LATT^{den}(e, t, X) \mid E_{e} = 1]
		 & =
		E[P(D_{t}(e) > D_{t}(\infty) \mid X, E_{e} = 1) \mid E_{1} = 1]
		\\&=
		\label{eq:latt-denom-lie}
		P(D_{t}(e) > D_{t}(\infty) \mid E_{e} = 1)
	\end{align}
	Combining \cref{eq:latt-num-lie,eq:latt-denom-lie}
	yields \cref{eq:latt-lie}.
\end{proof}

\begin{proposition}
	Recall mean functions in \cref{eq:mean-functions-panel}.
	\cref{eq:identification-panel}.
	Identification of conditional LATT:
	\begin{align}
		\label{eq:lattx-id}
		LATT(e, t, X)
		 & =
		\frac{
			m_{e, t}^{trt, p}(X) - m_{e, t}^{c, p}(X)
		}{
			g_{e, t}^{trt, p}(X) - g_{e, t}^{c, p}(X)
		}
	\end{align}
\end{proposition}
\begin{proof}
	\textbf{Numerator:}
	\begin{align}
		\nonumber
		          & m_{e, t}^{trt, p}(X) - m_{e, t}^{c, p}(X)
		\\
		\label{eq:numerator-proof-panel}
		          & =
		E[\Delta_{t-e+1}Y_{t} \mid X, E_{e} = 1]
		- E[\Delta_{t-e+1}Y_{t} \mid X, C = 1]
		\\
		\nonumber
		          & =
		E[Y_{t}(D_{t}(e)) - Y_{e - 1}(D_{e-1}(e)) \mid X, E_{e} = 1]
		- E[Y_{t}(D_{t}(\infty)) - Y_{e - 1}(D_{e-1}(\infty)) \mid X, C = 1]
		\\
		\nonumber
		          & =
		E[Y_{t}(D_{t}(e)) - Y_{e - 1}(D_{e-1}(\infty)) \mid X, E_{e} = 1]
		- E[Y_{t}(D_{t}(\infty)) - Y_{e - 1}(D_{e-1}(\infty)) \mid X, C = 1]
		\\
		\nonumber
		          & =
		E[Y_{t}(D_{t}(e)) - Y_{t}(D_{t}(\infty)) \mid X, E_{e} = 1]
		\\
		\nonumber & \quad +
		E[Y_{t}(D_{t}(\infty)) - Y_{e - 1}(D_{e-1}(\infty)) \mid X, E_{e} = 1]
		- E[Y_{t}(D_{t}(\infty)) - Y_{e - 1}(D_{e-1}(\infty)) \mid X, C = 1]
		\\
		\nonumber
		          & =
		E[Y_{t}(D_{t}(e)) - Y_{t}(D_{t}(\infty)) \mid X, E_{e} = 1]
		\\
		\nonumber & \quad +
		\sum_{l=0}^{t-e} E[\Delta Y_{t-l}(D_{t-l}(\infty))\mid X, E_{e} = 1]
		- \sum_{l=0}^{t-e} E[\Delta Y_{t-l}(D_{t-l}(\infty)) \mid X, C = 1]
		\\
		\nonumber
		          & =
		E[Y_{t}(D_{t}(e)) - Y_{t}(D_{t}(\infty)) \mid X, E_{e} = 1]
		\\
		\nonumber
		          & =
		E\{
		[Y_{t}(1) - Y_{t}(0)][D_{t}(e) - D_{t}(\infty)] \mid E_{e} = 1,
		X
		\}
		\\
		\nonumber
		          & =
		E \left(
		Y_{t}(1) - Y_{t}(0)
		\mid E_{e} = 1, D_{t}(e) > D_{t}(\infty), X
		\right)
		P(D_{t}(e) > D_{t}(\infty)\mid X, E_{e} = 1)
		\\
		\label{eq:lattx-id-num}
		          & =
		LATT^{num}(e, t, X)
	\end{align}
	where we in the second equal sign used the definition of the exposed/unexposed
	outcomes \cref{eq:exposed-unexposed-outcome},
	in the third equal sign used \Cref{ass:no-anticipation},
	in the fourth equal sign used the add-$0$-trick;
	likewise in the fifth equal sign
	for $t-1, t-2, \ldots, e - 2$,
	using the notation $\Delta Y_{t}(D_{t}(\infty)) := Y_{t}(D_{t}(\infty)) -
		Y_{t-1}(D_{t-1}(\infty))$;
	in the sixth equality used
	\Cref{ass:cpt-unexpout-ne}
	or \Cref{ass:cpt-unexpout-nye},
	and in the seventh equality used monotonicity \Cref{ass:monotonicity}
	and the LIE.

	\medskip\noindent
	\textbf{Denominator}
	\begin{align}
		\nonumber
		          & g_{e, t}^{trt, p}(X) - g_{e, t}^{c, p}(X)
		\\
		\label{eq:denominator-proof-panel}
		          & =
		E[\Delta_{t-e+1}D_{t} \mid X, E_{e} = 1]
		- E[\Delta_{t-e+1}D_{t} \mid X, C = 1]
		\\
		\nonumber
		          & =
		E[D_{t}(e) - D_{e-1}(e) \mid X, E_{e} = 1]
		- E[D_{t}(\infty) - D_{e-1}(\infty) \mid X, C = 1]
		\\
		\nonumber
		          & =
		E[D_{t}(e) - D_{e-1}(\infty) \mid X, E_{e} = 1]
		- E[D_{t}(\infty) - D_{e-1}(\infty) \mid X, C = 1]
		\\
		\nonumber
		          & =
		E[D_{t}(e) - D_{t}(\infty) \mid X, E_{e} = 1]
		\\
		\nonumber & \quad +
		E[D_{t}(\infty) - D_{e-1}(\infty) \mid X, E_{e} = 1]
		- E[D_{t}(\infty) - D_{e-1}(\infty) \mid X, C = 1]
		\\
		\nonumber
		          & =
		E[D_{t}(e) - D_{t}(\infty) \mid X, E_{e} = 1]
		\\
		\nonumber & \quad +
		\sum_{l=0}^{t-e} E[\Delta D_{t-l}(\infty)\mid X, E_{e} = 1]
		- \sum_{l=0}^{t-e} E[\Delta D_{t-l}(\infty) \mid X, C = 1]
		\\
		\nonumber
		          & =
		E[D_{t}(e) - D_{t}(\infty) \mid X, E_{e} = 1]
		\\
		\nonumber
		          & =
		P(D_{t}(e) > D_{t}(\infty) \mid X, E_{e} = 1)
		\\
		\label{eq:lattx-id-denom}
		          & =
		LATT^{den}(e, t, X)
	\end{align}
	using similar steps as for the outcome,
	the conditional parallel trends in the treatment,
	\Cref{ass:cpt-treat-ne}
	or \Cref{ass:cpt-treat-nye},
	and in the last equal sign that
	$D_{t}(e) - D_{t}(\infty)$ is binary by monotonicity.

	\medskip\noindent
	Dividing \eqref{eq:lattx-id-num} by \eqref{eq:lattx-id-denom} yields
	\cref{eq:lattx-id}.
\end{proof}

\begin{proofref}{prop:identification-panel}
	Taking the conditional expectation conditional on $(E_{e} = 1)$
	in
	\eqref{eq:lattx-id-num}
	and \eqref{eq:lattx-id-denom}
	yields:
	\begin{align}
		\label{eq:id-latt-reg}
		 & \frac{
		E[
		E[\Delta_{t-e+1}Y_{t} \mid X, E_{e} = 1]
		- E[\Delta_{t-e+1}Y_{t} \mid X, C = 1] \mid E_{e}
		= 1
		]
		}{
		E[
		E[\Delta_{t-e+1}D_{t} \mid X, E_{e} = 1]
		- E[\Delta_{t-e+1}D_{t} \mid X, C = 1]
		\mid E_{e} = 1
		]
		}
		\\
		\nonumber
		 & =
		\frac{E[
					m_{e, t}^{trt, p}(X) - m_{e, t}^{c, p}(X) \mid E_{e} = 1
				]
		}{E[
					g_{e, t}^{trt, p}(X) - g_{e, t}^{c, p}(X) \mid E_{e} = 1
				]}
		\\
		\nonumber
		 & =
		\frac{
			E[LATT^{num}(e, t, X) \mid E_{e} = 1]
		}{
			E[LATT^{den}(e, t, X) \mid E_{e} = 1]
		}
		\\
		\nonumber
		 & =
		LATT(e, t)
	\end{align}
	where the last equality follows by \Cref{prop:latt-lie}.
\end{proofref}

\subsection{Repeated Cross-Sections Case}
\label{sec:proofs-rc}
\begin{proposition}
	\label{prop:id-num-denom-rc}
	Let $C$ be either of the control variables $C^{nev}$ and $C^{nye}_{e, s}$.
	Recall the mean functions \cref{eq:mean-functions-rc}.
	Then:
	\begin{align}
		LATT^{num}(e, t, X)
		 & = m_{e,t}^{trt, rc}(X)
		- m_{e,e-1}^{trt, rc}(X)
		-
		[
		m_{e, t, t}^{rc, c}(X)
		- m_{e, t, e-1}^{rc, c}(X)
		]
		\; \text{a.s.}
		\\
		LATT^{den}(e, t, X)
		 & = g_{e,t}^{trt, rc}(X)
		- g_{e,e-1}^{trt, rc}(X)
		-
		[
		g_{e, t, t}^{rc, c}(X)
		- g_{e, t, e-1}^{rc, c}(X)
		]
		\; \text{a.s.}
	\end{align}
\end{proposition}
\begin{proof}
	For the numerator:
	\begin{align*}
		 & m_{e,t}^{trt, rc}(X)
		- m_{e,e-1}^{trt, rc}(X)
		-
		[
		m_{e, t, t}^{rc, c}(X)
		- m_{e, t, e-1}^{rc, c}(X)
		]
		\\
		 & = E[Y \mid X, E_{e} = 1, T = t]
		- E[Y \mid X, E_{e} = 1, T = e - 1]
		\\ & \quad
		-
		\left(
		E[Y \mid X, C = 1, T = t]
		- E[Y \mid X, C = 1, T = e - 1]
		\right)
		\\
		 & =
		E[Y_{t} - Y_{e - 1} \mid E_{e} = 1, X]
		- E[Y_{t} - Y_{e - 1} \mid C = 1, X]
		\\
		 & =
		LATT^{num}(e, t, X)
	\end{align*}
	where
	the second last equality follows by \Cref{ass:sample:mixture}
	and the last equality follows by \cref{eq:lattx-id-num}.
	For the denominator:
	\begin{align*}
		 & g_{e,t}^{trt, rc}(X)
		- g_{e,e-1}^{trt, rc}(X)
		-
		[
		g_{e, t, t}^{rc, c}(X)
		- g_{e, t, e-1}^{rc, c}(X)
		]
		\\
		 & = E[D \mid X, E_{e} = 1, T = t]
		- E[D \mid X, E_{e} = 1, T = e - 1]
		\\ & \quad
		-
		\left(
		E[D \mid X, C = 1, T = t]
		- E[D \mid X, C = 1, T = e - 1]
		\right)
		\\
		 & =
		E[D_{t} - D_{e - 1} \mid E_{e} = 1, X]
		- E[D_{t} - D_{e - 1} \mid C = 1, X]
		\\
		 & =
		LATT^{den}(e, t, X)
	\end{align*}
	where
	the second last equality follows by \Cref{ass:sample:mixture}
	and the last equality follows by \cref{eq:lattx-id-denom}.
\end{proof}

\medskip\noindent
\textbf{Proof of \Cref{prop:identification-rc}}
\begin{proof}
	A similar argument to \Cref{lemma:id-panel-num,}
	and \Cref{prop:latt-lie} using the mixture measure $P_{M}$
	allows us to write the $LATT(e, t)$ estimand in the ratio form as:
	\begin{align}
		LATT(e, t) =
		\frac{
			E[LATT^{num}(e, t, X) \mid E_{e} = 1, T = t]
		}{
			E[LATT^{den}(e, t, X) \mid E_{e} = 1, T = t]
		}
	\end{align}
	Next, exploiting that $LATT^{num}(e, t, X)$
	and $LATT^{den}(e, t, X)$ are functions of $X$ only,
	and the invariance in \Cref{ass:sample:mixture},
	yields:
	\begin{align}
		\nonumber
		 & \frac{
			E[LATT^{num}(e, t, X) \mid E_{e} = 1, T = t]
		}{
			E[LATT^{den}(e, t, X) \mid E_{e} = 1, T = t]
		}
		\\
		 & =
		\frac{
			E[LATT^{num}(e, t, X) \mid E_{e} = 1]
		}{
			E[LATT^{den}(e, t, X) \mid E_{e} = 1]
		}
		\\
		\nonumber
		 & =
		\frac{
		E[
		m_{e,t}^{trt, rc}(X)
		- m_{e,e-1}^{trt, rc}(X)
		-
		[
		m_{e, t, t}^{rc, c}(X)
		- m_{e, t, e-1}^{rc, c}(X)
		] \mid E_{e} = 1
		]
		}{
		E[
		g_{e,t}^{trt, rc}(X)
		- g_{e,e-1}^{trt, rc}(X)
		-
		[
		g_{e, t, t}^{rc, c}(X)
		- g_{e, t, e-1}^{rc, c}(X)
		] \mid E_{e} = 1
		]
		}
	\end{align}
	where the last equality follows from \Cref{prop:id-num-denom-rc}.
\end{proof}

\subsection{Other Proofs}
\label{sec:proof-other}

\begin{proofref}{prop:dr-both}
	For the panel case, write the normalized control weights
	\cref{eq:weights-nev-panel-norm,eq:weights-nye-panel-norm} as
	$w^{c,p}_{e,t}(C, X; \pi)$, where $\pi$ is a propensity score model and
	$C$ is a generic control indicator.

	For the repeated cross-sections case, write the normalized control weights
	\cref{eq:weights-nev-rc-norm,eq:weights-nye-rc-norm} as
	$w^{c,rc}_{e,t}(C, X; \pi)$, suppressing the period subscript $s$ in the
	not-yet-exposed case. Similarly, write
	$\mu^{c,rc}_{e,t}(X)$ and $\eta^{c,rc}_{e,t}(X)$ for working models of
	$m^{c,rc}_{e,t}(X)$ and $g^{c,rc}_{e,t}(X)$, again suppressing the period
	subscript $s$.

	Akin to \cite{csa}, the normalizations of the weights
	\cref{eq:weights-nye-panel-norm,eq:weights-nev-rc-norm} reduce under correct
	specification, in the panel case, as:
	\begin{align}
		\label{eq:normalizations-reduce-csa}
		E\left[C\frac{p(X)}{1 - p(X)}\right]
		= E[E_{e}],
	\end{align}
	and in the repeated cross-sections case,
	\cref{eq:weights-nev-rc-norm,eq:weights-nye-rc-norm}, as
	\begin{align}
		\label{eq:csa-denom-rc}
		E\left[
			\frac{C \cdot \mathbf{1}\{T = t\}p(X)}{1 - p(X)}
			\right]
		= E[E_{e}] \cdot \lambda_{t},
	\end{align}
	matching the normalization constants from the influence function
	panel data \cref{eq:weights-panel}
	and repeated cross-sections \cref{eq:weights0-rc}.
	Also, in the repeated cross-sections case,
	\begin{align}
		\label{eq:csa-denom-rc-treat}
		E[E_{e} \cdot \mathbf{1}\{T = t\}]
		=E[E_{e}] \cdot \lambda_{t}
	\end{align}
	by the invariance assumption in \Cref{ass:sample:mixture}.

	\medskip\noindent
	\paragraph{Case 1: Panel Data and propensity correctly specified: $\pi(X) = p(X)$ a.s.}
	Write:
	\begin{align*}
		 & E[
		\{ w^{trt,p}_{e} - w^{c,p}_{e, t}(C, X; \pi) \}
		\{\Delta_{t-e+1}Y_{t} - \mu_{e, t}^{c, p}(X)\}
		]
		\\
		 & =
		E[
		\{ w^{trt,p}_{e} - w^{c,p}_{e, t}(C, X; p) \}
		\{\Delta_{t-e+1}Y_{t} - \mu_{e, t}^{c, p}(X)\}
		]
		\\
		 & =
		E[
		\{ w^{trt,p}_{e} - w^{c,p}_{e, t}(C, X; p) \}
		\Delta_{t-e+1}Y_{t}
		]
		-
		E[
				\{ w^{trt,p}_{e} - w^{c,p}_{e, t}(C, X; p) \}
				\mu_{e, t}^{c, p}(X)
			]
	\end{align*}
	The second term equals $0$ by a LIE application:
	\begin{align*}
		E\left[
			\left\{ \frac{E_{e}}{E[E_{e}]} -
			\frac{
				C\frac{p(X)}{1 - p(X)}
			}{
				E\left[C\frac{p(X)}{1 - p(X)}\right]
			}
			\right\}
			\mu_{e, t}^{c, p}(X)
			\right]
		 & =
		E\left[
			\left\{ \frac{E_{e}}{E[E_{e}]} -
			\frac{C}{E[E_{e}]}\frac{p(X)}{1 - p(X)}
			\right\}
			\mu_{e, t}^{c, p}(X)
			\right]
		\\
		 & =
		E\left[
			\left\{ \frac{p(X)}{E[E_{e}]} -
			\frac{1 - p(X)}{E[E_{e}]}\frac{p(X)}{1 - p(X)}
			\right\}
			\mu_{e, t}^{c, p}(X)
			\right]
		\\
		 & = 0
	\end{align*}
	using \cref{eq:normalizations-reduce-csa},
	and $E_{e} + C = 1$ so $P(C = 1 \mid X) = 1 - P(E_{e} = 1 \mid X) = 1 -
		p(X)$.
	For the first term we have similarly:
	\begin{align*}
		 & E\left[
		\left\{ \frac{E_{e}}{E[E_{e}]} -
		\frac{C}{E[E_{e}]}\frac{p(X)}{1 - p(X)}
		\right\}
		\Delta_{t-e+1}Y_{t}
		\right]
		\\
		 & =
		E\left[
		\Delta_{t-e+1}Y_{t} \mid E_{e} = 1
		\right]
		-
		\frac{1}{E[E_{e}]}
		E\left[
		\frac{p(X)}{1 - p(X)}
		E[C \mid X]
		E[\Delta_{t-e+1}Y_{t} \mid X, C = 1]
		\right]
		\\
		 & =
		E\left[
		E[\Delta_{t-e+1}Y_{t} \mid X, E_{e} = 1] \mid E_{e} = 1
		\right]
		-
		\frac{1}{E[E_{e}]}
		E\left[
		p(X)
		E[\Delta_{t-e+1}Y_{t} \mid X, C = 1]
		\right]
		\\
		 & =
		E\left[
		E[\Delta_{t-e+1}Y_{t} \mid X, E_{e} = 1]
		-
		E[\Delta_{t-e+1}Y_{t} \mid X, C = 1]
		\mid E_{e} = 1
		\right]
	\end{align*}
	i.e. the numerator of \cref{eq:id-latt-reg}.

	A similar argument applied to the denominator yields:
	\begin{align*}
		E\left[
			\left\{ \frac{E_{e}}{E[E_{e}]} -
			\frac{
				C\frac{p(X)}{1 - p(X)}
			}{
				E\left[C\frac{p(X)}{1 - p(X)}\right]
			}
			\right\}
			\eta_{e, t}^{c, p}(X)
			\right] = 0
	\end{align*}
	and
	\begin{align*}
		 & E\left[
		\left\{ \frac{E_{e}}{E[E_{e}]} -
		\frac{C}{E[E_{e}]}\frac{p(X)}{1 - p(X)}
		\right\}
		\Delta_{t-e+1}D_{t}
		\right]
		\\
		 & =
		E\left[
		E[\Delta_{t-e+1}D_{t} \mid X, E_{e} = 1]
		-
		E[\Delta_{t-e+1}D_{t} \mid X, C = 1]
		\mid E_{e} = 1
		\right]
	\end{align*}
	i.e. the denominator of \cref{eq:id-latt-reg}.
	Taking this together yields the identification
	\begin{align*}
		\tau^{dr, p}_{e, t}
		 & = \frac{
		E[
		\{ w^{trt,p}_{e} - w^{c,p}_{e, t}(C, X; p) \}
		\{\Delta_{t-e+1}Y_{t} - \mu_{e, t}^{c, p}(X)\}
		]
		}{
		E[
		\{ w^{trt,p}_{e} - w^{c,p}_{e, t}(C, X; p) \}
		\{\Delta_{t-e+1}D_{t} - \eta_{e, t}^{c, p}(X)\}
		]
		}
		=
		LATT(e, t).
	\end{align*}

	\medskip\noindent
	\paragraph{ Case 2: Panel data and outcome regressions correctly specified }
	In this case,
	\begin{align*}
		\mu_{e, t}^{c, p}(X) = m_{e, t}^{c, p}(X) \, \text{a.s.},
		\quad
		\eta_{e, t}^{c, p}(X) = g_{e, t}^{c, p}(X) \, \text{a.s.}
	\end{align*}
	Consider the numerator:
	\begin{align*}
		 & E[
		\{ w^{trt,p}_{e} - w^{c,p}_{e, t}(C, X; \pi) \}
		\{\Delta_{t-e+1}Y_{t} - \mu_{e, t}^{c, p}(X)\}
		]
		\\
		 & =
		E[
		\{ w^{trt,p}_{e} - w^{c,p}_{e, t}(C, X; \pi) \}
		\{\Delta_{t-e+1}Y_{t} - m_{e, t}^{c, p}(X)\}
		]
		\\
		 & =
		E\left[
		\frac{E_{e}}{E[E_{e}]}
		\{E[\Delta_{t-e+1}Y_{t} \mid X, E_{e} = 1] - m_{e, t}^{c, p}(X)\}
		\right]
		\\ & \quad
		- E\left[
		\frac{
			C\frac{\pi(X)}{1 - \pi(X)}
		}{
			E\left[C\frac{\pi(X)}{1 - \pi(X)}\right]
		}
		\{E[\Delta_{t-e+1}Y_{t} \mid X, C = 1] - m_{e, t}^{c, p}(X)\}
		\right]
		\\
		 & =
		E\left[
			m_{e, t}^{trt, p}(X) - m_{e, t}^{c, p}(X)
			\mid E_{e} = 1
			\right]
	\end{align*}
	using the definitions \cref{eq:mean-functions-panel}.
	Similarly for the denominator,
	\begin{align*}
		 & E[
		\{ w^{trt,p}_{e} - w^{c,p}_{e, t}(C, X; \pi) \}
		\{\Delta_{t-e+1}D_{t} - \eta_{e, t}^{c, p}(X)\}
		]
		=
		E\left[
			g_{e, t}^{trt, p}(X) - g_{e, t}^{c, p}(X)
			\mid E_{e} = 1
			\right]
	\end{align*}
	These two terms constitute again the fraction \cref{eq:id-latt-reg},
	hence:
	\begin{align*}
		\tau^{dr, p}_{e, t}
		=
		\frac{
		E[
		\{ w^{trt,p}_{e} - w^{c,p}_{e, t}(C, X; \pi) \}
		\{\Delta_{t-e+1}Y_{t} - m_{e, t}^{c, p}(X)\}
		]
		}{
		E[
		\{ w^{trt,p}_{e} - w^{c,p}_{e, t}(C, X; \pi) \}
		\{\Delta_{t-e+1}D_{t} - g_{e, t}^{c, p}(X)\}
		]
		}
		= LATT(e, t).
	\end{align*}

	\paragraph{Case 3: Repeated cross-sections and propensity
		score model correctly specified: $\pi(X) = p(X)$ a.s.}
	Consider the numerator of \cref{eq:dr-estimand-rc}.
	As in \cite{sazhao},
	by \Cref{ass:sample:mixture},
	for any generic integrable and measurable function $\mu$,
	$E[\mu(X)|E_{e} = 1] = E[\mu(X)|E_{e} = 1, T = t]$
	for any $t$.
	Thus,
	\begin{equation}
		\label{eq:kappa-rc}
		E[\kappa_{e, t}^{Y, rc}] = 0,
	\end{equation}
	Hence:
	\begin{align}
		\nonumber
		 & E [
				\{w^{trt,rc}_{e} - w^{c,rc}_{e}(X, C; \pi)\}
				\{Y - \mu_{e, Y}^{c,rc}(X)\}
				+ \kappa_{e, t}^{Y, rc}
			]
		\\
		\label{eq:rc-case3-terms}
		 & =
		E [
				\{w^{trt,rc}_{e} - w^{c,rc}_{e}(X, C; p)\}Y
			]
		-
		E [
				\{w^{trt,rc}_{e} - w^{c,rc}_{e}(X, C; p)\}
				\mu_{e, Y}^{c,rc}(X)
			]
	\end{align}
	Hence,
	\begin{align*}
		E[w^{trt,rc}_{e}\mid X]
		 & = E[w^{trt,rc}_{e, t} \mid X]
		- E[w^{trt,rc}_{e, e - 1}  \mid X]
		\\
		 & =
		\frac{E[E_{e} \cdot \mathbf{1}\{T = t\} \mid X]}{E[E_{e}] \cdot \lambda_{t}}
		- \frac{E[E_{e} \cdot \mathbf{1}\{T = e - 1\} \mid X]}{E[E_{e}] \cdot
			\lambda_{e - 1}}
		\\
		 & =
		\frac{p(X) \lambda_{t}}{E[E_{e}] \cdot \lambda_{t}}
		- \frac{p(X) \lambda_{e - 1}}{E[E_{e}] \cdot \lambda_{e - 1}}
		\\
		 & = 0
	\end{align*}
	using \cref{eq:csa-denom-rc-treat}
	and that $T$ is invariant to $X$ cf. \Cref{ass:sample:mixture}.
	For the control weights:
	\begin{align*}
		E[ w^{c,rc}_{e}(X, C; p)  \mid  X]
		 & =
		E[ w^{c,rc}_{e, t}(X, C; p)  \mid  X]
		- E[ w^{c,rc}_{e, e - 1}(X, C; p) \mid  X]
		\\
		 & =
		\frac{1}{E[E_{e}] \cdot \lambda_{t}}
		E\left[\frac{C \cdot \mathbf{1}\{T = t\}p(X)}{1 - p(X)} \mid X\right]
		\\ & \quad
		- \frac{1}{E[E_{e}] \cdot \lambda_{e - 1}}
		E\left[\frac{C \cdot \mathbf{1}\{T = e - 1\}p(X)}{1 - p(X)} \mid X\right]
		\\
		 & =
		\frac{p(X)\lambda_{t}}{E[E_{e}] \cdot \lambda_{t}}
		- \frac{p(X)\lambda_{e - 1}}{E[E_{e}] \cdot \lambda_{e - 1}}
		\\
		 & = 0.
	\end{align*}
	Combining these two, we have
	\begin{align*}
		E [
		\{w^{trt,rc}_{e} - w^{c,rc}_{e}(X, C; p)\}
		\mu_{e, Y}^{c,rc}(X)
		]
		=
		E [
				E[w^{trt,rc}_{e} - w^{c,rc}_{e}(X, C; p) \mid X]
				\mu_{e, Y}^{c,rc}(X)
			]
		= 0.
	\end{align*}
	Thus we are only left with the first term in \cref{eq:rc-case3-terms}.
	Consider first:
	\begin{align*}
		E [ w^{trt,rc}_{e}Y ]
		=
		E [ w^{trt,rc}_{e, t}Y ]
		-
		E [ w^{trt,rc}_{e, e - 1}Y ]
	\end{align*}
	For the first term:
	\begin{align}
		\nonumber
		E [ w^{trt,rc}_{e, t}Y ]
		 & =
		E\left[
			\frac{E_{e} \cdot \mathbf{1}\{T = t\}}{
				E[E_{e} \cdot \mathbf{1}\{T = t\}]
			} Y
			\right]
		\\
		\nonumber
		 & =
		E\left[
			\frac{E_{e} \cdot \mathbf{1}\{T = t\}}{
				E[E_{e}] \cdot \lambda_{t}
			} E[Y \mid X, E_{e}, T]
			\right]
		\\
		\nonumber
		 & =
		E\left[
			\frac{E_{e} \cdot \mathbf{1}\{T = t\}}{
				E[E_{e}] \cdot \lambda_{t}
			} E[Y \mid X, E_{e} = 1, T = t]
			\right]
		\\
		\nonumber
		 & =
		E\left[
			\frac{E_{e}}{ E[E_{e}] } m_{e, t}^{trt, rc}(X)
			\right]
		\\
		\label{eq:rc-case3-1}
		 & =
		E\left[ m_{e, t}^{trt, rc}(X) \mid E_{e} = 1 \right],
	\end{align}
	where we in the second last equal sign again used the invariance assumption
	of \Cref{ass:sample:mixture}.
	With the exact same reasoning,
	\begin{align}
		\label{eq:rc-case3-2}
		E [ w^{trt,rc}_{e, e - 1}Y ]
		=
		E\left[ m_{e, e - 1}^{trt, rc}(X) \mid E_{e} = 1 \right].
	\end{align}
	Next, consider:
	\begin{align*}
		E [ w^{c,rc}_{e}(X, C; p)Y ]
		=
		E [ w^{c,rc}_{e, t}(X, C; p)Y ]
		-
		E [ w^{c,rc}_{e, e - 1}(X, C; p)Y ].
	\end{align*}
	Expanding out the definitions:
	\begin{align}
		\nonumber
		E [ w^{c,rc}_{e, t}(X, C; p)Y ]
		 & =
		E\left[
			\left(
			\frac{C \cdot \mathbf{1}\{T = t\}p(X)}{1 - p(X)}
			/
			E\left[
				\frac{C \cdot \mathbf{1}\{T = t\}p(X)}{1 - p(X)}
				\right]
			\right)
			Y
			\right]
		\\
		\nonumber
		 & =
		\frac{1}{E[E_{e}] \cdot \lambda_{t}}
		E\left[
			\frac{C \cdot \mathbf{1}\{T = t\}p(X)}{1 - p(X)}
			Y
			\right]
		\\
		\nonumber
		 & =
		\frac{1}{E[E_{e}] \cdot \lambda_{t}}
		E\left[
			\frac{C \cdot \mathbf{1}\{T = t\}p(X)}{1 - p(X)}
			E[Y \mid X, C, T]
			\right]
		\\
		\nonumber
		 & =
		\frac{1}{E[E_{e}] \cdot \lambda_{t}}
		E\left[
			\frac{E[C \mid X] \cdot E[\mathbf{1}\{T = t\} \mid X]p(X)}{1 - p(X)}
			E[Y \mid X, C = 1, T = t]
			\right]
		\\
		\nonumber
		 & =
		\frac{1}{E[E_{e}]}
		E\left[ p(X) m_{e,t}^{rc,c}(X) \right]
		\\
		\nonumber
		 & =
		\frac{1}{E[E_{e}]}
		E\left[ E_{e} m_{e,t}^{rc,c}(X) \right]
		\\
		\label{eq:rc-case3-3}
		 & =
		E\left[ m_{e,t}^{rc,c}(X) \mid E_{e} = 1 \right]
	\end{align}
	where we in the first equal sign used \cref{eq:csa-denom-rc} and the
	LIE repeatedly.
	With the exact same reasoning,
	\begin{align}
		\label{eq:rc-case3-4}
		E [ w^{c,rc}_{e, e - 1}(X, C; p)Y ]
		=
		E\left[ m_{e, e - 1}^{rc,c}(X) \mid E_{e} = 1 \right].
	\end{align}
	Combining \cref{eq:rc-case3-1,eq:rc-case3-2,eq:rc-case3-3,eq:rc-case3-4}
	we have
	\begin{align*}
		 & E [
				\{w^{trt,rc}_{e} - w^{c,rc}_{e}(X, C; p)\}
				\{Y - \mu_{e, Y}^{c,rc}(X)\}
				+ \kappa_{e, t}^{Y, rc}
			]
		\\
		 & =
		E[
		m_{e,t}^{trt,rc}(X)
		- m_{e,e-1}^{trt,rc}(X)
		-
		[
		m_{e,t}^{c,rc}(X)
		- m_{e,e-1}^{c,rc}(X)
		]
		\mid E_{e} = 1
		].
	\end{align*}
	The exact same calculations apply to the denominator of \cref{eq:dr-estimand-rc},
	yielding:
	\begin{align*}
		 & E [
				\{w^{trt,rc}_{e} - w^{c,rc}_{e}(X, C; p)\}
				\{D - \eta_{e, Y}^{c,rc}(X)\}
			] + \kappa_{e, t}^{D, rc}
		\\
		 & =
		E[
		g_{e,t}^{trt,rc}(X)
		- g_{e,e-1}^{trt,rc}(X)
		-
		[
		g_{e,t}^{c,rc}(X)
		- g_{e,e-1}^{c,rc}(X)
		]
		E_{e} = 1
		].
	\end{align*}
	Combining the above two expressions we have, using
	\cref{eq:identification-rc}, the result:
	\begin{align*}
		\frac{
		E [
		\{w^{trt,rc}_{e} - w^{c,rc}_{e}(X, C; p)\}
		\{Y - \mu_{e, Y}^{c,rc}(X)\}
		+ \kappa_{e, t}^{Y, rc}
		}{
		E [
				\{w^{trt,rc}_{e} - w^{c,rc}_{e}(X, C; p)\}
				\{D - \eta_{e, Y}^{c,rc}(X)\}
			] + \kappa_{e, t}^{D, rc}
		}
		= \tau^{rc}_{e, t}
		= LATT(e,t).
	\end{align*}

	\paragraph{Case 4: Repeated cross-sections and outcome regressions correctly
		specified }
	In this case:
	\begin{align*}
		 & \mu_{e,t}^{rc,trt}(X) = m_{e, t}^{trt, rc}(X) \, \text{a.s.},
		\quad
		\mu_{e, s, t}^{rc,c}(X) = m_{e, s, t}^{rc,c}(X) \, \text{a.s.},
		\\
		 & \eta_{e,t}^{rc,trt}(X) = g_{e, t}^{trt, rc}(X) \, \text{a.s.},
		\quad
		\eta_{e, s, t}^{rc,c}(X) = g_{e, s, t}^{rc,c}(X) \, \text{a.s.}.
	\end{align*}
	We consider again the numerator of \cref{eq:dr-estimand-rc}.
	In this case, note that $\kappa_{e, t}^{Y, rc}$ consists of the misspecified
	outcome regressions.
	However, the misspecified functions are also just measurable functions of
	$X$,
	and thus \cref{eq:kappa-rc} still holds.
	Hence,
	\begin{align}
		\nonumber
		 & E [
				\{w^{trt,rc}_{e} - w^{c,rc}_{e}(C, X; \pi)\}
				\{Y - \mu_{e, Y}^{c,rc}(X)\}
				+ \kappa_{e, t}^{Y, rc}
			]
		\\
		\nonumber
		 & = E [
				\{w^{trt,rc}_{e} - w^{c,rc}_{e}(C, X; \pi)\}
				\{Y - m_{e, Y}^{c,rc}(X)\}
			]
		\\
		\nonumber
		 & = E [
				w^{trt,rc}_{e}
				\{Y - m_{e, Y}^{c,rc}(X)\}
			]
		- E[
				w^{c,rc}_{e}(C, X; \pi)
				\{Y - m_{e, Y}^{c,rc}(X)\}
			]
	\end{align}
	By a consequence of the time indicators in the weights,
	$w^{c,rc}_{e}(C, X; \pi)$,
	we have, using the LIE:
	\begin{align*}
		 & E[ w^{c,rc}_{e}(C, X; \pi) \{Y - m_{e, Y}^{c,rc}(X)\} ]
		\\
		 & =  E[ w^{c,rc}_{e}(C, X; \pi) \{E[Y \mid X, C, T] - m_{e, Y}^{c,rc}(X)\} ]
		\\
		 & =  E[ w^{c,rc}_{e}(C, X; \pi) \{m_{e, Y}^{c,rc}(X) - m_{e, Y}^{c,rc}(X)\} ]
		\\
		 & = 0
	\end{align*}
	where the second last equality follows by constructing the mean function
	\cref{eq:mean-function-m-period-rc-control}
	under the expectation operator and multiplied by the time period indicators.

	Next,
	\begin{align*}
		E [
		w^{trt,rc}_{e}
		\{Y - m_{e, Y}^{c,rc}(X)\}
		]
		 & =
		E [ w^{trt,rc}_{e} Y ]
		-
		E [ w^{trt,rc}_{e} m_{e, Y}^{c,rc}(X) ]
	\end{align*}
	For the first term in the expression above,
	\cref{eq:rc-case3-1,eq:rc-case3-2} gives:
	\begin{align}
		\label{eq:rc-case4-1}
		E [ w^{trt,rc}_{e} Y ]
		=
		E[ m_{e, t}^{trt, rc}(X) - m_{e, e - 1}^{trt, rc}(X) \mid E_{e} = 1 ].
	\end{align}
	For the second term, we can expand the definitions as:
	\begin{align}
		\nonumber
		          & E [ w^{trt,rc}_{e} m_{e, Y}^{c,rc}(X) ]
		\\
		\nonumber
		          & =
		E \Big[
			\Big(
			\frac{E_{e} \cdot \mathbf{1}\{T = t\}}{ E[E_{e}] \cdot \lambda_{t} }
			- \frac{E_{e} \cdot \mathbf{1}\{T = e - 1\}}{ E[E_{e}] \cdot \lambda_{e - 1}
			}
			\Big)
		\\
		\nonumber & \quad
			\times
			(\mathbf{1}\{T = t\} \cdot m_{e, t, t}^{c, rc}(X)
			+ \mathbf{1}\{T = e-1\} \cdot m_{e, t, e-1}^{c, rc}(X))
			\Big]
		\\
		\nonumber
		          & =
		E \Big[
			\frac{E_{e} \cdot \mathbf{1}\{T = t\}}{ E[E_{e}] \cdot \lambda_{t} }
			\cdot m_{e, t, t}^{c, rc}(X)
			- \frac{E_{e} \cdot \mathbf{1}\{T = e - 1\}}{ E[E_{e}] \cdot \lambda_{e - 1}}
			m_{e, t, e-1}^{c, rc}(X))
			\Big]
		\\
		\nonumber
		          & =
		E \Big[
			\frac{E[E_{e} \cdot \mathbf{1}\{T = t\} \mid X]}{ E[E_{e}] \cdot \lambda_{t} }
			\cdot m_{e, t, t}^{c, rc}(X)
			- \frac{E[E_{e} \cdot \mathbf{1}\{T = e - 1\} \mid X]}
			{ E[E_{e}] \cdot \lambda_{e - 1}}
			m_{e, t, e-1}^{c, rc}(X))
			\Big]
		\\
		\nonumber
		          & =
		E \Big[
			\frac{
				p(X)\{m_{e, t, t}^{c, rc}(X) - m_{e, t, e-1}^{c, rc}(X))\}
			}{ E[E_{e}]}
			\Big]
		\\
		\nonumber
		          & =
		E \Big[
			\frac{
				p(X)\{m_{e, t, t}^{c, rc}(X) - m_{e, t, e-1}^{c, rc}(X))\}
			}{ E[E_{e}]}
			\Big]
		\\
		\label{eq:rc-case4-2}
		          & =
		E[ m_{e,t}^{c,rc}(X) - m_{e,e-1}^{c,rc}(X) \mid E_{e} = 1 ].
	\end{align}
	where we in the last equal sign went backwards from a LIE application and
	used the definition of the conditional expectation.
	Thus, by \cref{eq:rc-case4-1,eq:rc-case4-2}:
	\begin{align}
		\nonumber
		 & E [
				\{w^{trt,rc}_{e} - w^{c,rc}_{e}(C, X; \pi)\}
				\{Y - \mu_{e, Y}^{c,rc}(X)\}
				+ \kappa_{e, t}^{Y, rc}
			]
		\\
		\label{eq:rc-case-4-fin1}
		 & =
		E[
		m_{e, t}^{trt, rc}(X) - m_{e, e - 1}^{trt, rc}(X)
		-
		[m_{e,t}^{c,rc}(X) - m_{e,e-1}^{c,rc}(X)] \mid E_{e} = 1
		].
	\end{align}

	Again, the exact same calculations apply to the denominator of
	\cref{eq:dr-estimand-rc}, yielding:
	\begin{align}
		\nonumber
		 & E [
				\{w^{trt,rc}_{e} - w^{c,rc}_{e}(C, X; \pi)\}
				\{D - \eta_{e, Y}^{c,rc}(X)\}
				+ \kappa_{e, t}^{D, rc}
			]
		\\
		\label{eq:rc-case-4-fin2}
		 & =
		E[
		g_{e, t}^{trt, rc}(X) - g_{e, e - 1}^{trt, rc}(X)
		-
		[g_{e,t}^{c,rc}(X) - g_{e,e-1}^{c,rc}(X)] \mid E_{e} = 1
		].
	\end{align}
	Combining \cref{eq:rc-case-4-fin1,eq:rc-case-4-fin2}, and again using
	\cref{eq:identification-rc}, we have the double robustness result:
	\begin{align*}
		\frac{
			E [
					\{w^{trt,rc}_{e} - w^{c,rc}_{e}(C, X; \pi)\}
					\{Y - \mu_{e, Y}^{c,rc}(X)\}
					+ \kappa_{e, t}^{Y, rc}
				]
		}{
			E [
					\{w^{trt,rc}_{e} - w^{c,rc}_{e}(C, X; \pi)\}
					\{D - \eta_{e, Y}^{c,rc}(X)\}
					+ \kappa_{e, t}^{D, rc}
				]
		}
		= \tau^{rc}_{e, t}
		= LATT(e,t).
	\end{align*}
\end{proofref}

\medskip\noindent

\begin{proofref}{prop:bloom}
	First note that, as $E_{e} + C = 1$, the event $(C = 0)$ is equivalent to
	$(E_{e} = 1)$. Thus, taking the contrapositive of
	\cref{eq:assumption-bloom} yields:
	\begin{align}
		\label{eq:assumption-bloom-cp}
		D_{t} = 1 \implies E_{e} = 1.
	\end{align}
	Consider the numerator \cref{eq:numerator-proof-panel}.
	The control term equals:
	\begin{align*}
		E[Y_{t} - Y_{e - 1} \mid X, C = 1]
		=
		E[Y_{t}(0) - Y_{e - 1}(0) \mid X, C = 1]
	\end{align*}
	using \cref{eq:observed-outcome,eq:assumption-bloom}.
	Using \cref{eq:observed-outcome} again,
	the treated term equals:
	\begin{align*}
		 & E[Y_{t} - Y_{e - 1} \mid X, E_{e} = 1]
		\\
		 & =
		E\Big[
		Y_{t}(0) - Y_{e-1}(0)
		+ [Y_{t}(1) - Y_{t}(0)]D_{t}
		+ [Y_{e-1}(1) - Y_{e-1}(0)]D_{e-1}
		\mid X, E_{e} = 1 \Big]
		\\
		 & =
		E[ Y_{t}(0) - Y_{e-1}(0) \mid X, E_{e} = 1 ]
		+ E[ \tau_{t} D_{t} \mid X, E_{e} = 1 ]
	\end{align*}
	where we used that $E_{e} = 1$ implies $D_{e - 1} = 0$, since the unit has
	not yet been exposed in period $e - 1$, and defined $\tau_{t} := Y_{t}(1) -
		Y_{t}(0)$.
	By the exclusion restriction, \Cref{ass:no-carryover},
	the instrument is independent of the potential outcomes,
	hence
	\begin{align*}
		E[ Y_{t}(0) - Y_{e-1}(0) \mid X, E_{e} = 1 ]
		- E[Y_{t}(0) - Y_{e - 1}(0) \mid X, C = 1]
		= 0.
	\end{align*}
	Taking these expressions together, the numerator
	\cref{eq:numerator-proof-panel} equals:
	\begin{align*}
		 & E[Y_{t} - Y_{e - 1} \mid X, E_{e} = 1]
		- E[Y_{t} - Y_{e - 1} \mid X, C = 1]
		\\
		 & = E[ \tau_{t} D_{t} \mid X, E_{e} = 1 ]
	\end{align*}
	Similarly, for the denominator \cref{eq:denominator-proof-panel},
	\begin{align*}
		 & E[D_{t} - D_{e - 1} \mid X, E_{e} = 1]
		- E[D_{t} - D_{e - 1} \mid X, C = 1]
		\\
		 & = P(D_{t} = 1 \mid X, E_{e} = 1)
	\end{align*}
	by the same argument.

	Applying the LIE to the numerator yields:
	\begin{align}
		\nonumber
		 & E\left[E[ \tau_{t} D_{t} \mid X, E_{e} = 1 ] \mid E_{e} = 1\right]
		\\
		\nonumber
		 & = E[ \tau_{t} D_{t} \mid E_{e} = 1 ]
		\\
		\label{eq:bloom-num}
		 & =  E[ \tau_{t} \mid D_{t} = 1, E_{e} = 1 ]P(D_{t} = 1 \mid E_{e} = 1)
	\end{align}
	Likewise for the denominator,
	\begin{align}
		\label{eq:bloom-den}
		E[P(D_{t} = 1 \mid X, E_{e} = 1) \mid E_{e} = 1]
		= P(D_{t} = 1 \mid E_{e} = 1).
	\end{align}
	Hence, dividing \cref{eq:bloom-num} by \cref{eq:bloom-den} yields:
	\begin{align}
		\label{eq:bloom-id}
		\frac{
			E[ \tau_{t} \mid D_{t} = 1, E_{e} = 1 ]P(D_{t} = 1 \mid E_{e} = 1)
		}{P(D_{t} = 1 \mid E_{e} = 1).}
		= E[ \tau_{t} \mid D_{t} = 1, E_{e} = 1 ]
	\end{align}
	That treatment is absorbing implies that units can be uniquely partitioned
	into the treatment-cohort groups $G$, \cref{eq:csa-objects}.
	The event $(D_{t} = 1)$ consists of all individuals treated at time
	$t$, which by the previous point can be written as a union over the cohort
	groups, $(D_{t} = 1) = \cup_{g \leq t} (G_{g} = 1)$.
	Moreover, we can write the observed outcome in terms of the potential
	outcomes and treatment cohort indicator variables cf. \cite{csa} as
	$
		Y_{t} = Y_{ t}(0) + \sum_{g = 2}^{\mathcal{T}}
		[Y_{t}(g) - Y_{t}(0)]G_{g}.
	$
	In this case, the individual time $t$ treatment effect can be written as
	$\tau_{t} = Y_{t}(g) - Y_{t}(0)$.
	Hence, we can further decompose \cref{eq:bloom-id} as:
	\begin{align*}
		E[ \tau_{t} \mid D_{t} = 1, E_{e} = 1 ]
		 & = E[\tau_{t} \mid \cup_{g \leq t} (G_{g} = 1), E_{e} = 1]
		\\
		 & = \sum_{g \leq t}
		E[\tau_{t} \mid G_{g} = 1, \cup_{g \leq t} (G_{g} = 1), E_{e} = 1]
		P(G_{g} = 1 \mid \cup_{g \leq t} (G_{g} = 1), E_{e} = 1)
		\\
		 & =  \sum_{g \leq t}
		E[Y_{t}(g) - Y_{t}(0) \mid G_{g} = 1, E_{e} = 1]P(G_{g} = 1 \mid D_{t} = 1,
		E_{e} = 1)
		\\
		 & =  \sum_{g \leq t}  ATT(g, t, e)P(G_{g} = 1 \mid D_{t} = 1, E_{e} = 1)
	\end{align*}
	proving \cref{eq:bloom} and the proposition.
\end{proofref}

\newpage
\section{Influence Functions}
\label{sec:influence-functions}
In the following, we derive the DR estimands and corresponding EIFs in
both data settings.

As an overview of the construction, \Cref{tab:roadmap} summarizes the main
steps. In both sampling schemes, we first express the identified LATT as a
ratio of two simple ATT parameters; see
\Cref{rem:relabeling-idid}. We then
derive the LATT EIF by applying the quotient rule
(\cref{eq:if-quotient-rule-estimand}; to be defined),
and use this representation to construct the corresponding doubly robust
estimand and its EIF. \Cref{fig:roadmap} illustrates this
schematically\footnote{
	See also \Cref{table:comparison} for a comparison between DiD
	(two-period and staggered adoption) and the more general IDiD framework
	considered in this paper.
}.

\begin{remark}[Relabeling DiD objects for IDiD]
	\label{rem:relabeling-idid}
	As shown in \Cref{sec:identification},
	identification yields initial LATT estimands that are ratios of two ATT-type
	parameters, \cref{eq:identification-panel,eq:identification-rc}.
	We then derive EIFs for general ATT-type parameters corresponding to DiD with
	panel data and repeated cross-sections,
	and with the tools developed \cref{sec:eif-derivation}, use these DiD objects
	together with the ratio structure to get corresponding IDiD objects.

	Specifically, we apply the DiD formulas to each relevant $2 \times 2$
	comparison, replacing the DiD group indicator with the exposed-cohort
	indicator $E_e$ and the control indicator with the corresponding
	control variable $C$. In the panel case, the outcome is replaced by
	$\Delta_{t-e+1}Y_t$ in the numerator and by $\Delta_{t-e+1}D_t$ in the
	denominator. In the repeated cross-sections case, the same replacement
	is applied separately to the treated and control mean components in
	periods $t$ and $e-1$, using $Y$ in the numerator and $D$ in the
	denominator.
\end{remark}

\vspace{-5pt}
\begin{table}[H]
	\centering
	\small
	\caption{Roadmap of identification, influence functions, and doubly robust estimands}
	\label{tab:roadmap}
	\begin{tabular}{lll}
		\toprule
		Step & Panel data                                          & Repeated cross-sections \\
		\midrule
		DiD estimand and EIF
		     & \cref{eq:tau-drdid-estimand,eq:if-att-did-reduced}
		     & \cref{eq:estimand-rc-drdid,eq:if-estimand-rc-drdid}                           \\
		LATT identification
		     & \cref{eq:identification-panel}
		     & \cref{eq:identification-rc}                                                   \\

		LATT EIF
		     & \cref{eq:if-latt-panel}
		     & \cref{eq:if-latt-rc}                                                          \\

		DR estimand and EIF
		     & \cref{eq:dr-estimand-panel,eq:if-dr-estimand-panel}
		     & \cref{eq:dr-estimand-rc,eq:if-dr-estimand-rc}                                 \\
		\bottomrule
	\end{tabular}
\end{table}

\begin{figure}[H]
	\centering
	\begin{tikzpicture}[
		node distance=1.2cm and 1.5cm,
		box/.style={
				draw,
				rounded corners,
				align=center,
				inner sep=3pt,
				text width=2.8cm,
				font=\footnotesize
			},
		arrow/.style={
		-{Latex[length=1.4mm,width=0.9mm]},
		line width=0.5pt
		}
		]

		\node[box] (simplep) {DiD estimand + EIF\\
			\cref{eq:tau-drdid-estimand,eq:if-att-did-reduced}};

		\node[
			font=\bfseries\footnotesize,
			above=0.1cm of simplep,
			anchor=south
		] (panel_label) {Panel data};

		\node[box, right=of simplep] (idp) {LATT identification\\
			\cref{eq:identification-panel}};

		\node[box, right=of idp] (ifp) {LATT EIF\\
			\cref{eq:if-latt-panel}};

		\node[box, right=of ifp] (drp) {DR estimand + EIF\\
			\cref{eq:dr-estimand-panel,eq:if-dr-estimand-panel}};

		\node[
			font=\bfseries\footnotesize,
			below=1cm of simplep,
			anchor=south
		] (rc_label) {Repeated cross-sections};

		\node[box, below=0.1cm of rc_label] (simplerc) {DiD estimand + EIF\\
			\cref{eq:estimand-rc-drdid,eq:if-estimand-rc-drdid}};

		\node[box, right=of simplerc] (idrc) {LATT identification\\
			\cref{eq:identification-rc}};

		\node[box, right=of idrc] (ifrc) {LATT EIF\\
			\cref{eq:if-latt-rc}};

		\node[box, right=of ifrc] (drrc) {DR estimand + EIF\\
			\cref{eq:dr-estimand-rc,eq:if-dr-estimand-rc}};

		\draw[arrow]
		(simplep.east) to[out=40, in=160]
		node[midway, above=5pt, align=center, font=\scriptsize]
		{ATT ratio + relabeling + normalization}
		(ifp.north);

		\draw[arrow]
		(idp) -- node[midway, above=3pt, align=center, font=\scriptsize]
		{quotient\\rule} (ifp);

		\draw[arrow]
		(ifp) -- node[midway, above=5pt, align=center, font=\scriptsize]
		{EE\\and EIF \\ derivation} (drp);

		\draw[arrow]
		(simplerc.east) to[out=-40, in=-160]
		node[midway, below=5pt, align=center, font=\scriptsize]
		{ATT ratio + relabeling + normalization}
		(ifrc.south);

		\draw[arrow]
		(idrc) -- node[midway, below=3pt, align=center, font=\scriptsize]
		{quotient\\rule} (ifrc);

		\draw[arrow]
		(ifrc) -- node[midway, below=5pt, align=center, font=\scriptsize]
		{EE\\and EIF \\ derivation} (drrc);

	\end{tikzpicture}

	\caption{Roadmap from simple DiD objects to DR estimands and EIF for the LATT.
		The ``ATT ratio + relabeling'' step is described in
		\Cref{rem:relabeling-idid}. The "+ normalization" step refers to normalizing the
		control weights after relabeling.
		The ``EE and EIF derivation'' step
		means we derive the DR estimand using the estimation equation approach
		shown in \cref{eq:dr-estimand-foc}
		and then derive the DR estimand's EIF.
	}
	\label{fig:roadmap}
\end{figure}
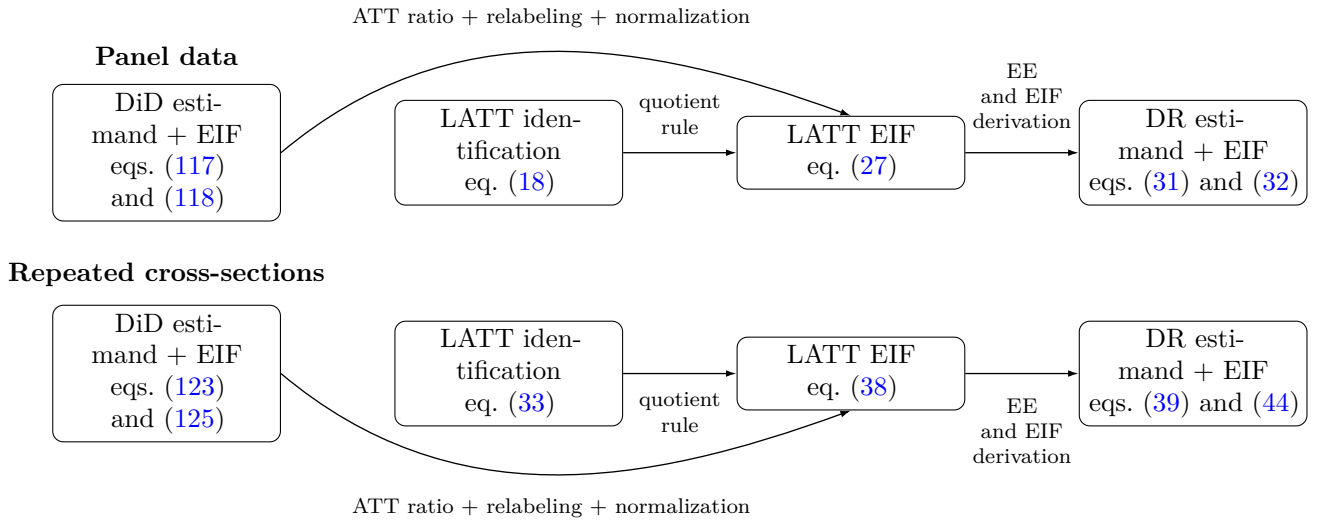

\begin{table}[H]
	\centering
	\small
	\setlength{\tabcolsep}{4pt}
	\resizebox{\textwidth}{!}{%
		\begin{threeparttable}
			\begin{tabular}{lccc}
				\toprule
				                 & Two-period DiD
				                 & CSA (DiD)
				                 & IDiD
				\\
				\midrule
				Timing structure
				                 & Single treatment
				                 & Staggered treatment
				                 & Staggered exposure
				\\

				Group variable
				                 & $D$
				                 & $G := \min \{t \mid D_{t} = 1\} $;
				$G_{g} = \mathbf{1}\{G = g\}$
				                 & $E := \min \{t \mid Z_{t} = 1\}$;
				$E_{e} = \mathbf{1}\{E = e\}$
				\\

				Target parameter
				                 & $ATT = E[Y_t(1)-Y_t(0)\mid D=1]$
				                 & $ATT(g,t) = E[Y_t(g)-Y_t(0)\mid G_g=1]$
				                 & $LATT(e,t) = E[Y_t(1)-Y_t(0)\mid E_e=1,\, D_t(e)>D_t(\infty)]$ \\
				Control group(s) &
				$D = 0$
				                 &
				$C := \mathbf{1}\{G = \infty\} = 1$
				or
				$(1 - D_{s})(1 - G_{g}) = 1$
				                 &
				$C^{nev} := \mathbf{1}\{E = \infty\} = 1$
				or
				$C^{nye}_{e, s} = 1$
				\\

				Propensity score(s)
				                 & $p(X) = P(D = 1 \mid X)$
				                 &
				\makecell[c]{
				$p_g(X) = P(G_g = 1 \mid X, G_g + C = 1)$                                         \\
					or $p_{g,s}(X) = P(G_g = 1 \mid X, G_g + (1-D_s)(1-G_g) = 1)$
				}
				                 &
				\makecell[c]{
				$p_{e}(X) = P(E_{e} = 1 \mid X, E_{e} + C^{nev} = 1)$                             \\
					or $p_{e, s}(X) := P(E_{e} = 1 \mid X, E_{e} + C^{nye}_{e, s} = 1)$
				}
				\\

				\bottomrule
			\end{tabular}

			\begin{tablenotes}
				\footnotesize
				\item \textit{Notes}: The table compares the treatment definition,
				causal estimand, and propensity score(s) used in
				classical two-period DiD e.g. \cite{sazhao},
				DiD with staggered adoption \cite{csa},
				and general IDiD with staggered exposure and covariates.
				The control variables for the IDiD setup are given in
				\cref{eq:control-variables}.
			\end{tablenotes}

		\end{threeparttable}
	}
	\caption{Comparison of treatment indicators, estimands, and propensity scores across designs.}
	\label{table:comparison}
\end{table}

\subsection{EIF derivation strategy}
\label{sec:eif-derivation}
Instead of using the tangent space approach as in \cite{sazhao} (which follows
the tradition in the econometrics literature cf. \cite{newey1990}),
we follow the approach of \cite{kennedy} to derive the EIFs in a
nonparametric model.
Specifically, we use Kennedy's ``Strategy 2'', which consists of:
\begin{enumerate}
	\item Computing Gateaux derivatives assuming data are discrete
	\item Using derivative rules with simple influence functions as building
	      blocks
\end{enumerate}
Below, we state the strategy for completeness.

\paragraph{Strategy 2 \cite{kennedy}}
Introduce the operator
\[
	\mathbb{IF} : \Psi \to L^2(P),
\]
mapping functionals $\psi : \mathcal{P} \to \mathbb{R}$ (i.e. estimands) to
their influence functions
$\varphi \in L^2(P)$ in a nonparametric model.
The strategy to derive efficient influence functions in a nonparametric model
can then be stated as follows:
\begin{tricks}
	\item \label{trick:1} Pretend the data are discrete
	\item \label{trick:2} Treat influence functions as derivatives, allowing
	use of differentiation rules
	\begin{tricks}
		\item \label{trick:2a}
		\emph{(Product rule)}
		\begin{equation}
			\label{eq:if-product-rule}
			\mathbb{IF}(\psi_1 \psi_2)
			= \mathbb{IF}(\psi_1)\,\psi_2 + \psi_1\,\mathbb{IF}(\psi_2).
		\end{equation}

		\item \label{trick:2b}
		\emph{(Chain rule)}
		\[
			\mathbb{IF}(f(\psi)) = f'(\psi)\,\mathbb{IF}(\psi).
		\]
	\end{tricks}
	\item \label{trick:3} Use \textit{influence function building blocks}.
	E.g.
	that the influence function of the density $p(x)$ equals:
	\begin{align}
		\label{eq:block-density}
		\mathbb{IF}(p(x)) = \mathbf{1}\{X = x\} - p(x),
	\end{align}
	the influence function of the expectation $E[Y]$
	equals
	\begin{align}
		\label{eq:block-mean}
		\mathbb{IF}(E[Y]) = Y - E[Y],
	\end{align}
	and the influence function of the conditional expectation $E[Y
				\mid X = x]$ equals:
	\begin{align}
		\label{eq:block-cexp}
		\mathbb{IF}(E[Y \mid X = x])
		= \frac{\mathbf{1}\{X = x\}}{P(X = x)}\{Y - E[Y \mid X = x]\}.
	\end{align}
\end{tricks}

\subsection{Building Blocks}
In this section, we derive some more building blocks that we will use to arrive
at the main influence functions for the LATT parameter in the panel-data and
repeated-cross-sections settings.

\paragraph{Influence Function Probability}
Let $D \in \{0, 1\}$ be binary.
Using the building block \cref{eq:block-mean}
we derive the influence function of the probability $P(D = 1)$ as:
\begin{align}
	\label{eq:if-block-prob}
	\mathbb{IF}\left(P(D = 1)\right)
	=
	\mathbb{IF}\left(E[\mathbf{1}\{D = 1\}]\right)
	=
	\mathbf{1}\{D = 1\} - P(D = 1).
\end{align}
where we used that $D$ is binary.
In a similar fashion,
using the building block \cref{eq:block-cexp},
we derive the influence function of the conditional probability $P(D = 1 \mid X
	= x)$
as:
\begin{align}
	\nonumber
	\mathbb{IF}(P(D = 1 \mid X = x))
	 & = \mathbb{IF}(E(D\mid X = x))
	\\
	\label{eq:if-block-cprob}
	 & = \frac{\mathbf{1}\{X = x\}}{P(X = x)}\{D - P(D = 1 \mid X = x)\}.
\end{align}

\paragraph{Quotient Rule for Influence Functions}
We can also derive the influence function of an estimand that is a ratio of two
parameters:
\begin{align*}
	\psi^{ratio} =
	\frac{\psi^{num}}{\psi^{den}}
\end{align*}
which is done as follows:
\begin{align}
	\nonumber
	\mathbb{IF}(\psi^{ratio})
	 & =
	\mathbb{IF}\left(\frac{\psi^{num}}{\psi^{den}}\right)
	\\
	\nonumber
	 & =
	\frac{\mathbb{IF}(\psi^{num})}{\psi^{den}}
	+ \psi^{num}\mathbb{IF}\left(\frac{1}{\psi^{den}}\right)
	\\
	\nonumber
	 & =
	\frac{\mathbb{IF}\left(\psi^{num}\right)}{\psi^{den}}
	- \psi^{num}\frac{1}{(\psi^{den})^{2}}\mathbb{IF}(\psi^{den})
	\\
	\nonumber
	 & =
	\frac{\mathbb{IF}(\psi^{num})}{\psi^{den}}
	- \frac{\psi^{num}}{\psi^{den}}
	\frac{\mathbb{IF}(\psi^{den})}{\psi^{den}}
	\\
	 & =
	\label{eq:if-quotient-rule-estimand}
	\frac{1}{\psi^{den}}
	\left[
		\mathbb{IF}(\psi^{num})
		- \psi^{ratio}
		\mathbb{IF}(\psi^{den})
		\right],
\end{align}
where
the second equality follows by \Cref{trick:2a},
the third equality follows by \Cref{trick:2b},
and the rest by simple algebra.

\paragraph{Ratio Estimand Linear IF}
Often we can write the influence functions for the numerator and denominator as
linear in their respective target parameter weighted by some (possibly
data and nuisance-function-dependent) weight $w(O; \eta)$:
\begin{equation}
	\label{eq:if-ratio-components-linear}
	\mathbb{IF}(\psi^{num})
	=
	\phi^{num}(\cdot; \eta^{num}) - w(\cdot; \eta) \psi^{num},
	\quad
	\mathbb{IF}(\psi^{den})
	=
	\phi^{den}(\cdot; \eta^{den}) - w(\cdot; \eta) \psi^{den}.
\end{equation}
In this case, \cref{eq:if-quotient-rule-estimand} equals:
\begin{align}
	\label{eq:if-quotient-rule-estimand-linear-if}
	\mathbb{IF}(\psi^{ratio})
	=
	\frac{1}{\psi^{den}}
	\left[
		\phi^{num}(\cdot; \eta^{num})
		- \psi^{ratio}
		\phi^{den}(\cdot; \eta^{den})
		\right]
\end{align}
as
$
	- \frac{w(\eta^{num}) \psi^{num}}{\psi^{den}}
	+ \frac{1}{\psi^{den}}
	\frac{\psi^{num}}{\psi^{den}}w(\eta) \psi^{den}
	= 0.
$

\paragraph{Influence Function Conditional Density}
Let again $D \in {0, 1}$ be binary. We derive the influence function of the
conditional density $p_{X \mid D}(x \mid d)$ using the building blocks
established above.
The influence functions of the joint and marginal density equal:
\begin{align*}
	\mathbb{IF}(p_{D}(1))       & = \mathbf{1}\{D = 1\} - p_{D}(1)
	\\
	\mathbb{IF}(p_{X, D}(x, 1)) & = \mathbf{1}\{X = x,D = 1\} - p_{X, D}(x, 1).
\end{align*}
With the above and the quotient rule \cref{eq:if-quotient-rule-estimand},
we derive the influence function as:
\begin{align}
	\nonumber
	 & \mathbb{IF}(p_{X \mid D}(x \mid 1))
	\\
	\nonumber
	 & =
	\mathbb{IF}\left(\frac{p_{X, D}(x, 1)}{p_{D}(1)}\right)
	\\
	\nonumber
	 & =
	\frac{1}{p_{D}(1)}
	\left[
	\mathbb{IF}(p_{X, D}(x, 1))
	- \frac{p_{X, D}(x, 1)}{p_{D}(1)}
	\mathbb{IF}(p_{D}(1))
	\right]
	\\
	\nonumber
	 & =
	\frac{1}{p_{D}(1)}
	\left[
		\{
		\mathbf{1}\{X = x,D = 1\} - p_{X, D}(x, 1)
		\}
		- \frac{p_{X, D}(x, 1)}{p_{D}(1)}
		\{
		\mathbf{1}\{D = 1\} - p_{D}(1)
		\}
		\right]
	\\
	\nonumber
	 & =
	\frac{1}{p_{D}(1)}
	\left[
		\{
		D \mathbf{1}\{X = x\} - p_{X, D}(x, 1)
		\}
		-
		\left\{
		D \frac{p_{X, D}(x, 1)}{p_{D}(1)}
		- p_{X, D}(x, 1)
		\right\}
		\right]
	\\
	\nonumber
	 & =
	\frac{D}{p_{D}(1)}
	\left[
		\mathbf{1}\{X = x\}
		- \frac{p_{X, D}(x, 1)}{p_{D}(1)}
		\right]
	\\
	\label{eq:if-cdensity}
	 & =
	\frac{D}{p_{D}(1)}
	\left[
		\mathbf{1}\{X = x\}
		- p_{X \mid D}(x \mid 1)
		\right].
\end{align}

\subsection{ATT parameter}
\subsubsection{Panel case}
\label{sec:panel-if-derivation}
\begin{proposition}
	\label{prop:att-panel}
	Let $D \in {0, 1}$ be binary,
	$X$ some vector of covariates
	and $\Delta Y := Y_{1} - Y_{0}$ the differenced outcome variable.
	Define the estimand of interest to be the average treatment effect on the
	treated when panel data is available:
	\begin{align}
		\label{eq:tau-drdid-estimand}
		\tau^{p} & = E[
				E[\Delta Y \mid D = 1, X]
				-
				E[\Delta Y \mid D = 0, X]
				\mid D = 1
			]
		\\
		\nonumber
		         & =
		E[ \mu(1, X) - \mu(0, X) \mid D = 1 ]
	\end{align}
	for $\mu(d, x) := E[\Delta Y \mid D = d, X = x]$\footnote{
		This estimable parameter identifies the ATT in the two-period DiD setup
		assuming no anticipation and parallel trends.
	}.
	Then the influence function of $\tau^{p}$ equals:
	\begin{align}
		\label{eq:if-att-did-reduced}
		\varphi(O; \tau^{p}, \eta^{p})
		=
		\left\{ w^{p}_{1}(D) - w^{p}_{0}(D, X; p) \right\}
		\{\Delta Y - \mu(0, X)\}
		- \frac{D}{\rho}\tau^{p}
	\end{align}
	where $\rho := P(D = 1)$,
	$p(X) := P(D = 1 \mid X)$
	and
	\begin{align}
		\label{eq:weights-panel}
		w^{p}_{1}(D) := \frac{D}{\rho},
		\quad
		w^{p}_{0}(D, X; p)
		:=
		\frac{(1 - D)}{\rho}\frac{p(X)}{1 - p(X)}.
	\end{align}
\end{proposition}
\begin{remark}
	An equivalent representation of the influence function
	\cref{eq:if-att-did-reduced} follows from
	\begin{align*}
		w^{p}_{1}(D) - w^{p}_{0}(D, X; p)
		=
		\frac{D[1 - p(X)] - (1 - D)p(X)}{\rho[1 - p(X)]}
		= \frac{D - p(X)}{\rho[1 - p(X)]}
	\end{align*}
	such that
	\begin{align*}
		\varphi(O; \tau^{p}, \eta^{p})
		=
		\frac{D - p(X)}{\rho[1 - p(X)]}
		\{\Delta Y - \mu(0, X)\}
		- \frac{D}{\rho}\tau^{p}.
	\end{align*}
	We retain the formulation in \cref{eq:if-att-did-reduced}, as the distinction
	between treated and control weights is more convenient for extensions to the
	$LATT$ estimands.
\end{remark}
\begin{proof}
	Split up the estimand in two parts:
	\begin{align*}
		\tau^{p} & = \psi_{1} - \psi_{2},       \\
		\psi_{1} & := E[ \mu(1, X) \mid D = 1 ]
		=
		\int \mu(1, x) dP(X \in dx \mid D = 1),
		\\
		\psi_{2} & := E[ \mu(0, X) \mid D = 1 ]
		=
		\int \mu(0, x) dP(X \in dx \mid D = 1).
	\end{align*}
	\paragraph{Influence function of $\psi_{1}$.}
	We proceed as follows:
	\begin{align}
		\nonumber
		 & \mathbb{IF}(\psi_{1})
		\\
		\nonumber
		 & = \mathbb{IF}\left(\sum_{x}  \mu(1, x)p_{X \mid D}(x \mid 1)\right)
		\\
		\nonumber
		 & =
		\sum_{x} \left[\mathbb{IF}\{\mu(1, x)\}p_{X \mid D}(x \mid 1)
		+ \mu(1, x)\mathbb{IF}\{p_{X \mid D}(x \mid 1)\}\right]
		\\
		\nonumber
		 & =
		\sum_{x}
		\left(\frac{\mathbf{1}\{X = x, D = 1\}}{p_{X, D}(x, 1)}\{\Delta Y - \mu(1, x)\}p_{X \mid D}(x \mid 1)
		+ \mu(1, x)\left[\frac{D}{p_{D}(1)}
			\left\{
			\mathbf{1}\{X = x\}
			- p_{X \mid D}(x \mid 1)
			\right\}\right]\right)
		\\
		\nonumber
		 & =
		\sum_{x}
		\left(D \frac{\mathbf{1}\{X = x\}}{p_{D}(1)}\{\Delta Y - \mu(1, x)\}
		+ \mu(1, x)\left[\frac{D}{p_{D}(1)}
				\left\{
				\mathbf{1}\{X = x\}
				- p_{X \mid D}(x \mid 1)
				\right\}\right]\right)
		\\
		\nonumber
		 & =
		\sum_{x}
		D \frac{\mathbf{1}\{X = x\}}{p_{D}(1)}\{\Delta Y - \mu(1, x)\}
		+
		\sum_{x}
		\mathbf{1}\{X = x\}
		\frac{D}{p_{D}(1)}
		\mu(1, x)
		-
		\sum_{x}
		p_{X \mid D}(x \mid 1)
		\frac{D}{p_{D}(1)}
		\mu(1, x)
		\\
		\nonumber
		 & =
		\frac{D}{p_{D}(1)}\{\Delta Y - \mu(1, X)\}
		+
		\frac{D}{p_{D}(1)}
		\mu(1, X)
		-
		\sum_{x}
		\frac{D}{p_{D}(1)}
		\mu(1, x)
		p_{X \mid D}(x \mid 1)
		\\
		\label{eq:if-psi1-drp}
		 & =
		\frac{D}{\rho}\{\Delta Y - \mu(1, X)\}
		+
		\frac{D}{\rho}
		\mu(1, X)
		-
		\frac{D}{\rho}
		\psi_{1}
	\end{align}
	where we
	in the first equality used \Cref{trick:1},
	in the second equality used \Cref{trick:2a},
	and defined $\rho := p_{D}(1) = P(D = 1)$.

	\paragraph{Influence function of $\psi_{2}$.}
	Using the identity:
	\begin{align}
		\label{eq:helper-if-zhao-panel}
		\frac{p_{X \mid D}(x \mid 1)}{p_{X, D}(x, 0)}
		=
		\frac{p_{X, D}(x, 1)/p_{D}(1)}{p_{D \mid X}(0 \mid x)p(x)}
		= \frac{1}{p_{D}(1)}\frac{p_{D \mid X}(1 \mid x)}{1 - p_{D \mid X}(1 \mid x)}
	\end{align}
	we get:
	\begin{align}
		\nonumber
		 & \mathbb{IF}(\psi_{2})
		\\
		\nonumber
		 & = \sum_{x}
		\left[ \mathbb{IF}\{\mu(0, x)\}p_{X \mid D}(x \mid 1)
		+ \mu(0, x)\mathbb{IF}\{p_{X \mid D}(x \mid 1)\} \right]
		\\
		\nonumber
		 & =
		\sum_{x}
		\left(\frac{\mathbf{1}\{X = x, D = 0\}}{p_{X, D}(x, 0)}\{\Delta Y - \mu(0, x)\}p_{X \mid D}(x \mid 1)
		+ \mu(0, x)\left[\frac{D}{p_{D}(1)}
			\left\{
			\mathbf{1}\{X = x\}
			- p_{X \mid D}(x \mid 1)
			\right\}\right]\right)
		\\
		\nonumber
		 & =
		\sum_{x}
		\left((1 - D)
		\frac{\mathbf{1}\{X = x\}}{p_{D}(1)}\frac{p_{D \mid X}(1 \mid x)}{1 - p_{D \mid X}(1 \mid x)}
		\{\Delta Y - \mu(0, x)\}
		+ \mu(0, x)\left[\frac{D}{p_{D}(1)}
			\left\{
			\mathbf{1}\{X = x\}
			- p_{X \mid D}(x \mid 1)
			\right\}\right]\right)
		\\
		\label{eq:if-psi2-drp}
		 & =
		\frac{(1 - D)}{\rho}\frac{p(X)}{1 - p(X)}
		\{\Delta Y - \mu(0, X)\}
		+\frac{D}{\rho}\mu(0, X)
		-
		\frac{D}{\rho}
		\psi_{0}
	\end{align}
	where we in the third equality used
	\eqref{eq:helper-if-zhao-panel},
	and defined $p(X) := p_{D \mid X}(1 \mid x) = P(D = 1 \mid X)$.
	\paragraph{Combining Influence functions.}
	Combining \cref{eq:if-psi1-drp,eq:if-psi2-drp}:
	\begin{align*}
		 & \mathbb{IF}(\tau^{p})
		\\
		 & = \mathbb{IF}(\psi_{1}) - \mathbb{IF}(\psi_{2})
		\\
		 & =
		\frac{D}{\rho}\{\Delta Y - \mu(1, X)\}
		+
		\frac{D}{\rho}
		\mu(1, X)
		-
		\frac{D}{\rho}
		\psi_{1}
		\\ & \quad -
		\left(
		\frac{(1 - D)}{\rho}\frac{p(X)}{1 - p(X)}
		\{\Delta Y - \mu(0, X)\}
		+ \mu(0, X)\frac{D}{\rho}
		-
		\frac{D}{\rho}
		\psi_{0}
		\right)
		\\
		 & =
		\frac{D}{\rho}\{\Delta Y - \mu(1, X)\}
		-
		\frac{(1 - D)}{\rho}\frac{p(X)}{1 - p(X)}
		\{\Delta Y - \mu(0, X)\}
		+ \frac{D}{\rho}
		[
			\mu(1, X) - \mu(0, X)
		]
		- \frac{D}{\rho}[\psi_{1} - \psi_{0}]
		\\
		 & =
		\frac{D}{\rho}\{\Delta Y - \mu(1, X)\}
		-
		\frac{(1 - D)}{\rho}\frac{p(X)}{1 - p(X)}
		\{\Delta Y - \mu(0, X)\}
		+ \frac{D}{\rho}
		[
			\mu(1, X) - \mu(0, X) - \tau^{p}
		]
	\end{align*}
	which matches the expression in the appendix of \citep[p. 6]{sazhao}.

	This expression can be further reduced as:
	\begin{align}
		\nonumber
		\mathbb{IF}(\tau^{p})
		 & =
		\frac{D}{\rho}\{\Delta Y - \mu(1, X)\}
		-
		\frac{(1 - D)}{\rho}\frac{p(X)}{1 - p(X)}
		\{\Delta Y - \mu(0, X)\}
		+ \frac{D}{\rho}
		[
			\mu(1, X) - \mu(0, X) - \tau^{p}
		]
		\\
		\nonumber
		 & =
		\frac{D}{\rho}\{\Delta Y - \mu(0, X)\}
		-
		\frac{(1 - D)}{\rho}\frac{p(X)}{1 - p(X)}
		\{\Delta Y - \mu(0, X)\}
		- \frac{D}{\rho}\tau^{p}
		\\
		\nonumber
		 & =
		\nonumber
		\left[
			\frac{D}{\rho}
			-
			\frac{(1 - D)}{\rho}\frac{p(X)}{1 - p(X)}
			\right]
		\{\Delta Y - \mu(0, X)\}
		- \frac{D}{\rho}\tau^{p}
		\\
		\nonumber
		 & =
		\left\{ w^{p}_{1}(D) - w^{p}_{0}(D, X; p) \right\}
		\{\Delta Y - \mu(0, X)\}
		- \frac{D}{\rho}\tau^{p}
	\end{align}
	completing the proof.
\end{proof}
Note that reducing the estimand in \cref{eq:tau-drdid-estimand}
by the LIE to $E[\Delta Y \mid D  = 1] - E[E[\Delta Y \mid D = 0, X] \mid D = 1]$
and then applying the IF-derivation-procedure of course leads to the same IF,
namely \cref{eq:if-att-did-reduced}.

\subsubsection{Repeated Cross-Sections}
\label{sec:rc-if-derivation}
\begin{proposition}
	\label{prop:att-if-estimand-rc}
	Let $D, T \in \{0, 1\}$ be binary,
	for $D$ a treatment dummy
	and $T$ a time dummy;
	$X$ some vector of covariates;
	and $Y$ an outcome variable.
	Define the estimand of interest to be the
	\begin{align}
		\label{eq:estimand-rc-drdid}
		\tau^{rc} & = E\Big[
			E\left[ Y \mid D = 1, T = 1, X \right]
			- E\left[ Y \mid D = 1, T = 0, X \right]
		\\ & \quad
			\nonumber
			- \Big(
			E\left[ Y \mid D = 0, T = 1, X \right]
			- E\left[ Y \mid D = 0, T = 0, X \right]
			\Big)
			\mid D = 1
			\Big]
		\\
		\nonumber
		          & = E[
		\mu_{1, 1}(X) - \mu_{1, 0}(X)
		- [
		\mu_{0, 1}(X) - \mu_{0, 0}(X)
		]
		\mid D = 1
		]
	\end{align}
	for $\mu_{d, t}(X) := E\left[ Y \mid D = d, T = t, X \right]$.
	Assume that $(D, X)$ is invariant to $T$ and that we sample from a mixture
	distribution with probability $\lambda :=
		P(T  = 1)$\footnote{
		I.e. akin to assumption 1.b of \cite{sazhao}.
	}.
	Define:
	\begin{align}
		\label{eq:mean-dY-rc}
		\mu_{d, Y}(T, X) :=
		T \cdot \mu_{d, 1}(X)
		+ (1 - T) \cdot \mu_{d, 0}(X)
	\end{align}
	Then the influence function of $\tau^{rc}$ equals:
	\begin{align}
		\label{eq:if-estimand-rc-drdid}
		\mathbb{IF}(\tau^{rc})
		          & =
		w^{rc}_{1}(D, T) \{Y - \mu_{1, Y}(T, X)\}
		-
		w^{rc}_{0}(D, T, X; p)
		\{Y - \mu_{0, Y}(T, X)\}
		\\
		\nonumber & \quad
		+ \frac{D}{\rho}
		\left\{
		\mu_{1, 1}(X)
		-
		\mu_{1, 0}(X)
		-
		[
		\mu_{0, 1}(X)
		-
		\mu_{0, 0}(X)
		]
		- \tau^{rc}
		\right\}
	\end{align}
	where $\lambda = P(T = 1)$
	and
	\begin{align}
		\label{eq:weights1-rc}
		w^{rc}_{1}(D, T)
		 & = w^{rc}_{1, 1}(D, T) - w^{rc}_{1, 0}(D, T)
		:= \frac{D T}{\rho \lambda}
		- \frac{D (1 - T)}{\rho (1 - \lambda)}
		\\
		\label{eq:weights0-rc}
		w^{rc}_{0}(D, T, X; p)
		 & = 		w^{rc}_{0, 1}(D, T, X; p)
		- w^{rc}_{0, 0}(D, T, X; p)
		\\
		 & :=
		\nonumber
		\frac{(1 - D) T}{\rho \lambda}
		\frac{p(X)}{1 - p(X)}
		- \frac{(1 - D) (1 - T)}{\rho (1 - \lambda)}
		\frac{p(X)}{1 - p(X)}
	\end{align}
\end{proposition}
\begin{proof}
	Split up the estimand as
	\begin{align}
		 & \tau^{rc} = E\Big[
		\mu_{1, 1}(X) - \mu_{1, 0}(X)
		- [
		\mu_{0, 1}(X) - \mu_{0, 0}(X)
		]
		\mid D = 1
		\Big]
		\label{eq:estimands-rc}
		=:
		\psi_{1, 1} - \psi_{1, 0} - [\psi_{0, 1} - \psi_{0, 0}]
	\end{align}
	As a helper, consider the influence function for $d, t \in \{0, 1\}$:
	\begin{align}
		\nonumber
		 & \mathbb{IF}(E[\mu_{d, t}(X) \mid D = 1])
		\\
		\nonumber
		 & =
		\sum_{x}
		\Big(
		\frac{\mathbf{1}\{D = d, T = t, X = x\}}{p_{D, T, X}(d, t, x)}
		\{Y - \mu_{d, t}(x)\}
		p_{X \mid D}(x \mid 1)
		\\ & \quad \quad \quad  \quad
		\nonumber
		+ \mu_{d, t}(x)\left[\frac{D}{p_{D}(1)}
			\left\{
			\mathbf{1}\{X = x\}
			- p_{X \mid D}(x \mid 1)
			\right\}\right]\Big)
		\\
		\nonumber
		 & =
		\sum_{x}
		\frac{\mathbf{1}\{D = d\} \mathbf{1}\{T = t\}\mathbf{1}\{X = x\}}{p_{D, X \mid T}(d, x \mid t)p_{T}(t)}
		\{Y - \mu_{d, t}(x)\}
		p_{X \mid D}(x \mid 1)
		+ \frac{D}{p_{D}(1)}
		\left\{
		\mu_{d, t}(X)
		- \psi_{d, t}
		\right\}
		\\
		\label{eq:helper-rc-if}
		 & =
		\sum_{x}
		\frac{\mathbf{1}\{D = d\} \mathbf{1}\{T = t\}\mathbf{1}\{X = x\}}{p(d, x)p_{T}(t)}
		\{Y - \mu_{d, t}(x)\}
		p_{X \mid D}(x \mid 1)
		+ \frac{D}{p_{D}(1)}
		\left\{
		\mu_{d, t}(X)
		- \psi_{d, t}
		\right\}
	\end{align}
	where the last line follows by the assumption of $(D, X)$ being invariant to
	$T$.

	\medskip\noindent
	Setting $d = 1$ in \cref{eq:helper-rc-if} gives:
	\begin{align*}
		 & \mathbb{IF}(E[\mu_{1, t}(X) \mid D = 1])
		\\
		 & =
		\sum_{x}
		\frac{\mathbf{1}\{D = 1\} \mathbf{1}\{T = t\}\mathbf{1}\{X = x\}}{p_{X \mid
					D}(x \mid 1)p_{D}(1)p_{T}(t)}
		\{Y - \mu_{1, t}(x)\}
		p_{X \mid D}(x \mid 1)
		+ \frac{D}{p_{D}(1)}
		\left\{
		\mu_{1, t}(X)
		- \psi_{1, t}
		\right\}
		\\
		 & =
		\sum_{x}
		D \mathbf{1}\{T = t\} \frac{\mathbf{1}\{X = x\}}{p_{D}(1)p_{T}(t)}
		\{Y - \mu_{1, t}(x)\}
		+ \frac{D}{p_{D}(1)}
		\left\{
		\mu_{1, t}(X)
		- \psi_{1, t}
		\right\}
		\\
		 & =
		\frac{D \mathbf{1}\{T = t\}}{\rho [t \lambda + (1 - t)(1 - \lambda)]}
		\{Y - \mu_{1, t}(X)\}
		+ \frac{D}{p_{D}(1)}
		\left\{
		\mu_{1, t}(X)
		- \psi_{1, t}
		\right\}
	\end{align*}
	where we set $p_{T}(t) = t \lambda + (1 - t)(1 - \lambda)$.
	Hence, the influence function of the first two components in
	\cref{eq:estimands-rc} equal:
	\begin{align*}
		\mathbb{IF}(\psi_{1, 1})
		 & =
		\mathbb{IF}(E[\mu_{1, 1}(X) \mid D = 1])
		=
		\frac{D T}{\rho \lambda}
		\{Y - \mu_{1, 1}(X)\}
		+ \frac{D}{\rho}
		\left\{
		\mu_{1, 1}(X)
		- \psi_{1, 1}
		\right\}
		\\
		\mathbb{IF}(\psi_{1, 0})
		 & =
		\mathbb{IF}(E[\mu_{1, 0}(X) \mid D = 1])
		=
		\frac{D (1 - T)}{\rho (1 - \lambda)}
		\{Y - \mu_{1, 0}(X)\}
		+ \frac{D}{\rho}
		\left\{
		\mu_{1, 0}(X)
		- \psi_{1, 0}
		\right\}
	\end{align*}
	where $\lambda := p_{T}(1) = P(T = 1)$.

	\medskip\noindent
	Setting $d = 0$ in \cref{eq:helper-rc-if} gives:
	\begin{align*}
		 & \mathbb{IF}(E[\mu_{0, t}(X) \mid D = 1])
		\\
		 & =
		\sum_{x}
		\frac{\mathbf{1}\{D = 0\} \mathbf{1}\{T = t\}\mathbf{1}\{X = x\}}{p_{X, D}(x, 0)p_{T}(t)}
		\{Y - \mu_{d, t}(x)\}
		p_{X \mid D}(x \mid 1)
		+ \frac{D}{p_{D}(1)}
		\left\{
		\mu_{d, t}(X)
		- \psi_{0, t}
		\right\}
		\\
		 & =
		\sum_{x}
		\frac{\mathbf{1}\{D = 0\} \mathbf{1}\{T = t\}\mathbf{1}\{X = x\}}{p_{D}(1)p_{T}(t)}\frac{p_{D \mid X}(1 \mid x)}{1 - p_{D \mid X}(1 \mid x)}
		\{Y - \mu_{0, t}(x)\}
		+ \frac{D}{p_{D}(1)}
		\left\{
		\mu_{0, t}(X)
		- \psi_{0, t}
		\right\}
		\\
		 & =
		\frac{(1 - D) \mathbf{1}\{T = t\}}{\rho [t \lambda + (1 - t)(1 - \lambda)]}
		\frac{p(X)}{1 - p(X)}
		\{Y - \mu_{0, t}(X)\}
		+ \frac{D}{\rho}
		\left\{
		\mu_{0, t}(X)
		- \psi_{0, t}
		\right\}
	\end{align*}
	where we in the second equality used
	\eqref{eq:helper-if-zhao-panel}.
	So:
	\begin{align*}
		\mathbb{IF}(\psi_{0, 1})
		 & =
		\mathbb{IF}(E[\mu_{0, 1}(X) \mid D = 1])
		\\
		 & =
		\frac{(1 - D) T}{\rho \lambda}
		\frac{p(X)}{1 - p(X)}
		\{Y - \mu_{0, 1}(X)\}
		+ \frac{D}{\rho}
		\left\{
		\mu_{0, 1}(X)
		- \psi_{0, 1}
		\right\}
		\\
		\mathbb{IF}(\psi_{0, 0})
		 & =
		\mathbb{IF}(E[\mu_{0, 0}(X) \mid D = 1])
		\\
		 & =
		\frac{(1 - D) (1 - T)}{\rho (1 - \lambda)}
		\frac{p(X)}{1 - p(X)}
		\{Y - \mu_{0, 0}(X)\}
		+ \frac{D}{\rho}
		\left\{
		\mu_{0, 0}(X)
		- \psi_{0, 0}
		\right\}
	\end{align*}

	\paragraph{Combining:}
	\begin{align}
		\nonumber
		\mathbb{IF}(\tau^{rc})
		          & = \mathbb{IF}(\psi_{1, 1}) - \mathbb{IF}(\psi_{1, 0})
		- \mathbb{IF}(\psi_{0,1}) + \mathbb{IF}(\psi_{0, 0})
		\\
		\nonumber
		          & =
		\frac{D T}{\rho \lambda}
		\{Y - \mu_{1, 1}(X)\}
		+ \frac{D}{\rho}
		\left\{
		\mu_{1, 1}(X)
		- \psi_{1, 1}
		\right\}
		\\
		\nonumber & \quad
		- \frac{D (1 - T)}{\rho (1 - \lambda)}
		\{Y - \mu_{1, 0}(X)\}
		- \frac{D}{\rho}
		\left\{
		\mu_{1, 0}(X)
		- \psi_{1, 0}
		\right\}
		\\
		\nonumber & \quad
		- \frac{(1 - D) T}{\rho \lambda}
		\frac{p(X)}{1 - p(X)}
		\{Y - \mu_{0, 1}(X)\}
		- \frac{D}{\rho}
		\left\{
		\mu_{0, 1}(X)
		- \psi_{0, 1}
		\right\}
		\\
		\nonumber & \quad
		+ \frac{(1 - D) (1 - T)}{\rho (1 - \lambda)}
		\frac{p(X)}{1 - p(X)}
		\{Y - \mu_{0, 0}(X)\}
		+ \frac{D}{\rho}
		\left\{
		\mu_{0, 0}(X)
		- \psi_{0, 0}
		\right\}
		\\
		\nonumber
		          & =
		\frac{D T}{\rho \lambda}
		\{Y - \mu_{1, 1}(X)\}
		- \frac{D (1 - T)}{\rho (1 - \lambda)}
		\{Y - \mu_{1, 0}(X)\}
		\\
		\nonumber & \quad
		- \frac{(1 - D) T}{\rho \lambda}
		\frac{p(X)}{1 - p(X)}
		\{Y - \mu_{0, 1}(X)\}
		+ \frac{(1 - D) (1 - T)}{\rho (1 - \lambda)}
		\frac{p(X)}{1 - p(X)}
		\{Y - \mu_{0, 0}(X)\}
		\\
		\nonumber & \quad
		+ \frac{D}{\rho}
		\left\{
		\mu_{1, 1}(X)
		-
		\mu_{1, 0}(X)
		-
		[
		\mu_{0, 1}(X)
		-
		\mu_{0, 0}(X)
		]
		- [\psi_{1, 1} - \psi_{1, 0} - \{\psi_{0, 1} - \psi_{0, 0}\}]
		\right\}
	\end{align}
	This expression almost matches the one in the
	appendix of \citep[p.
		8]{sazhao}\footnote{
	\cite{sazhao} haven't reduced their expression entirely as there is a
	$
		\frac{D}{p(X)}\left\{
		\mu_{1, 1}(X)
		-
		\mu_{1, 0}(X)
		-
		[
		\mu_{0, 1}(X)
		-
		\mu_{0, 0}(X)
		]
		\right\}
	$
	left to be to be cancelled.
	}.
	Using $T^{2} = T$ such that $T (1 - T) = (1 - T) T = 0$,
	we can reduce the previous expression into the form
	\begin{align*}
		\mathbb{IF}(\tau^{rc})
		          & =
		w^{rc}_{1}(D, T)Y - w^{rc}_{0}(D, T, X; p) Y
		\\ & \quad
		-
		\frac{D T}{\rho \lambda}
		\{T \cdot \mu_{1, 1}(X) + (1 - T) \cdot \mu_{1, 0}(X)\}
		+ \frac{D (1 - T)}{\rho (1 - \lambda)}
		\{T \cdot \mu_{1, 1}(X) + (1 - T) \cdot \mu_{1, 0}(X)\}
		\\ & \quad
		+ \frac{(1 - D) T}{\rho \lambda}
		\frac{p(X)}{1 - p(X)}
		\{T \cdot \mu_{0, 1}(X) + (1 - T) \cdot \mu_{0, 0}(X)\}
		\\ & \quad
		- \frac{(1 - D) (1 - T)}{\rho (1 - \lambda)}
		\frac{p(X)}{1 - p(X)}
		\{T \cdot \mu_{0, 1}(X) + (1 - T) \cdot \mu_{0, 0}(X)\}
		\\
		\nonumber & \quad
		+ \frac{D}{\rho}
		\left\{
		\mu_{1, 1}(X)
		-
		\mu_{1, 0}(X)
		-
		[
		\mu_{0, 1}(X)
		-
		\mu_{0, 0}(X)
		]
		- [\psi_{1, 1} - \psi_{1, 0} - \{\psi_{0, 1} - \psi_{0, 0}\}]
		\right\}
		\\
		          & =
		w^{rc}_{1}(D, T)Y - w^{rc}_{0}(D, T, X; p) Y
		-
		w^{rc}_{1}(D, T) \mu_{1, Y}(X)
		+
		w^{rc}_{0}(D, T) \mu_{0, Y}(X)
		\\
		\nonumber & \quad
		+ \frac{D}{\rho}
		\left\{
		\mu_{1, 1}(X)
		-
		\mu_{1, 0}(X)
		-
		[
		\mu_{0, 1}(X)
		-
		\mu_{0, 0}(X)
		]
		- [\psi_{1, 1} - \psi_{1, 0} - \{\psi_{0, 1} - \psi_{0, 0}\}]
		\right\}
	\end{align*}
	i.e. \cref{eq:if-estimand-rc-drdid}, completing the proof.
\end{proof}

\subsection{LATT parameter}
\label{sec:latt-parameter-if}
In this section, we apply the influence function derived for the ATT parameter
to arrive at the influence function of the LATT parameter \eqref{eq:latt}.
we do this directly for the case of staggered exposure; the case of two-period
IDiD then follows by setting $\mathcal{T} = 2$.
The key step is using the quotient rule \cref{eq:if-quotient-rule-estimand} together
with the just derived influence functions.

\subsubsection{Panel Data}
\label{sec:latt-parameter-if-panel}
\begin{proofref}{prop:if-latt-panel}
	We apply \Cref{prop:att-panel} using the relabeling described in
	\Cref{rem:relabeling-idid} to obtain the influence function
	\cref{eq:if-att-did-reduced} for the numerator and denominator components.
	Also, we normalize the control weights when going from $w^{p}_{0}(D, X; p)$
	in \cref{eq:weights-panel}
	to
	$w^{nev,p}_{e, t},w^{nye,p}_{e, s, t}$
	in \cref{eq:weights-nev-panel-norm,eq:weights-nye-panel-norm} after relabeling.

	Then:
	\begin{align}
		\label{eq:if-latt-panel-num}
		\varphi_{num}(O; \tau^{p, num}_{e, t}, \eta^{p, num}_{e, t})
		 & :=
		\mathbb{IF}(\tau^{p, num}_{e, t})
		\\
		\nonumber
		 & =
		\left\{ w^{trt,p}_{e} - w^{c,p}_{e, t} \right\}
		\{\Delta_{t-e+1}Y_{t} - m_{e, t}^{c, p}(X)\}
		- \frac{E_{e}}{E[E_{e}]}\tau^{p, num}_{e, t},
		\\
		\label{eq:if-latt-panel-den}
		\varphi_{den}(O; \tau^{p, den}_{e, t}, \eta^{p, den}_{e, t})
		 & :=
		\mathbb{IF}(\tau^{p, den}_{e, t})
		\\
		\nonumber
		 & =
		\left\{ w^{trt,p}_{e} - w^{c,p}_{e, t} \right\}
		\{\Delta_{t-e+1}D_{t} - g_{e, t}^{c, p}(X)\}
		- \frac{E_{e}}{E[E_{e}]}\tau^{p, den}_{e, t}.
	\end{align}

	\medskip\noindent
	Note that the above components are of the form
	\cref{eq:if-ratio-components-linear}; hence the terms involving the target
	parameter for the numerator and denominator will cancel.
	Thus cf. the quotient rule \cref{eq:if-quotient-rule-estimand-linear-if}:
	\begin{align}
		 & \varphi(O; \tau^{p}_{e, t}, \eta^{p}_{e, t})
		:=
		\mathbb{IF}\left(\frac{\tau^{p, num}_{e, t}}{\tau^{p, den}_{e, t}}\right)
		\\
		\nonumber
		 & =
		\frac{1}{\tau^{p, den}_{e, t}}
		\left[
			\mathbb{IF}\left(\tau^{p, num}_{e, t}\right)
			- \frac{\tau^{p, num}_{e, t}}{\tau^{p, den}_{e, t}}
			\mathbb{IF}\left(\tau^{p, den}_{e, t}\right)
			\right]
		\\
		 & =
		\nonumber
		\frac{1}{\tau^{p, den}_{e, t}}
		\Big[
		\left\{ w^{trt,p}_{e} - w^{c,p}_{e, t} \right\}
		\{\Delta_{t-e+1}Y_{t} - m_{e, t}^{c, p}(X)\}
		\nonumber
		- \tau^{p}_{e, t}
		\left[
		\left\{ w^{trt,p}_{e} - w^{c,p}_{e, t} \right\}
		\{\Delta_{t-e+1}D_{t} - g_{e, t}^{c, p}(X)\}
		\right]
		\Big].
	\end{align}
\end{proofref}

\subsubsection{Repeated Cross-Sections}
\label{sec:latt-parameter-if-rc}
\begin{proofref}{prop:if-latt-rc}
	As done in the panel setting, we apply \Cref{prop:att-if-estimand-rc} using
	the relabeling described in \Cref{rem:relabeling-idid} to obtain the
	influence function \cref{eq:if-estimand-rc-drdid} for the numerator and
	denominator components.
	Also, we normalize the control weights when going from $w^{rc}_{0}(D, T, X; p)$
	in \cref{eq:weights0-rc}
	to
	$w^{nev,rc}_{e, t},w^{nye,rc}_{e, s, t}$
	in \cref{eq:weights-nev-rc-norm,eq:weights-nye-rc-norm} after relabeling.

	Then:
	\begin{align}
		\label{eq:if-num-rc}
		          & \mathbb{IF}(\tau^{rc, num}_{e, t})
		\\
		\nonumber
		          & =
		w^{trt,rc}_{e} \{Y - m_{e,Y}^{trt,rc}(X)\}
		-
		w^{c,rc}_{e}
		\{Y - m_{e, Y}^{c,rc}(X)\}
		\\
		\nonumber & \quad
		+ \frac{E_{e}}{\rho}
		\left\{
		m_{e,t}^{trt, rc}(X) - m_{e,e-1}^{trt, rc}(X)
		- [ m_{e, t, t}^{c, rc}(X) - m_{e, t, e-1}^{c, rc}(X) ]
		- \tau^{rc, num}_{e, t}
		\right\},
		\\
		\label{eq:if-denom-rc}
		          & \mathbb{IF}(\tau^{rc, den}_{e, t})
		\\
		\nonumber
		          & =
		w^{trt,rc}_{e} \{D - g_{e,D}^{trt,rc}(X)\}
		-
		w^{c,rc}_{e}
		\{D - g_{e, D}^{c,rc}(X)\}
		\\
		\nonumber & \quad
		+ \frac{E_{e}}{\rho}
		\left\{
		g_{e,t}^{trt, rc}(X) - g_{e,e-1}^{trt, rc}(X)
		-
		[ g_{e, t, t}^{c, rc}(X) - g_{e, t, e-1}^{c, rc}(X) ]
		- \tau^{rc, den}_{e, t}
		\right\}.
	\end{align}
	The above components are of the form
	\cref{eq:if-ratio-components-linear}; hence applying the quotient rule
	\cref{eq:if-quotient-rule-estimand-linear-if} yields:
	\begin{align}
		\nonumber
		          &
		\varphi^{rc}(O; \tau^{rc}_{e, t}, \eta^{rc}_{e, t})
		:=
		\mathbb{IF}\left(\frac{\tau^{rc, num}_{e, t}}{\tau^{rc, den}_{e, t}}\right)
		\\
		\label{eq:latt-ratio-step-rc}
		          & =
		\frac{1}{\tau^{rc, den}_{e, t}}
		\left[
			\mathbb{IF}\left(\tau^{rc, num}_{e, t}\right)
			- \frac{\tau^{rc, num}_{e, t}}{\tau^{rc, den}_{e, t}}
			\mathbb{IF}\left(\tau^{rc, den}_{e, t}\right)
			\right]
		\\
		          & =
		\frac{1}{\tau^{rc, den}_{e, t}}
		\Big[
		w^{trt,rc}_{e} \{Y - m_{e,Y}^{trt,rc}(X)\}
		-
		w^{c,rc}_{e}
		\{Y - m_{e, Y}^{c,rc}(X)\}
		\\
		\nonumber & \quad
		+ \frac{E_{e}}{\rho}
		\left\{
		m_{e,t}^{trt, rc}(X) - m_{e,e-1}^{trt, rc}(X)
		- [ m_{e, t, t}^{c, rc}(X) - m_{e, t, e-1}^{c, rc}(X) ]
		\right\}
		\\ & \quad
		\nonumber
		- \tau^{rc}_{e, t}
		\Big(
		w^{trt,rc}_{e} \{D - g_{e,D}^{trt,rc}(X)\}
		-
		w^{c,rc}_{e}
		\{D - g_{e, D}^{c,rc}(X)\}
		\\
		\nonumber & \quad
		+ \frac{E_{e}}{\rho}
		\left\{
		g_{e,t}^{trt, rc}(X) - g_{e,e-1}^{trt, rc}(X)
		-
		[ g_{e, t, t}^{c, rc}(X) - g_{e, t, e-1}^{c, rc}(X) ]
		\right\}
		\Big)
		\Big].
	\end{align}
\end{proofref}

\section{DR Estimands motivated by EIFs}
\subsection{ATT}
\subsubsection{Panel data}
\begin{proposition}
	When panel data are available,
	a doubly robust estimand for the ATT equals
	\begin{align}
		\label{eq:att-did-id}
		\tau & =
		E\left[
			\left\{
			w^{p}_{1}(D)
			-
			w^{p}_{0}(D, X; p)
			\right\}
			\{\Delta Y - \mu(0, X)\}
			\right]
	\end{align}
\end{proposition}
\begin{proof}
	Using \cref{eq:if-att-did-reduced}:
	\begin{align}
		\nonumber
		E[\mathbb{IF}(\tau)] & =
		E\left[
			\left(
			\frac{D}{\rho}
			-
			\frac{(1 - D)}{\rho}\frac{p(X)}{1 - p(X)}
			\right)
			\{\Delta Y - \mu(0, X)\}
			\right]
		- E\left[\frac{D}{\rho}\tau\right] = 0
		\\
		\iff
		\nonumber
		\tau
		                     & =
		E\left[
			\left\{
			w^{p}_{1}(D)
			-
			w^{p}_{0}(D, X; p)
			\right\}
			\{\Delta Y - \mu(0, X)\}
			\right]
	\end{align}
	using $E[D] = \rho$.
\end{proof}
\subsubsection{Repeated Cross-Sections}
\medskip\noindent
\begin{proposition}
	\label{prop:rc-estimand-set-zero}
	When repeated cross-sections are available,
	the expected value of the influence function \cref{eq:if-estimand-rc-drdid}
	equals:
	\begin{align}
		\label{eq:rc-estimand-exp-if}
		E\left[\mathbb{IF}(\tau)\right]
		 & =
		E \left[
			\{w^{rc}_{1}(D, T) - w^{rc}_{0}(D, T, X; p)\}
			\{Y - \mu_{0, Y}(T, X)\}
			\right]
		+ \kappa - \tau,
	\end{align}
	where
	\begin{align*}
		\kappa & :=
		E \left[ \mu_{1, 1}(X) - \mu_{0, 1}(X) \mid D = 1 \right]
		-
		E[\mu_{1, 1}(X) - \mu_{0, 1}(X) \mid D = 1, T = 1]
		\\ & \quad
		\nonumber
		-
		\left(
		E \left[ \mu_{1, 0}(X) - \mu_{0, 0}(X) \mid D = 1 \right]
		-
		E[\mu_{1, 0}(X) - \mu_{0, 0}(X) \mid D = 1, T = 0]
		\right),
	\end{align*}
	and hence, a doubly robust estimand for the LATT equals
	\begin{align}
		\label{eq:dr-did-rc-estimand}
		\tau & =
		E \left[
			\{w^{rc}_{1}(D, T) - w^{rc}_{0}(D, T, X; p)\}
			\{Y - \mu_{0, Y}(T, X)\}
			\right] + \kappa.
	\end{align}
\end{proposition}
\begin{proof}
	The expected value of the influence function in
	\cref{eq:if-estimand-rc-drdid} equals:
	\begin{align}
		\label{eq:dr-did-rc-estimand-eif0}
		E[\mathbb{IF}(\tau)]
		          & = E \left[w^{rc}_{1}(D, T)
			\{Y - \mu_{1, Y}(T, X)\}\right]
		-
		E \left[w^{rc}_{0}(D, T, X; p)
			\{Y - \mu_{0, Y}(T, X)\}\right]
		\\
		\nonumber & \quad
		+ E \left[\frac{D}{\rho}
		\left\{
		\mu_{1, 1}(X)
		-
		\mu_{1, 0}(X)
		-
		[
		\mu_{0, 1}(X)
		-
		\mu_{0, 0}(X)
		]
		\right\}\right]
		- \tau
	\end{align}
	using $E[D/\rho \tau] = \tau E[D]/\rho = \tau$.

	\medskip\noindent
	Inspecting the third term we see that:
	\begin{align}
		\nonumber
		 & E \left[\frac{D}{\rho}
		\left\{
		\mu_{1, 1}(X)
		-
		\mu_{1, 0}(X)
		-
		[
		\mu_{0, 1}(X)
		-
		\mu_{0, 0}(X)
		]
		\right\}\right]
		\\
		\nonumber
		 & =
		E \left[
			\mu_{1, 1}(X)
			-
			\mu_{1, 0}(X)
			\mid D = 1
			\right]
		-
		E \left[
			\mu_{0, 1}(X)
			-
			\mu_{0, 0}(X)
			\mid D = 1
			\right]
		\\
		\label{eq:dr-did-rc-estimand-term1}
		 & =
		E \left[
			\mu_{1, 1}(X)
			-
			\mu_{0, 1}(X)
			\mid D = 1
			\right]
		-
		E \left[
			\mu_{1, 0}(X)
			-
			\mu_{0, 0}(X)
			\mid D = 1
			\right]
	\end{align}
	by the definition of the conditional expectation
	and rearranging in the last equal sign.

	\medskip\noindent
	By invariance:
	\begin{align}
		\label{eq:cexp-invariance-trick}
		E[\mathbf{1}\{D = 1, T = t\}]
		= E[D]E[\mathbf{1}\{T = t\}]
		= \begin{cases}
			  \rho \lambda,       & t = 1, \\
			  \rho (1 - \lambda), & t = 0.
		  \end{cases}
	\end{align}
	As a helper, for $t \in \{0, 1\}$ and generic $\mu(\cdot)$:
	\begin{align}
		\label{eq:helper-cexp-drdid}
		E[\mu(X) \mid D = 1, T = t]
		 & =
		\frac{E[\mathbf{1}\{D = 1, T = t\} \mu(X)]}
		{E[\mathbf{1}\{D = 1, T = t\}]}
		=
		\frac{E[D \mathbf{1}\{T = t\} \mu(X)]}
		{E[D]E[\mathbf{1}\{T = t\}]}
	\end{align}
	Hence, expanding the weight
	\cref{eq:weights1-rc} we get:
	\begin{align*}
		 & E \left[w^{rc}_{1}(D, T)
			\{Y - \mu_{1, Y}(T, X)\}\right]
		\\
		 & =
		E \left[
			\left( \frac{D T}{\rho \lambda}
			- \frac{D (1 - T)}{\rho (1 - \lambda)} \right )
			Y
			\right]
		-
		E \left[
			\frac{D T}{\rho \lambda} \mu_{1, 1}(X)
			\right]
		+
		E \left[
			\frac{D (1 - T)}{\rho (1 - \lambda)}
			\mu_{1, 0}(X)
			\right]
		\\
		 & =
		E \left[
			w^{rc}_{1}(D, T) Y
			\right]
		-
		E[\mu_{1, 1}(X) \mid D = 1, T = 1]
		+
		E[\mu_{1, 0}(X) \mid D = 1, T = 0]
	\end{align*}
	where the last equality follows by
	\cref{eq:helper-cexp-drdid}.
	Continuing from the last equality, adding and subtracting $0$ twice gives:
	\begin{align}
		\nonumber
		          & E \left[
			w^{rc}_{1}(D, T) Y
			\right]
		-
		E[\mu_{1, 1}(X) \mid D = 1, T = 1]
		+
		E[\mu_{1, 0}(X) \mid D = 1, T = 0]
		\\
		\nonumber & \quad
		\pm
		E[\mu_{0, 1}(X) \mid D = 1, T = 1]
		\pm
		E[\mu_{0, 0}(X) \mid D = 1, T = 0]
		\\
		\nonumber
		          & =
		E \left[
			w^{rc}_{1}(D, T) Y
			\right]
		-
		E[\mu_{0, 1}(X) \mid D = 1, T = 1]
		+
		E[\mu_{0, 0}(X) \mid D = 1, T = 0]
		\\
		\nonumber & \quad
		+
		E[\mu_{0, 1}(X) \mid D = 1, T = 1]
		-
		E[\mu_{0, 0}(X) \mid D = 1, T = 0]
		\\
		\nonumber & \quad
		- E[\mu_{1, 1}(X) \mid D = 1, T = 1]
		+
		E[\mu_{1, 0}(X) \mid D = 1, T = 0]
		\\
		\nonumber
		          & =
		E \left[
			w^{rc}_{1}(D, T) Y
			\right]
		-
		E\left[
			\frac{D T}{\rho \lambda} \mu_{0, 1}(X)\right]
		+
		E\left[
			\frac{D (1 - T)}{\rho (1 - \lambda)} \mu_{0, 0}(X)
			\right]
		\\
		\nonumber & \quad
		- E[\mu_{1, 1}(X) - \mu_{0, 1}(X) \mid D = 1, T = 1]
		+
		E[\mu_{1, 0}(X) - \mu_{0, 0}(X) \mid D = 1, T = 0]
		\\
		\nonumber
		          & =
		E \left[
			w^{rc}_{1}(D, T) Y
			\right]
		-
		E\left[ \frac{D T}{\rho \lambda} T \cdot \mu_{0, 1}(X)\right]
		+
		E\left[
			\frac{D (1 - T)}{\rho (1 - \lambda)} (1 - T) \cdot \mu_{0, 0}(X)
			\right]
		\\
		\nonumber & \quad
		+
		E\left[
			\frac{D (1 - T)}{\rho (1 - \lambda)} T \cdot \mu_{0, 1}(X)
			\right]
		-
		E\left[ \frac{D T}{\rho \lambda} (1 - T) \cdot \mu_{0, 0}(X) \right]
		\\
		\nonumber & \quad
		- E[\mu_{1, 1}(X) - \mu_{0, 1}(X) \mid D = 1, T = 1]
		+
		E[\mu_{1, 0}(X) - \mu_{0, 0}(X) \mid D = 1, T = 0]
		\\
		\nonumber
		          & =
		E \left[
			w^{rc}_{1}(D, T) Y
			\right]
		-
		E\left[
			\left( \frac{D T}{\rho \lambda}
			- \frac{D (1 - T)}{\rho (1 - \lambda)} \right )
			\{
			T \cdot \mu_{0, 1}(X) + (1 - T) \cdot \mu_{0, 0}(X)
			\}
			\right]
		\\
		\nonumber & \quad
		- E[\mu_{1, 1}(X) - \mu_{0, 1}(X) \mid D = 1, T = 1]
		+
		E[\mu_{1, 0}(X) - \mu_{0, 0}(X) \mid D = 1, T = 0]
		\\
		\label{eq:dr-did-rc-estimand-term2}
		          & =
		E \left[
			w^{rc}_{1}(D, T)
			\{Y - \mu_{0, Y}(X)\}
			\right]
		\\
		\nonumber & \quad
		- E[\mu_{1, 1}(X) - \mu_{0, 1}(X) \mid D = 1, T = 1]
		+
		E[\mu_{1, 0}(X) - \mu_{0, 0}(X) \mid D = 1, T = 0]
	\end{align}
	using the helper \cref{eq:helper-cexp-drdid},
	$T^{2} = T; (1 - T)^{2} = (1 - T)$,
	adding zero twice again and
	the mean function definition \cref{eq:mean-dY-rc},
	Hence, combining
	\cref{eq:dr-did-rc-estimand-eif0,eq:dr-did-rc-estimand-term1,eq:dr-did-rc-estimand-term2}
	gives:
	\begin{align*}
		E\left[\mathbb{IF}(\tau)\right] & =
		E \left[
			w^{rc}_{1}(D, T)
			\{Y - \mu_{0, Y}(T, X)\}
			\right]
		\\ & \quad
		- E[\mu_{1, 1}(X) - \mu_{0, 1}(X) \mid D = 1, T = 1]
		+
		E[\mu_{1, 0}(X) - \mu_{0, 0}(X) \mid D = 1, T = 0]
		\\ & \quad
		-
		E \left[w^{rc}_{0}(D, T, X; p)
			\{Y - \mu_{0, Y}(T, X)\}\right]
		\\
		\nonumber                       & \quad
		+ E \left[ \mu_{1, 1}(X) - \mu_{0, 1}(X) \mid D = 1 \right]
		-
		E \left[ \mu_{1, 0}(X) - \mu_{0, 0}(X) \mid D = 1 \right]
		\\
		                                & =
		E \left[
			\{w^{rc}_{1}(D, T) - w^{rc}_{0}(D, T, X; p)\}
			\{Y - \mu_{0, Y}(T, X)\}
			\right]
		\\
		\nonumber                       & \quad
		+
		E \left[ \mu_{1, 1}(X) - \mu_{0, 1}(X) \mid D = 1 \right]
		-
		E[\mu_{1, 1}(X) - \mu_{0, 1}(X) \mid D = 1, T = 1]
		\\ & \quad
		\nonumber
		-
		\left(
		E \left[ \mu_{1, 0}(X) - \mu_{0, 0}(X) \mid D = 1 \right]
		-
		E[\mu_{1, 0}(X) - \mu_{0, 0}(X) \mid D = 1, T = 0]
		\right) - \tau
		\\
		                                & =
		E \left[
			\{w^{rc}_{1}(D, T) - w^{rc}_{0}(D, T, X; p)\}
			\{Y - \mu_{0, Y}(T, X)\}
			\right]
		+ \kappa  - \tau,
	\end{align*}
	and by using that the expected value of the influence function equals zero,
	\cref{eq:dr-did-rc-estimand} now follows.
\end{proof}

\subsection{LATT}
\subsubsection{Panel Data}
\label{sec:dr-estimand-panel}

\begin{proofref}{prop:dr-estimand-panel}
	Setting the expected value of \cref{eq:if-latt-panel} equal to $0$
	yields:
	\begin{align}
		\nonumber
		E\left[\varphi(O; \tau^{p}_{e, t}, \eta^{p}_{e, t})\right]
		 & = 0
		\\
		\iff
		\nonumber
		\tau^{dr, p}
		 & = \frac{
		E[
		\{ w^{trt,p}_{e} - w^{c,p}_{e, t} \}
		\{\Delta_{t-e+1}Y_{t} - m_{e, t}^{c, p}(X)\}
		]
		}{
		E[
		\{ w^{trt,p}_{e} - w^{c,p}_{e, t} \}
		\{\Delta_{t-e+1}D_{t} - g_{e, t}^{c, p}(X)\}
		]
		}.
	\end{align}
\end{proofref}

\subsubsection{Repeated Cross-Sections}
\label{sec:dr-estimand-rc}

\begin{proofref}{prop:dr-estimand-rc}
	The expected value of influence function in repeated cross-sections case
	can be reduced as in \cref{eq:rc-estimand-exp-if}.
	We can apply this to \cref{eq:if-num-rc} and \cref{eq:if-denom-rc}
	to get expressions for $E\left[\mathbb{IF}(\tau^{rc, num}_{e, t})\right]$
	and $E\left[\mathbb{IF}(\tau^{rc, den}_{e, t})\right]$,
	respectively.
	These influence functions are linear in the parameters, and
	hence by the quotient rule \cref{eq:if-quotient-rule-estimand-linear-if}:
	\begin{align}
		\nonumber
		E\left[\mathbb{IF}\left(\frac{\tau^{rc, num}_{e, t}}{\tau^{rc, den}_{e, t}}\right)\right]
		 & =
		\frac{1}{\tau^{rc, den}_{e, t}}
		\left[
			E\left[\mathbb{IF}\left(\tau^{rc, num}_{e, t}\right)\right]
			- \frac{\tau^{rc, num}_{e, t}}{\tau^{rc, den}_{e, t}}
			E\left[\mathbb{IF}\left(\tau^{rc, den}_{e, t}\right)\right]
			\right]
		= 0
		\\
		\iff
		\tau^{dr, rc}_{e, t}
		:=
		\frac{\tau^{rc, num}_{e, t}}{\tau^{rc, den}_{e, t}}
		 & =
		\frac{
			E \left[
				\{w^{trt,rc}_{e} - w^{c,rc}_{e}\}
				\{Y - m_{e, Y}^{c,rc}(X)\}
				\right] + \kappa_{e, t}^{Y, rc}
		}{
			E \left[
				\{w^{trt,rc}_{e} - w^{c,rc}_{e}\}
				\{D - g_{e, Y}^{c,rc}(X)\}
				\right] + \kappa_{e, t}^{D, rc}
		}.
	\end{align}
\end{proofref}

\section{EIFs for DR estimands}
In the following we will derive the influence functions for the DR estimands
in \cref{eq:dr-estimand-panel} and \cref{eq:dr-estimand-rc}.
The normalization of the weights complicates the derivations a little bit.
Fortunately, the quotient rule \cref{eq:if-quotient-rule-estimand} can be applied in
a similar fashion as previously.
To do this, we will first derive some helpers for deriving the influence
function of the weighted terms in the estimands.

\paragraph{EIF Weight Term}
Let $w := \mathbf{1}\{B\}/E[\mathbf{1}\{B\}]$ for some event $B$.
This form encompasses all weights in
\cref{eq:weights1-panel-norm,eq:weights-nye-panel-norm,eq:weights-nev-panel-norm}
and likewise for the repeated cross-section setting.
Let $V_t$ be a generic random variable.
Then
\begin{align*}
	E[w V_t]
	=
	\frac{E[\mathbf{1}\{B\} V_t]}{E[\mathbf{1}\{B\}]}.
\end{align*}
Applying the quotient rule \cref{eq:if-quotient-rule-estimand} yields
\begin{align}
	\nonumber
	\mathbb{IF}\big(E[w V_t]\big)
	 & =
	\frac{1}{E[\mathbf{1}\{B\}]}
	\Big[
		\mathbf{1}\{B\} V_t
		-
		E[\mathbf{1}\{B\} V_t]
		-
		\frac{E[\mathbf{1}\{B\} V_t]}{E[\mathbf{1}\{B\}]}
		\big(
		\mathbf{1}\{B\} - E[\mathbf{1}\{B\}]
		\big)
		\Big]
	\\
	\nonumber
	 & =
	\frac{\mathbf{1}\{B\}}{E[\mathbf{1}\{B\}]}
	\left[
		V_t
		-
		\frac{E[\mathbf{1}\{B\} V_t]}{E[\mathbf{1}\{B\}]}
		\right]
	\\
	\label{eq:if-w1-term-panel}
	 & =
	w \left\{ V_t - E[w V_t] \right\}.
\end{align}

\subsubsection{Panel Data}
\label{sec:eif-dr-estimand-panel}

\medskip\noindent
\begin{proofref}{prop:dr-estimand-weighted-panel}
	Recall definition of estimand \cref{eq:dr-estimand-panel}.
	We will apply \cref{eq:if-w1-term-panel}
	to each term in the numerator and denominator.
	First the numerator:
	\begin{align}
		\nonumber
		\varphi^{dr, p}_{num}
		:=
		\mathbb{IF}(\tau^{dr, p}_{num})
		          & =
		w^{trt,p}_{e}
		\left\{ \Delta_{t-e+1}Y_{t} - E[w^{trt,p}_{e} \Delta_{t-e+1}Y_{t} ] \right\}
		\\
		\nonumber & \quad
		-
		w^{c,p}_{e, t}
		\left\{\Delta_{t-e+1}Y_{t} - E[w^{c,p}_{e, t} \Delta_{t-e+1}Y_{t} ] \right\}
		\\
		\nonumber & \quad
		-w^{trt,p}_{e}
		\left\{ m_{e, t}^{c, p}(X) - E[w^{trt,p}_{e} m_{e, t}^{c, p}(X) ] \right\}
		\\
		\nonumber & \quad
		+
		w^{c,p}_{e, t}
		\left\{m_{e, t}^{c, p}(X) - E[w^{c,p}_{e, t} m_{e, t}^{c, p}(X) ] \right\}
		\\
		\nonumber
		          & =
		w^{trt,p}_{e}
		\{
		\Delta_{t-e+1}Y_{t} - m_{e, t}^{c, p}(X)
		-
		E[ w^{trt,p}_{e} \{\Delta_{t-e+1}Y_{t} - m_{e, t}^{c, p}(X)\} ]
		\}
		\\
		\nonumber & \quad
		-
		w^{c,p}_{e, t}
		\{
		\Delta_{t-e+1}Y_{t} - m_{e, t}^{c, p}(X)
		-
		E[ w^{c,p}_{e, t} \{\Delta_{t-e+1}Y_{t} - m_{e, t}^{c, p}(X)\} ]
		\}
		\\
		\label{eq:if-dr-estimand-panel-num}
		          & =
		w^{trt,p}_{e}
		\{
		\varepsilon^{Y, p}_{e, t}
		-
		E[ w^{trt,p}_{e} \varepsilon^{Y, p}_{e, t} ]
		\}
		-
		w^{c,p}_{e, t}
		\{
		\varepsilon^{Y, p}_{e, t}
		-
		E[ w^{c,p}_{e, t} \varepsilon^{Y, p}_{e, t} ]
		\}
	\end{align}
	Likewise for the denominator:
	\begin{align}
		\varphi^{dr, p}_{den}
		:=
		\mathbb{IF}(\tau^{dr, p}_{den})
		          & =
		w^{trt,p}_{e}
		\{
		\Delta_{t-e+1}D_{t} - g_{e, t}^{c, p}(X)
		-
		E[ w^{trt,p}_{e} \{\Delta_{t-e+1}D_{t} - g_{e, t}^{c, p}(X)\} ]
		\}
		\\
		\nonumber & \quad
		-
		w^{c,p}_{e, t}
		\{
		\Delta_{t-e+1}D_{t} - g_{e, t}^{c, p}(X)
		-
		E[ w^{c,p}_{e, t} \{\Delta_{t-e+1}D_{t} - g_{e, t}^{c, p}(X)\} ]
		\}
		\\
		\label{eq:if-dr-estimand-panel-den}
		          & =
		w^{trt,p}_{e}
		\{
		\varepsilon^{D, p}_{e, t}
		-
		E[ w^{trt,p}_{e} \varepsilon^{D, p}_{e, t} ]
		\}
		-
		w^{c,p}_{e, t}
		\{
		\varepsilon^{D, p}_{e, t}
		-
		E[ w^{c,p}_{e, t} \varepsilon^{D, p}_{e, t} ]
		\}
	\end{align}
	Hence applying the quotient rule \cref{eq:if-quotient-rule-estimand}
	to $\mathbb{IF}(\tau^{dr, p}_{num}/\tau^{dr, p}_{den})$
	gives
	\cref{eq:if-dr-estimand-panel}.
\end{proofref}

\begin{remark}[Relating the results to \cite{sazhao}]
	\label{eq:remark-panel-saz}
	\leavevmode
	\begin{itemize}[nosep]
		\item  \Cref{eq:if-dr-estimand-panel-num} is exactly, when relabeling the
		      exposure variable $E$ to the treatment variable $D$, the $\eta^{p}_{1}(W;
			      \beta) - \eta^{p}_{0}(W; \beta, \gamma)$ term of \cite{sazhao} just
		      without the parametrization.
		      Inserting estimators into the expression and Taylor-expanding the
		      estimators will yield the final $\eta^{p}_{est}(W; \beta^{*},
			      \gamma^{*})$ term of their influence function.
		\item The core of their DR-for-inference argument is
		      that using inverse propensity tilting for the propensity score model,
		      and weighted least squares for the outcome regression,
		      will exactly set the estimation effects $\eta^{p}_{est}(W; \beta^{*},
			      \gamma^{*})$ to $0$,
		      and hence
		      that their \textit{improved} DR estimator,
		      $\hat{\tau}^{dr, p}_{imp}$, admits an asymptotic linear
		      representation that is insensitive to first-step estimators.
	\end{itemize}
\end{remark}

\subsubsection{Repeated Cross-Sections}
\label{sec:eif-dr-estimand-rc}
\paragraph{IF of the weights}
Weights joint rc \cref{eq:weights-control-joint-rc}.
By a similar argument to \cref{eq:if-w1-term-panel},
\begin{align}
	\label{eq:if-w1t-term-rc}
	\mathbb{IF}\left(E[w^{trt,rc}_{e, t} V_{t} ]\right)
	 & =
	w^{trt,rc}_{e, t} \left\{ V_{t} - E[w^{trt,rc}_{e, t} V_{t} ] \right\}
	\\
	\label{eq:if-w0t-term-rc}
	\mathbb{IF}\left(E[w^{c,rc}_{e} V_{t} ]\right)
	 & =
	w^{c,rc}_{e}
	\left\{ V_{t} - E[w^{c,rc}_{e} V_{t} ] \right\}.
\end{align}

\begin{proofref}{prop:dr-estimand-weighted-rc}
	Define
	\begin{align}
		\label{eq:resid-rc-case-joined}
		\varepsilon^{Y, rc}_{e, t} = Y - m_{e, Y}^{c,rc}(X),
		\quad
		\varepsilon^{D, rc}_{e, t} = D - g_{e, D}^{c,rc}(X).
	\end{align}
	\textbf{Numerator:}
	Consider numerator of \cref{eq:dr-estimand-rc}:
	\begin{align}
		\mathbb{IF}(\tau^{dr, rc}_{num})
		=
		\mathbb{IF}(
		E \left[
			\{w^{trt,rc}_{e} - w^{c,rc}_{e}\}
			\varepsilon^{Y, rc}_{e, t}
			\right] + \kappa_{e, t}^{Y, rc}
		)
	\end{align}
	For the first term:
	\begin{align}
		\nonumber
		          & \mathbb{IF}\left(
		E \left[
			\{w^{trt,rc}_{e} - w^{c,rc}_{e}\}
			\varepsilon^{Y, rc}_{e, t}
			\right]
		\right)
		\\
		\nonumber
		          & =
		\mathbb{IF}\left(
		E \left[
			[w^{trt,rc}_{e, t} - w^{trt,rc}_{e, e-1}]
		\varepsilon^{Y, rc}_{e, t}
		-
		[w^{c,rc}_{e, t, t} - w^{c,rc}_{e, t, e - 1}]
		\varepsilon^{Y, rc}_{e, t}
		\right]
		\right)
		\\
		\label{eq:if-first-term-rc-num-dr}
		          & =
		w^{trt,rc}_{e, t}
		\left\{ \varepsilon^{Y, rc}_{e, t}
		- E[w^{trt,rc}_{e, t} \varepsilon^{Y, rc}_{e, t} ] \right\}
		-
		w^{trt,rc}_{e, e-1}
		\left\{ \varepsilon^{Y, rc}_{e, t}
		- E[w^{trt,rc}_{e, e-1} \varepsilon^{Y, rc}_{e, t} ] \right\}
		\\
		\nonumber & \quad
		- w^{c,rc}_{e, t, t}
		\left\{ \varepsilon^{Y, rc}_{e, t}
		- E[w^{c,rc}_{e, t, t} \varepsilon^{Y, rc}_{e, t} ] \right\}
		+ w^{c,rc}_{e, t, e - 1}
		\left\{ \varepsilon^{Y, rc}_{e, t}
		- E[w^{c,rc}_{e, t, e - 1} \varepsilon^{Y, rc}_{e, t} ] \right\}
	\end{align}
	which matches the
	$\eta^{rc,1}_{1}(W; \beta) - \eta^{rc,1}_{0}(W; \beta)$
	term of \cite{sazhao} without the parametrization.

	Recall def of $\kappa_{e, t}^{Y, rc}$ in \cref{eq:kappa-Y-rc}
	and the conditional expectation trick \cref{eq:helper-cexp-drdid} (under the
	mixture assumption \Cref{ass:sample:mixture}).
	We can write it out as:
	\begin{align}
		\nonumber
		 & \kappa_{e, t}^{Y, rc}
		\\
		\nonumber
		 & =
		E [ m_{e, t}^{trt, rc}(X)
				-  m_{e, t, t}^{c, rc}(X) \mid E_{e} = 1 ]
		-
		E[ m_{e, t}^{trt, rc}(X)
				-  m_{e, t, t}^{c, rc}(X) \mid E_{e} = 1, T = t ]
		\\ & \quad
		\nonumber
		-
		\left(
		E [ m_{e, e-1}^{trt, rc}(X)
					-  m_{e, t, e-1}^{c, rc}(X) \mid E_{e} = 1 ]
		-
		E[m_{e, e-1}^{trt, rc}(X)
					-  m_{e, t, e-1}^{c, rc}(X) \mid E_{e} = 1, T = e - 1]
		\right),
		\\
		\label{eq:kappa-y-rc-exp}
		 & =
		E\left[\{w^{trt,p}_{e} - w^{trt,rc}_{e, t}\}
			\{m_{e, t}^{trt, rc}(X) -  m_{e, t, t}^{c, rc}(X)\}\right]
		-
		E\left[[w^{trt,p}_{e} - w^{trt,rc}_{e, e-1}]
			\{m_{e, e-1}^{trt, rc}(X) -  m_{e, t, e-1}^{c, rc}(X)\}\right]
	\end{align}
	Thus,
	we can use
	\cref{eq:if-w1-term-panel,eq:if-w1t-term-rc}
	to get:
	\begin{align}
		\label{eq:if-kappa-Y}
		 & \mathbb{IF}(\kappa_{e, t}^{Y, rc})
		\\
		\nonumber
		 & =
		w^{trt,p}_{e} \left\{ \{m_{e, t}^{trt, rc}(X) -  m_{e, t, t}^{c, rc}(X)\}
		- E[w^{trt,p}_{e} \{m_{e, t}^{trt, rc}(X) -  m_{e, t, t}^{c, rc}(X)\} ] \right\}
		\\ & \quad
		\nonumber
		- w^{trt,rc}_{e, t} \left\{ \{m_{e, t}^{trt, rc}(X) -  m_{e, t, t}^{c, rc}(X)\}
		- E[w^{trt,rc}_{e, t} \{m_{e, t}^{trt, rc}(X) -  m_{e, t, t}^{c, rc}(X)\} ] \right\}
		\\ & \quad
		\nonumber
		- w^{trt,p}_{e} \left\{ \{m_{e, e-1}^{trt, rc}(X) -  m_{e, t, e-1}^{c, rc}(X)\}
		- E[w^{trt,p}_{e} \{m_{e, e-1}^{trt, rc}(X) -  m_{e, t, e-1}^{c, rc}(X)\} ] \right\}
		\\ & \quad
		\nonumber
		+ w^{trt,rc}_{e, e-1} \left\{ \{m_{e, e-1}^{trt, rc}(X) -  m_{e, t, e-1}^{c, rc}(X)\}
		- E[w^{trt,rc}_{e, e-1} \{m_{e, e-1}^{trt, rc}(X) -  m_{e, t, e-1}^{c, rc}(X)\} ] \right\}
		\\
		\nonumber
		 & =
		w^{trt,p}_{e} \left\{ m_{e, \Delta}^{trt, rc}(X)
		- E[w^{trt,p}_{e} m_{e, \Delta}^{trt, rc}(X) ] \right\}
		- w^{trt,p}_{e} \left\{ m_{e, t, \Delta}^{c, rc}(X)
		- E[w^{trt,p}_{e} m_{e, t, \Delta}^{c, rc}(X) ] \right\}
		\\ & \quad
		\nonumber
		- w^{trt,rc}_{e, t} \left\{ \{m_{e, t}^{trt, rc}(X) -  m_{e, t, t}^{c, rc}(X)\}
		- E[w^{trt,rc}_{e, t} \{m_{e, t}^{trt, rc}(X) -  m_{e, t, t}^{c, rc}(X)\} ] \right\}
		\\ & \quad
		\nonumber
		+ w^{trt,rc}_{e, e-1} \left\{ \{m_{e, e-1}^{trt, rc}(X) -  m_{e, t, e-1}^{c, rc}(X)\}
		- E[w^{trt,rc}_{e, e-1} \{m_{e, e-1}^{trt, rc}(X) -  m_{e, t, e-1}^{c, rc}(X)\} ] \right\}
	\end{align}
	For the residuals \cref{eq:resid-rc-case-joined}, weights
	\cref{eq:weights-trt-joint-rc,eq:weights-control-joint-rc}
	and mean function \cref{eq:mean-function-both-period-rc},
	multiplying the control residual with the weight picks out the mean function in
	the same period as the weight:
	\begin{align*}
		w^{trt,rc}_{e, t} \varepsilon^{Y, rc}_{e, t}
		 & =
		w^{trt,rc}_{e, t} [Y - m_{e, Y}^{c,rc}(X)]
		=
		w^{trt,rc}_{e, t} [Y - m_{e, t, t}^{c,rc}(X)]
		\\
		w^{c,rc}_{e, t, t} \varepsilon^{Y, rc}_{e, t}
		 & =
		w^{c,rc}_{e, t, t} [Y - m_{e, t, t}^{c,rc}(X)]
	\end{align*}
	We can use this on \cref{eq:if-first-term-rc-num-dr} to get:
	\begin{align*}
		          & \mathbb{IF}\left(
		E \left[
			\{w^{trt,rc}_{e} - w^{c,rc}_{e}\}
			\varepsilon^{Y, rc}_{e, t}
			\right]
		\right)
		\\
		          & =
		w^{trt,rc}_{e, t}
		\left\{ [Y - m_{e, t, t}^{c,rc}(X)]
		- E[w^{trt,rc}_{e, t} [Y - m_{e, t, t}^{c,rc}(X)] ] \right\}
		\\ & \quad
		-
		w^{trt,rc}_{e, e-1}
		\left\{ [Y - m_{e, t, e - 1}^{c,rc}(X)]
		- E[w^{trt,rc}_{e, e-1} [Y - m_{e, t, e - 1}^{c,rc}(X)] ] \right\}
		\\
		\nonumber & \quad
		- w^{c,rc}_{e, t, t}
		\left\{ [Y - m_{e, t, t}^{c,rc}(X)]
		- E[w^{c,rc}_{e, t, t} [Y - m_{e, t, t}^{c,rc}(X)] ] \right\}
		\\ & \quad
		+ w^{c,rc}_{e, t, e - 1}
		\left\{ [Y - m_{e, t, e - 1}^{c,rc}(X)]
		- E[w^{c,rc}_{e, t, e - 1} [Y - m_{e, t, e - 1}^{c,rc}(X)] ] \right\}
	\end{align*}
	Thus,
	\begin{align*}
		          & \mathbb{IF}\left(
		E \left[
			\{w^{trt,rc}_{e} - w^{c,rc}_{e}\}
			\varepsilon^{Y, rc}_{e, t}
			\right]
		\right)
		\\ & \quad
		- w^{trt,rc}_{e, t} \left\{ \{m_{e, t}^{trt, rc}(X) -  m_{e, t, t}^{c, rc}(X)\}
		- E[w^{trt,rc}_{e, t} \{m_{e, t}^{trt, rc}(X) -  m_{e, t, t}^{c, rc}(X)\} ] \right\}
		\\ & \quad
		\nonumber
		+ w^{trt,rc}_{e, e-1} \left\{ \{m_{e, e-1}^{trt, rc}(X) -  m_{e, t, e-1}^{c, rc}(X)\}
		- E[w^{trt,rc}_{e, e-1} \{m_{e, e-1}^{trt, rc}(X) -  m_{e, t, e-1}^{c, rc}(X)\} ] \right\}
		\\
		          & =
		w^{trt,rc}_{e, t}
		\left\{ [Y - m_{e, t}^{trt, rc}(X)]
		- E[w^{trt,rc}_{e, t} [Y - m_{e, t}^{trt, rc}(X)] ] \right\}
		\\ & \quad
		-
		w^{trt,rc}_{e, e-1}
		\left\{ [Y - m_{e, e - 1}^{trt, rc}(X)]
		- E[w^{trt,rc}_{e, e-1} [Y - m_{e, e - 1}^{trt, rc}(X)] ] \right\}
		\\
		\nonumber & \quad
		- w^{c,rc}_{e, t, t}
		\left\{ [Y - m_{e, t, t}^{c,rc}(X)]
		- E[w^{c,rc}_{e, t, t} [Y - m_{e, t, t}^{c,rc}(X)] ] \right\}
		\\ & \quad
		+ w^{c,rc}_{e, t, e - 1}
		\left\{ [Y - m_{e, t, e - 1}^{c,rc}(X)]
		- E[w^{c,rc}_{e, t, e - 1} [Y - m_{e, t, e - 1}^{c,rc}(X)] ] \right\}
	\end{align*}
	i.e. the mean-functions of the controls cancel under the treated weights.
	Hence, combining \cref{eq:if-first-term-rc-num-dr,eq:if-kappa-Y}
	\begin{align}
		          & \mathbb{IF}(\tau^{dr, rc}_{num})
		\\
		\nonumber
		          & =
		w^{trt,rc}_{e, t}
		\left\{ [Y - m_{e, t}^{trt, rc}(X)]
		- E[w^{trt,rc}_{e, t} [Y - m_{e, t}^{trt, rc}(X)] ] \right\}
		\\ & \quad
		\nonumber
		+ w^{trt,p}_{e} \left\{ m_{e, \Delta}^{trt, rc}(X)
		- E[w^{trt,p}_{e} m_{e, \Delta}^{trt, rc}(X) ] \right\}
		\\ & \quad
		\nonumber
		-
		w^{trt,rc}_{e, e-1}
		\left\{ [Y - m_{e, e - 1}^{trt, rc}(X)]
		- E[w^{trt,rc}_{e, e-1} [Y - m_{e, e - 1}^{trt, rc}(X)] ] \right\}
		\\ & \quad
		\nonumber
		- w^{trt,p}_{e} \left\{ m_{e, t, \Delta}^{c, rc}(X)
		- E[w^{trt,p}_{e} m_{e, t, \Delta}^{c, rc}(X) ] \right\}
		\\
		\nonumber & \quad
		- w^{c,rc}_{e, t, t}
		\left\{ [Y - m_{e, t, t}^{c,rc}(X)]
		- E[w^{c,rc}_{e, t, t} [Y - m_{e, t, t}^{c,rc}(X)] ] \right\}
		\\ & \quad
		\nonumber
		+ w^{c,rc}_{e, t, e - 1}
		\left\{ [Y - m_{e, t, e - 1}^{c,rc}(X)]
		- E[w^{c,rc}_{e, t, e - 1} [Y - m_{e, t, e - 1}^{c,rc}(X)] ] \right\}
		\\
		\label{eq:if-match-sazhao}
		          & =
		w^{trt,rc}_{e, t}
		\left\{ \varepsilon^{Y,trt,rc}_{e,t}
		- E[w^{trt,rc}_{e, t} \varepsilon^{Y,trt,rc}_{e,t} ] \right\}
		+ w^{trt,p}_{e} \left\{ m_{e, \Delta}^{trt, rc}(X)
		- E[w^{trt,p}_{e} m_{e, \Delta}^{trt, rc}(X) ] \right\}
		\\ & \quad
		\label{eq:if-match-sazhao2}
		-
		\left[w^{trt,rc}_{e, e-1}
			\left\{ \varepsilon^{Y,trt,rc}_{e,e-1}
			- E[w^{trt,rc}_{e, e-1} \varepsilon^{Y,trt,rc}_{e,e-1} ] \right\}
			+ w^{trt,p}_{e} \left\{ m_{e, t, \Delta}^{c, rc}(X)
			- E[w^{trt,p}_{e} m_{e, t, \Delta}^{c, rc}(X) ] \right\}\right]
		\\
		\label{eq:if-match-sazhao3}
		          & \quad
		- [w^{c,rc}_{e, t, t}
		\left\{ \varepsilon^{Y,c,rc}_{e,t,t}
		- E[w^{c,rc}_{e, t, t} \varepsilon^{Y,c,rc}_{e,t,t} ] \right\}
		- w^{c,rc}_{e, t, e - 1}
		\left\{ \varepsilon^{Y,c,rc}_{e,t,e-1}
		- E[w^{c,rc}_{e, t, e - 1} \varepsilon^{Y,c,rc}_{e,t,e-1} ] \right\}]
	\end{align}
	A similar expression is found for $\mathbb{IF}(\tau^{dr, rc}_{den})$
	by applying the exact same previous steps.
	Applying the quotient rule \cref{eq:if-quotient-rule-estimand} now yields the
	result \cref{eq:if-dr-estimand-rc}.
\end{proofref}

\begin{remark}[Relating the results to \cite{sazhao}]
	\label{eq:remark-rc-saz}
	\leavevmode
	\begin{itemize}[nosep]
		\item Relabeling the exposure variable $E$ to the treatment variable $D$,
		      we see that \cref{eq:if-match-sazhao}, \cref{eq:if-match-sazhao2}, and
		      \cref{eq:if-match-sazhao3} correspond to $\eta^{rc,2}_{1,1}(W; \beta)$,
		      $\eta^{rc,2}_{1,0}(W; \beta)$, and $\eta^{rc,2}_{0}(W; \beta)$ of
		      \cite{sazhao} (without the parametrization), respectively.
		\item Again, the same DR-for-inference comment as in
		      \Cref{eq:remark-panel-saz} also holds for the repeated cross-section
		      case.
		      Interestingly, as noted in \cite{sazhao},
		      estimating the OR coefficients associated with the treated group does
		      not lead to any estimation effect.
	\end{itemize}
\end{remark}

\section{Remainder terms}
\label{sec:remainder-terms}
Below we analyze the remainder terms of the estimators.
We show that the remainder is second order, implying double robustness:
it is proportional to the product of the $L_{2}(P)$ errors of the propensity
score and outcome regression estimators, so each nuisance estimator may
converge at rates slower than $n^{-1/2}$ (in the panel case).
For repeated cross-sections, we additionally require consistency of the
treated mean function.

\paragraph{Decomposition of panel and rcs estimators}
The numerator and denominator can be written as in
\Cref{eq:if-ratio-components-linear}
\begin{align*}
	\varphi(\cdot; \tau_{e, t}, \eta_{e, t})
	 & =
	\phi(\cdot; \eta_{e, t}) - w(\cdot; \eta_{e, t}) \tau_{e, t}
\end{align*}
The estimating equation implies that the estimator $\hat{\tau}_{e, t}$ solves:
\begin{align}
	\label{eq:ee-estimator-generic}
	P_{n}\varphi(\cdot; \hat{\tau}_{e, t}, \hat{\eta}_{e, t}) = 0
	\iff
	\hat{\tau}_{e, t}
	=
	P_{n}\phi(\cdot; \hat{\eta}_{e, t})
\end{align}
using $P_{n}w(\cdot; \hat{\eta}_{e, t}) = 1$.
Thus, the ratio estimator is also an estimating-equation estimator.

Using this:
\begin{align*}
	\hat{\tau}_{e, t} - \tau_{e, t}
	:=
	\frac{\hat{\tau}^{num}_{e, t}}{\hat{\tau}^{den}_{e, t}}
	- \frac{\tau^{num}_{e, t}}{\tau^{den}_{e, t}}
	 & =
	\frac{
		\hat{\tau}^{num}_{e, t} \tau^{den}_{e, t}
		-
		\tau^{num}_{e, t} \hat{\tau}^{den}_{e, t}
	}{
		\hat{\tau}^{den}_{e, t}\tau^{den}_{e, t}
	}
	\\
	 & =
	\frac{1}{\hat{\tau}^{den}_{e, t}}
	\left[
		\hat{\tau}^{num}_{e, t}
		-
		\frac{\tau^{num}_{e, t}}{\tau^{den}_{e, t}} \hat{\tau}^{den}_{e, t}
		\right]
	\\
	 & =
	\frac{1}{\hat{\tau}^{den}_{e, t}}
	\left[
		P_{n}\phi_{num}(\cdot; \hat{\eta}^{num}_{e, t})
		-
		\frac{\tau^{num}_{e, t}}{\tau^{den}_{e, t}} P_{n}\phi_{den}(\cdot; \hat{\eta}^{den}_{e, t})
		\right]
	\\
	 & =
	\frac{1}{\hat{\tau}^{den}_{e, t}}
	P_{n}
	\left[
		\phi_{num}(\cdot; \hat{\eta}^{num}_{e, t})
		-
		\frac{\tau^{num}_{e, t}}{\tau^{den}_{e, t}} \phi_{den}(\cdot; \hat{\eta}^{den}_{e, t})
		\right],
\end{align*}
where in the third equality we used that the numerator and denominator
estimators are of the estimating-equation type, i.e. satisfy
\cref{eq:ee-estimator-generic}.

Multiplying both sides by $\sqrt{n}$ gives:
\begin{align}
	\nonumber
	\sqrt{n}(\hat{\tau}_{e, t} - \tau_{e, t})
	 & =
	\frac{1}{\hat{\tau}^{den}_{e, t}}
	\sqrt{n}
	P_{n}
	\left[
		\phi_{num}(\cdot; \hat{\eta}^{num}_{e, t})
		-
		\frac{\tau^{num}_{e, t}}{\tau^{den}_{e, t}} \phi_{den}(\cdot; \hat{\eta}^{den}_{e, t})
		\right]
	\\
	\nonumber
	 & =
	\frac{1}{\tau^{den}_{e, t}}
	\sqrt{n}
	P_{n}
	\left[
		\phi_{num}(\cdot; \hat{\eta}^{num}_{e, t})
		-
		\frac{\tau^{num}_{e, t}}{\tau^{den}_{e, t}} \phi_{den}(\cdot; \hat{\eta}^{den}_{e, t})
		\right]
	+
	o_{P}(1)
	\\
	\label{eq:latt-estimator-expansion-generic}
	 & =
	\sqrt{n} P_{n} \varphi(\cdot; \tau_{e, t}, \hat{\eta}_{e, t}) + o_{P}(1),
\end{align}
using $\hat{\tau}^{den}_{e, t} \overset{P}{\to} \tau^{den}_{e, t}$ and the definition
of the influence function \Cref{eq:if-quotient-rule-estimand-linear-if}.

\paragraph{Further Decomposition}
A further decomposition of \cref{eq:latt-estimator-expansion-generic} gives:
\begin{align*}
	 & \sqrt{n}(\hat{\tau}_{e,t} - \tau_{e,t})
	\\
	\nonumber
	 & =
	\sqrt{n}
	P_{n} \varphi(\cdot; \tau_{e,t}, \hat{\eta}_{e,t})
	\pm \sqrt{n} P_{n} \varphi(\cdot; \tau_{e,t}, \eta_{e,t})
	\pm \sqrt{n} P \varphi(\cdot; \tau_{e,t}, \hat{\eta}_{e,t})
	+
	o_{P}(1)
	\\
	 & =
	\sqrt{n} (P_{n} - P) \varphi(\cdot; \tau_{e,t}, \eta_{e,t})
	+
	\sqrt{n} (P_{n} - P)
	[
		\varphi(\cdot; \tau_{e,t}, \hat{\eta}_{e,t})
		-
		\varphi(\cdot; \tau_{e,t}, \eta_{e,t})
	]
	\\ & \quad
	\nonumber
	+ \sqrt{n} P \varphi(\cdot; \tau_{e,t}, \hat{\eta}_{e,t})
	+
	o_{P}(1)
\end{align*}
This yields the standard decomposition in \Cref{eq:decomposition}.
The first term is mean-zero with variance
$\mathrm{Var}(\varphi(O; \tau_{e,t}, \eta_{e,t})) =
	E[\{\varphi(O; \tau_{e,t}, \eta_{e,t})\}^{2}]$
and, by a central limit theorem, converges in distribution to
$N(0, E[\{\varphi(O; \tau_{e,t}, \eta_{e,t})\}^{2}])$.
The second term is an empirical process term and is $o_{P}(1)$ under either
Donsker class conditions or by cross-fitting.
The final term is a remainder term, which we analyze for each estimator in this
section.

\paragraph{Estimators satisfying expansion above}
\textbf{Panel EE estimator:}
Define the weights
\cref{eq:weights1-panel-norm,eq:weights-nev-panel-norm,eq:weights-nye-panel-norm}
with estimated nuisances as
$\hat{w}^{trt,p}_{e}, \hat{w}^{c,p}_{e}$ (both control cases)
which sum to one in sample.
Numerator and denominator \cref{eq:if-latt-panel-num,eq:if-latt-panel-den}
satisfies \cref{eq:ee-estimator-generic}.
\begin{align*}
	P_{n}\varphi^{p}_{num}(\cdot; \tau^{p, num}_{e, t}, \hat{\eta}^{p, num}_{e, t})
	 & =
	P_{n}\phi^{p}_{num}(\cdot; \hat{\eta}^{p, num}_{e, t}) - \tau^{p, num}_{e, t}
	=
	\hat{\tau}^{p, num}_{e, t} - \tau^{p, num}_{e, t}
	\\
	P_{n}\varphi^{p}_{den}(\cdot; \tau^{p, den}_{e, t}, \hat{\eta}^{p, den}_{e, t})
	 & =
	P_{n}\phi^{p}_{den}(\cdot; \hat{\eta}^{p, den}_{e, t}) - \tau^{p, den}_{e, t}
	=
	\hat{\tau}^{p, den}_{e, t} - \tau^{p, den}_{e, t}
\end{align*}
where
\begin{align*}
	\hat{\tau}^{p, num}_{e, t}
	=
	P_{n}[
			( \hat{w}^{trt,p}_{e} - \hat{w}^{c,p}_{e, t})
			\hat{\varepsilon}^{Y, p}_{e, t}
		],
	\quad
	\hat{\tau}^{p, den}_{e, t}
	=
	P_{n}[
			( \hat{w}^{trt,p}_{e} - \hat{w}^{c,p}_{e, t})
			\hat{\varepsilon}^{D, p}_{e, t}
		].
\end{align*}
\textbf{Panel EE estimator double IF:}
Have influence functions:
\cref{eq:if-dr-estimand-panel-num,eq:if-dr-estimand-panel-den}.
We can write the influence function components out as:
\begin{align*}
	\varphi^{dr, p, num}(O; \tau^{dr, p, num}_{e, t}, \eta^{dr, p, num}_{e, t})
	 & =
	w^{trt,p}_{e}
	\{
	\varepsilon^{Y, p}_{e, t}
	-
	E[ w^{trt,p}_{e} \varepsilon^{Y, p}_{e, t} ]
	\}
	-
	w^{c,p}_{e, t}
	\{
	\varepsilon^{Y, p}_{e, t}
	-
	E[ w^{c,p}_{e, t} \varepsilon^{Y, p}_{e, t} ]
	\}
	\\
	 & =
	( w^{trt,p}_{e} - w^{c,p}_{e, t})
	\varepsilon^{Y, p}_{e, t}
	- w^{trt,p}_{e}
	E[ w^{trt,p}_{e} \varepsilon^{Y, p}_{e, t} ]
	+
	w^{c,p}_{e, t}
	E[ w^{c,p}_{e, t} \varepsilon^{Y, p}_{e, t} ]
	\\
	 & =:
	( w^{trt,p}_{e} - w^{c,p}_{e, t})
	\varepsilon^{Y, p}_{e, t}
	- w^{trt,p}_{e}
	\tau^{dr,p,num}_{e,t,1}
	+
	w^{c,p}_{e, t}
	\tau^{dr,p,num}_{e,t,0}
\end{align*}
where $\tau^{dr,p,num}_{e,t} = \tau^{dr,p,num}_{e,t,1} - \tau^{dr,p,num}_{e,t,0}$.
Applying the empirical measure to the expression above evaluated at
$\hat{\eta}^{dr, p, num}_{e, t}$ yields:
\begin{align*}
	P_{n} \varphi^{dr, p, num}(\cdot; \tau^{dr, p, num}_{e, t}, \hat{\eta}^{dr, p, num}_{e, t})
	 & =
	P_{n}[
			( \hat{w}^{trt,p}_{e} - \hat{w}^{c,p}_{e, t})
			\hat{\varepsilon}^{Y, p}_{e, t}
		]
	- P_{n}[\hat{w}^{trt,p}_{e}]
	\tau^{dr,p,num}_{e,t,1}
	+
	P_{n}[\hat{w}^{c,p}_{e, t}]
	\tau^{dr,p,num}_{e,t,0}
	\\
	 & =
	P_{n}[
			( \hat{w}^{trt,p}_{e} - \hat{w}^{c,p}_{e, t})
			\hat{\varepsilon}^{Y, p}_{e, t}
		]
	-
	\tau^{dr,p,num}_{e,t}
	\\
	 & =
	\hat{\tau}^{dr, p, num}_{e, t}
	-
	\tau^{dr,p,num}_{e,t}
\end{align*}
The same argument applies to the denominator.
\paragraph{Comparing the decompositions}
The estimators above are very similar.
The latter is based on the DR-estimand while the former is based on the
original LATT estimand.
Hence the estimands, $\tau^{p}_{e, t}$ and $\tau^{dr, p}_{e, t}$, are
different but the meat,
i.e. the estimators
\begin{align*}
	P_{n} \phi^{dr, p}_{num}(\cdot; \hat{\eta}^{dr, p, num}_{e, t})
	= P_{n}\phi^{p}_{num}(\cdot; \hat{\eta}^{p, num}_{e, t}),
	\quad
	P_{n} \phi^{dr, p}_{den}(\cdot; \hat{\eta}^{dr, p, den}_{e, t})
	= P_{n}\phi^{p}_{den}(\cdot; \hat{\eta}^{p, den}_{e, t}),
\end{align*}
are the same due to the cancellation. This is consistent with the fact that
both estimators solve the estimating equation with the original influence
function; what differs is the influence function used for inference.

Applying the decomposition \Cref{eq:latt-estimator-expansion-generic} yields:
\begin{align}
	\nonumber
	\sqrt{n}(\hat{\tau}^{dr,p}_{e,t} - \tau^{dr,p}_{e,t})
	 & =
	\frac{1}{\tau^{dr,p,den}_{e,t}}
	\sqrt{n}
	P_{n}
		[
			\phi^{dr, p, num}(\cdot; \hat{\eta})
			-
			\frac{\tau^{dr,p,num}_{e,t}}{\tau^{dr,p,den}_{e,t}}
			\phi^{dr, p, den}(\cdot; \hat{\eta})
		]
	+
	o_{P}(1)
	\\
	\label{eq:dr-estimator-panel-expansion}
	 & =
	\sqrt{n}
	P_{n}
	\varphi^{dr, p}(\cdot; \tau^{dr, p}_{e,t}, \hat{\eta}^{dr, p}_{e,t})
	+
	o_{P}(1)
\end{align}
and similarly for the other.
Hence the remainder terms,
\begin{align*}
	P \varphi^{dr, p}(\cdot; \tau^{dr, p}, \hat{\eta}^{dr, p}),
	\quad
	P \varphi^{p}(\cdot; \tau^{p}, \hat{\eta^{p}})
\end{align*}
are the same as well; see \Cref{remark:identical-remainders}.

\subsubsection{Remainder Panel Data}
\label{sec:remainder-panel}
\begin{proofref}{prop:remainder-term-panel}
	Consider the DR estimand \cref{eq:dr-estimand-panel}
	with influence function \cref{eq:if-dr-estimand-panel}.
	Write the estimand as $\tau^{dr, p}_{e, t} =
		\tau^{dr, p,num}_{e, t}/\tau^{dr, p,den}_{e, t}$.
	Define notation for the parameters in the influence function
	\cref{eq:if-dr-estimand-panel}:
	\begin{align*}
		\tau^{dr,p,num}_{e, t, 1}
		 & :=
		E[ w^{trt,p}_{e} \varepsilon^{Y, p}_{e, t} ],
		\quad
		\tau^{dr,p,num}_{e, t, 0}
		:=
		E[ w^{c,p}_{e, t} \varepsilon^{Y, p}_{e, t} ]
		\\
		\tau^{dr,p,den}_{e, t, 1}
		 & :=
		E[ w^{trt,p}_{e} \varepsilon^{D, p}_{e, t} ],
		\quad
		\tau^{dr,p,den}_{e, t, 0}
		:=
		E[ w^{c,p}_{e, t} \varepsilon^{D, p}_{e, t} ]
	\end{align*}
	Consider the remainder:
	\begin{align*}
		 & P \varphi^{dr, p}(\cdot; \tau^{dr, p}_{e, t}, \hat{\eta}^{dr, p}_{e, t})
		\\
		 & =
		E\Big[
			\frac{1}{\tau^{dr, p, den}_{e, t}}
			\Big[
				\hat{w}^{trt,p}_{e}
				\{
				\hat{\varepsilon}^{Y, p}_{e, t} - \tau^{dr,p,num}_{e, t, 1}
				\}
				-
				\hat{w}^{c,p}_{e, t}
				\{
				\hat{\varepsilon}^{Y, p}_{e, t} - \tau^{dr,p,num}_{e, t, 0}
				\}
		\\ & \quad
				\nonumber
				- \frac{\tau^{dr, p, num}_{e, t}}{\tau^{dr, p, den}_{e, t}}
				(
				\hat{w}^{trt,p}_{e}
				\{
				\hat{\varepsilon}^{D, p}_{e, t} - \tau^{dr,p,den}_{e, t, 1}
				\}
				-
				\hat{w}^{c,p}_{e, t}
				\{
				\hat{\varepsilon}^{D, p}_{e, t} - \tau^{dr,p,den}_{e, t, 0}
				\}
				)
				\Big]
			\Big]
	\end{align*}
	We focus on the numerator components; the denominator follows by simlar
	arguments.
	Consider the components weighted by the treated weight:
	\begin{align*}
		\tau^{dr,p,num}_{e, t, 1}
		=
		E[ w^{trt,p}_{e} \varepsilon^{Y, p}_{e, t} ]
		=
		E\left[ \frac{E_{e}}{\rho} \varepsilon^{Y, p}_{e, t} \right]
		\iff
		\rho \tau^{dr,p,num}_{e, t, 1}
		=
		E\left[ E_{e} \varepsilon^{Y, p}_{e, t} \right].
	\end{align*}
	Hence, as $\rho \tau^{dr,p,num}_{e, t, 1} = E[E_{e} \tau^{dr,p,num}_{e, t, 1}]$
	and
	\begin{align*}
		E\left[ E_{e} \varepsilon^{Y, p}_{e, t} \right]
		=
		E\left[ E_{e} \{E[\Delta_{t-e+1}Y_{t} \mid X, E_{e}] - m_{e, t}^{c, p}(X)\} \right]
		=
		E\left[ E_{e} \{m_{e, t}^{trt, p}(X) - m_{e, t}^{c, p}(X)\} \right],
	\end{align*}
	where we used the LIE,
	we have:
	\begin{align}
		\label{eq:dummy-times-estimand-trick-panel}
		E[E_{e} \tau^{dr,p,num}_{e, t, 1}]
		=
		E\left[ E_{e} \{m_{e, t}^{trt, p}(X) - m_{e, t}^{c, p}(X)\} \right].
	\end{align}
	Likewise,
	\begin{align*}
		E\left[ \hat{w}^{trt,p}_{e}\hat{\varepsilon}^{Y, p}_{e, t} \right]
		=
		E\left[ \frac{E_{e}}{\hat{\rho}}
		\{\Delta_{t-e+1}Y_{t} - \hat{m}_{e, t}^{c, p}(X)\} \right]
		=
		E\left[ \hat{w}^{trt,p}_{e}
			\{m_{e, t}^{trt, p}(X) - \hat{m}_{e, t}^{c, p}(X)\} \right],
	\end{align*}
	and with an analogous LIE argument for the controls:
	\begin{align*}
		E\left[
			\hat{w}^{c,p}_{e, t}\hat{\varepsilon}^{Y, p}_{e, t}
			\right]
		=
		E\left[
			\hat{w}^{c,p}_{e, t}\{m_{e, t}^{c, p}(X) - \hat{m}_{e, t}^{c, p}(X)\}
			\right].
	\end{align*}
	Lastly,
	\begin{align*}
		\tau^{dr,p,num}_{e, t, 0}
		=
		E[w^{c,p}_{e, t}\varepsilon^{Y, p}_{e, t}]
		=
		E[w^{c,p}_{e, t}\{m_{e, t}^{c, p}(X) - m_{e, t}^{c, p}(X)\}]
		= 0,
	\end{align*}
	where we implicitly use that the mean function for the controls is
	correctly specified,
	so
	\begin{align}
		\label{eq:estimand-left-panel}
		\tau^{dr,p,num}_{e, t}
		=
		E[ w^{trt,p}_{e} \{m_{e, t}^{trt, p}(X) - m_{e, t}^{c, p}(X)\} ]
		= E[
				m_{e, t}^{trt, p}(X) - m_{e, t}^{c, p}(X) \mid E_{e} = 1
			]
	\end{align}
	i.e. RHS of \cref{eq:identification-panel}.
	Hence,
	\begin{align}
		\label{eq:remainder-num-panel-compare}
		 & P \varphi^{dr, p, num}(\cdot; \tau^{dr, p, num}_{e, t},
		\hat{\eta}^{dr, p, num}_{e, t})
		\\
		\nonumber
		 & =
		E[
				\hat{w}^{trt,p}_{e}
				\{
				\hat{\varepsilon}^{Y, p}_{e, t} - \tau^{dr,p,num}_{e, t, 1}
				\}
				-
				\hat{w}^{c,p}_{e, t}
				\{
				\hat{\varepsilon}^{Y, p}_{e, t} - \tau^{dr,p,num}_{e, t, 0}
				\}
			]
		\\
		\nonumber
		 & =
		E[
				\hat{w}^{trt,p}_{e}
				\{
				m_{e, t}^{trt, p}(X) - \hat{m}_{e, t}^{c, p}(X)
				- \{m_{e, t}^{trt, p}(X) - m_{e, t}^{c, p}(X)\}
				\}
				-
				\hat{w}^{c,p}_{e, t}
				\{m_{e, t}^{c, p}(X) - \hat{m}_{e, t}^{c, p}(X)\}
			]
		\\
		\nonumber
		 & =
		E[
				(
				\hat{w}^{trt,p}_{e} - \hat{w}^{c,p}_{e, t}
				)
				\{
				m_{e, t}^{c, p}(X) - \hat{m}_{e, t}^{c, p}(X)
				\}
			]
	\end{align}
	Next, the expected value of the weights equals:
	\begin{align*}
		E [\hat{w}^{trt,p}_{e} - \hat{w}^{c,p}_{e, t}]
		 & =
		E \left[\frac{E_{e}}{\hat{\rho}}
			- \frac{C}{\hat{\rho}_{0}}
			\frac{\hat{p}(X)}{1 - \hat{p}(X)}
			\right]
		\\
		 & =
		E \left[\frac{p(X)}{\hat{\rho}}
			- \frac{1 - p(X)}{\hat{\rho}_{0}}
			\frac{\hat{p}(X)}{1 - \hat{p}(X)}
			\right]
		\\
		 & =
		E \left[
			\{1 - p(X)\}
			\left[
				\frac{1}{\hat{\rho}}\frac{p(X)}{1 - p(X)}
				- \frac{1}{\hat{\rho}_{0}}
				\frac{\hat{p}(X)}{1 - \hat{p}(X)}
				\right]
			\right]
	\end{align*}
	where we used the LIE in the second equals sign, that
	$E[C \mid X] = 1 - p(X)$ for a generic control variable and corresponding
	propensity score,
	and defined $\hat{\rho}_{0}$ as the estimated value of the denominator in the
	corresponding control weight.
	Therefore:
	\begin{align*}
		 & E[
				\hat{w}^{trt,p}_{e}
				\{
				\hat{\varepsilon}^{Y, p}_{e, t} - \tau^{dr,p,num}_{e, t, 1}
				\}
				-
				\hat{w}^{c,p}_{e, t}
				\{
				\hat{\varepsilon}^{Y, p}_{e, t} - \tau^{dr,p,num}_{e, t, 0}
				\}
			]
		\\
		 & =
		E \left[
			\{1 - p(X)\}
			\left[
				\frac{1}{\hat{\rho}}\frac{p(X)}{1 - p(X)}
				- \frac{1}{\hat{\rho}_{0}}
				\frac{\hat{p}(X)}{1 - \hat{p}(X)}
				\right]
			\{
			m_{e, t}^{c, p}(X) - \hat{m}_{e, t}^{c, p}(X)
			\}
			\right]
	\end{align*}
	The exact same argument applies to the denominator, yielding:
	\begin{align*}
		 & E[
				\hat{w}^{trt,p}_{e}
				\{
				\hat{\varepsilon}^{D, p}_{e, t} - \tau^{dr,p,den}_{e, t, 1}
				\}
				-
				\hat{w}^{c,p}_{e, t}
				\{
				\hat{\varepsilon}^{D, p}_{e, t} - \tau^{dr,p,den}_{e, t, 0}
				\}
			]
		\\
		 & =
		E \left[
			\{1 - p(X)\}
			\left[
				\frac{1}{\hat{\rho}}\frac{p(X)}{1 - p(X)}
				- \frac{1}{\hat{\rho}_{0}}
				\frac{\hat{p}(X)}{1 - \hat{p}(X)}
				\right]
			\{
			g_{e, t}^{c, p}(X) - \hat{g}_{e, t}^{c, p}(X)
			\}
			\right]
	\end{align*}
	Using the above, we can bound the remainder term as:
	\begin{align*}
		 & \vert P \varphi^{dr, p}(\cdot; \tau^{dr, p}, \hat{\eta}^{dr, p}) \vert
		\\
		 & \leq
		\Big\vert \frac{1}{\tau^{dr, p, den}_{e, t}} \Big\vert
		E \Big[
			\Big\vert
			\{1 - p(X)\}
			\left[
				\frac{1}{\hat{\rho}}\frac{p(X)}{1 - p(X)}
				- \frac{1}{\hat{\rho}_{0}}
				\frac{\hat{p}(X)}{1 - \hat{p}(X)}
				\right]
			\{
			m_{e, t}^{c, p}(X) - \hat{m}_{e, t}^{c, p}(X)
			\}
			\Big\vert
			\Big]
		\\ & \quad
		+
		\Big\vert \frac{\tau^{dr, p}}{\tau^{dr, p, den}_{e, t}} \Big\vert
		E \Big[
			\Big\vert
			\{1 - p(X)\}
			\left[
				\frac{1}{\hat{\rho}}\frac{p(X)}{1 - p(X)}
				- \frac{1}{\hat{\rho}_{0}}
				\frac{\hat{p}(X)}{1 - \hat{p}(X)}
				\right]
			\{
			g_{e, t}^{c, p}(X) - \hat{g}_{e, t}^{c, p}(X)
			\}
			\Big\vert
			\Big]
		\\
		 & \leq
		\Big\vert \frac{1}{\tau^{dr, p, den}_{e, t}} \Big\vert
		\frac{1}{c}
		E \Big[
			\Big\vert
			\left[
				\frac{1}{\hat{\rho}}\frac{p(X)}{1 - p(X)}
				- \frac{1}{\hat{\rho}_{0}}
				\frac{\hat{p}(X)}{1 - \hat{p}(X)}
				\right]
			\{
			m_{e, t}^{c, p}(X) - \hat{m}_{e, t}^{c, p}(X)
			\}
			\Big\vert
			\Big]
		\\ & \quad
		+ \Big\vert \frac{\tau^{dr, p}}{\tau^{dr, p, den}_{e, t}} \Big\vert
		\frac{1}{c}
		E \Big[
			\Big\vert
			\left[
				\frac{1}{\hat{\rho}}\frac{p(X)}{1 - p(X)}
				- \frac{1}{\hat{\rho}_{0}}
				\frac{\hat{p}(X)}{1 - \hat{p}(X)}
				\right]
			\{
			g_{e, t}^{c, p}(X) - \hat{g}_{e, t}^{c, p}(X)
			\}
			\Big\vert
			\Big]
		\\
		 & \leq
		\Big\vert \frac{1}{\tau^{dr, p, den}_{e, t}} \Big\vert
		\frac{1}{c}
		\Big\Vert
		\frac{1}{\hat{\rho}}\frac{p}{1 - p}
		- \frac{1}{\hat{\rho}_{0}}
		\frac{\hat{p}}{1 - \hat{p}}
		\Big\Vert_{2, P}
		\Big\Vert
		m_{e, t}^{c, p} - \hat{m}_{e, t}^{c, p}
		\Big\Vert_{2, P}
		\\ & \quad
		+ \Big\vert \frac{\tau^{dr, p}}{\tau^{dr, p, den}_{e, t}} \Big\vert
		\frac{1}{c}
		\Big\Vert
		\frac{1}{\hat{\rho}}\frac{p}{1 - p}
		- \frac{1}{\hat{\rho}_{0}}
		\frac{\hat{p}}{1 - \hat{p}}
		\Big\Vert_{2, P}
		\Big\Vert
		g_{e, t}^{c, p} - \hat{g}_{e, t}^{c, p}
		\Big\Vert_{2, P}
	\end{align*}
	where we
	in the first inequality used
	Jensen's inequality and the triangle inequality for integrals,
	in the second inequality used
	Hölder's inequality and that
	$1 - p(X) \leq 1 < 1/c$ cf. \Cref{ass:overlap},
	and
	in the third inequality used Hölder's inequality again.
	Provided the assumptions on the nuisances of \Cref{prop:remainder-term-panel}
	hold the proposition now follows.
\end{proofref}

\subsubsection{Remainder Repeated Cross-Sections}
\label{sec:remainder-rc}
\begin{proofref}{prop:remainder-term-rc}
	Consider the DR estimand \cref{eq:dr-estimand-rc}
	with influence function \cref{eq:if-dr-estimand-rc}.
	Write the estimand as $\tau^{dr, rc}_{e, t} =
		\tau^{dr, rc,num}_{e, t}/\tau^{dr, rc,den}_{e, t}$.
	As in the panel case, we consider the numerator and denominator separately.

	First we rewrite the numerator of the estimand in \cref{eq:dr-estimand-rc}
	into another form.
	Write out \cref{eq:kappa-y-rc-exp} recalling mean function definitions
	\cref{eq:mean-functions-rc-diffs}:
	\begin{align*}
		 & \kappa_{e, t}^{Y, rc}
		\\
		 & =
		E\left[\{w^{trt,p}_{e} - w^{trt,rc}_{e, t}\}
			\{m_{e, t}^{trt, rc}(X) -  m_{e, t, t}^{c, rc}(X)\}\right]
		-
		E\left[[w^{trt,p}_{e} - w^{trt,rc}_{e, e-1}]
			\{m_{e, e-1}^{trt, rc}(X) -  m_{e, t, e-1}^{c, rc}(X)\}\right]
		\\
		 & =
		E\left[
			w^{trt,rc}_{e, e-1}
			\{m_{e, e-1}^{trt, rc}(X) -  m_{e, t, e-1}^{c, rc}(X)\}
			- w^{trt,rc}_{e, t}
			\{m_{e, t}^{trt, rc}(X) -  m_{e, t, t}^{c, rc}(X)\}
			\right]
		\\ & \quad
		+
		E\left[
		w^{trt,p}_{e}
		[
		m_{e, \Delta}^{trt, rc}(X)
		-
		m_{e, t, \Delta}^{c, rc}(X)
		]
		\right]
	\end{align*}
	Next, write out the term with the residuals \cref{eq:resid-rc-case-individual}:
	\begin{align*}
		E \left[
			\{w^{trt,rc}_{e} - w^{c,rc}_{e}\}
			\varepsilon^{Y, rc}_{e, t}
			\right]
		 & =
		E \Big[
		w^{trt,rc}_{e, t}
		[Y - m_{e, t, t}^{c,rc}(X)]
		-
		w^{trt,rc}_{e, e-1}
		[Y - m_{e, t, e - 1}^{c,rc}(X)]
		\\ & \quad
		- w^{c,rc}_{e, t, t}
		[Y - m_{e, t, t}^{c,rc}(X)]
		+ w^{c,rc}_{e, t, e - 1}
		[Y - m_{e, t, e - 1}^{c,rc}(X)]
		\Big]
	\end{align*}
	Using the expressions above we see, as before, that
	the mean functions for the controls under the treated rc-weights cancel:
	\begin{align}
		\nonumber
		\tau^{dr, rc, num}_{e, t}
		          & =
		E \left[
			\{w^{trt,rc}_{e} - w^{c,rc}_{e}\}
			\varepsilon^{Y, rc}_{e, t}
			\right]
		+
		\kappa_{e, t}^{Y, rc}
		\\
		\nonumber
		          & =
		E \Big[
		w^{trt,rc}_{e, t}
		[Y - m_{e, t}^{trt, rc}(X)]
		-
		w^{trt,rc}_{e, e-1}
		[Y - m_{e, e-1}^{trt, rc}(X)]
		\\
		\nonumber & \quad
		- w^{c,rc}_{e, t, t}
		[Y - m_{e, t, t}^{c,rc}(X)]
		+ w^{c,rc}_{e, t, e - 1}
		[Y - m_{e, t, e - 1}^{c,rc}(X)]
		\Big]
		+
		E\left[
		w^{trt,p}_{e}
		[
		m_{e, \Delta}^{trt, rc}(X)
		-
		m_{e, t, \Delta}^{c, rc}(X)
		]
		\right]
		\\
		          & =
		\label{eq:rc-eps-terms-num}
		E \Big[
		w^{trt,rc}_{e, t}
		\varepsilon^{Y,trt,rc}_{e,t}
		-
		w^{trt,rc}_{e, e-1}
		\varepsilon^{Y,trt,rc}_{e,e-1}
		- w^{c,rc}_{e, t, t}
		\varepsilon^{Y,c,rc}_{e,t,t}
		\\
		\nonumber & \quad
		+ w^{c,rc}_{e, t, e - 1}
		\varepsilon^{Y,c,rc}_{e,t,e - 1}
		+
		w^{trt,p}_{e}
		[
		m_{e, \Delta}^{trt, rc}(X)
		-
		m_{e, t, \Delta}^{c, rc}(X)
		]
		\Big].
	\end{align}
	Next, we show the terms involving the residual $\varepsilon$ equal $0$ in
	expectation,.
	Recall the weights: \cref{eq:weights-trt-rc-norm} and consider:
	\begin{align*}
		E[w^{trt,rc}_{e, t} \varepsilon^{Y,trt,rc}_{e,t} ]
		 & =
		E\left[
			\frac{E_{e} \cdot \mathbf{1}\{T = t\}}{
				E[E_{e} \cdot \mathbf{1}\{T = t\}]
			}
			\varepsilon^{Y,trt,rc}_{e,t}
			\right]
		\\
		 & =
		E\left[
			\frac{E_{e} \cdot \mathbf{1}\{T = t\}}{
				E[E_{e} \cdot \mathbf{1}\{T = t\}]
			}
			[Y - m_{e, t}^{trt, rc}(X)]
			\right]
		\\
		 & =
		E\left[
			\frac{E_{e} \cdot \mathbf{1}\{T = t\}}{
				E[E_{e} \cdot \mathbf{1}\{T = t\}]
			}
			[m_{e, t}^{trt, rc}(X) - m_{e, t}^{trt, rc}(X)]
			\right]
		= 0
	\end{align*}
	using the LIE; similar calculation shows the other weighted residual terms in
	\cref{eq:rc-eps-terms-num} equal $0$.
	Hence we are left with:
	\begin{align}
		\label{eq:estimand-left-rc}
		\tau^{dr, rc, num}_{e, t}
		 & =
		E
		\Big[
		w^{trt,p}_{e}
		[
		m_{e, \Delta}^{trt, rc}(X)
		-
		m_{e, t, \Delta}^{c, rc}(X)
		]
		\Big]
		=:
		\tau^{dr, rc, num}_{e, t, 1} - \tau^{dr, rc, num}_{e, t, 0}
	\end{align}
	matching the numerator of \cref{eq:identification-rc} (which we already knew).

	Can write the IF \cref{eq:if-dr-num-estimand-rc} using the notation for the
	implicit estimand \cref{eq:estimand-left-rc} as:
	\begin{align}
		 & \varphi^{dr,rc}_{num}(O; \tau^{dr, rc, num}_{e, t}; \eta^{dr, rc,
			num}_{e, t})
		\\
		\nonumber
		 & =
		w^{trt,rc}_{e, t}
		\varepsilon^{Y,trt,rc}_{e,t}
		+ w^{trt,p}_{e} \left\{ m_{e, \Delta}^{trt, rc}(X)
		- \tau^{dr, rc, num}_{e, t, 1} \right\}
		\\ & \quad
		\nonumber
		-
		\left[w^{trt,rc}_{e, e-1}
			\varepsilon^{Y,trt,rc}_{e,e-1}
			+ w^{trt,p}_{e} \left\{ m_{e, t, \Delta}^{c, rc}(X)
			- \tau^{dr, rc, num}_{e, t, 0} \right\}\right]
		\\
		\nonumber
		 & \quad
		- [w^{c,rc}_{e, t, t}
		\varepsilon^{Y,c,rc}_{e,t,t}
		- w^{c,rc}_{e, t, e - 1}
		\varepsilon^{Y,c,rc}_{e,t,e-1}
		]
		\\
		\nonumber
		 & =
		w^{trt,rc}_{e, t}
		\varepsilon^{Y,trt,rc}_{e,t}
		-
		w^{trt,rc}_{e, e-1}
		\varepsilon^{Y,trt,rc}_{e,e-1}
		- [w^{c,rc}_{e, t, t}
		\varepsilon^{Y,c,rc}_{e,t,t}
		- w^{c,rc}_{e, t, e - 1}
		\varepsilon^{Y,c,rc}_{e,t,e-1}
		]
		\\ & \quad
		\nonumber
		+ w^{trt,p}_{e}  m_{e, \Delta}^{trt, rc}(X)
		- w^{trt,p}_{e}  m_{e, t, \Delta}^{c, rc}(X)
		- w^{trt,p}_{e} \tau^{dr, rc, num}_{e, t}
	\end{align}

	Thus the relevant part of the influence function is exactly:
	\begin{align}
		\label{eq:remainder-num-rc}
		 & P \varphi^{dr,rc}_{num}(\cdot; \tau^{dr, rc, num}_{e, t};
		\hat{\eta}^{dr, rc}_{num})
		\\
		\label{eq:remainder-num-rc-compare}
		 & =
		E \Big[
			\hat{w}^{trt,rc}_{e, t} \hat{\varepsilon}^{Y,trt,rc}_{e,t}
			- \hat{w}^{trt,rc}_{e, e-1} \hat{\varepsilon}^{Y,trt,rc}_{e,e-1}
			- \hat{w}^{c,rc}_{e, t, t} \hat{\varepsilon}^{Y,c,rc}_{e,t,t}
			+ \hat{w}^{c,rc}_{e, t, e - 1} \hat{\varepsilon}^{Y,c,rc}_{e,t,e-1}
		\\ & \quad
			\nonumber
			+ \hat{w}^{trt,p}_{e}  \hat{m}_{e, \Delta}^{trt, rc}(X)
			- \hat{w}^{trt,p}_{e}  \hat{m}_{e, t, \Delta}^{c, rc}(X)
			- \hat{w}^{trt,p}_{e} \tau^{dr, rc, num}_{e, t}
			\Big]
	\end{align}
	\medskip\noindent
	As done for the panel case in \cref{eq:dummy-times-estimand-trick-panel},
	\begin{align}
		\nonumber
		\tau^{dr, rc, num}_{e, t}
		 & =
		E
		\Big[
			\frac{E_{e}}{E[E_{e}]}
			[
				m_{e, \Delta}^{trt, rc}(X)
				-
				m_{e, t, \Delta}^{c, rc}(X)
			]
			\Big]
		\\
		\iff
		\label{eq:dummy-times-estimand-trick-rc}
		E[E_{e}\tau^{dr, rc, num}_{e, t}]
		 & =
		E
		\Big[
			E_{e}
				[
					m_{e, \Delta}^{trt, rc}(X)
					-
					m_{e, t, \Delta}^{c, rc}(X)
				]
			\Big]
	\end{align}

	\medskip\noindent
	Now:
	\begin{align}
		\label{eq:remainder-num-rc-components}
		E [
				\hat{w}^{trt,rc}_{e, t} \hat{\varepsilon}^{Y,trt,rc}_{e,t}
			]
		= E [
				\hat{w}^{trt,rc}_{e, t} \{Y - \hat{m}_{e, t}^{trt, rc}(X)\}
			]
		= E [
		\hat{w}^{trt,rc}_{e, t} [m_{e, t}^{trt, rc}(X) - \hat{m}_{e, t}^{trt, rc}(X)]
		]
	\end{align}
	by the LIE, and similar for the other three terms:
	\begin{align}
		\label{eq:remainder-num-rc-components-other}
		E [
				\hat{w}^{trt,rc}_{e, e-1} \hat{\varepsilon}^{Y,trt,rc}_{e,e-1}
			]
		 & =
		E [
				\hat{w}^{trt,rc}_{e, e-1}
				\{m_{e, e-1}^{trt, rc}(X) - \hat{m}_{e, e-1}^{trt, rc}(X)\}
			]
		\\
		\nonumber
		E [
				\hat{w}^{c,rc}_{e, t, t} \hat{\varepsilon}^{Y,c,rc}_{e,t,t}
			]
		 & =
		E [
				\hat{w}^{c,rc}_{e, t, t}
				\{m_{e, t}^{c, rc}(X) - \hat{m}_{e, t}^{c, rc}(X)\}
			]
		\\
		\nonumber
		E [
				\hat{w}^{c,rc}_{e, t, e-1} \hat{\varepsilon}^{Y,c,rc}_{e,t,e-1}
			]
		 & =
		E [
				\hat{w}^{c,rc}_{e, t, e-1}
				\{m_{e, e-1}^{c, rc}(X) - \hat{m}_{e, e-1}^{c, rc}(X)\}
			]
	\end{align}

	Hence, we can rewrite the remainder term \cref{eq:remainder-num-rc},
	using \cref{eq:dummy-times-estimand-trick-rc} and
	\cref{eq:remainder-num-rc-components,eq:remainder-num-rc-components-other},
	as:
	\begin{align}
		\nonumber
		          & P \varphi^{dr,rc}_{num}(\cdot; \tau^{dr, rc, num}_{e, t}; \hat{\eta}^{dr, rc}_{num})
		\\
		\nonumber
		          & =
		E \Big[
		\hat{w}^{trt,rc}_{e, t} [m_{e, t}^{trt, rc}(X) - \hat{m}_{e, t}^{trt, rc}(X)]
		-
		\hat{w}^{trt,rc}_{e, e-1}
		\{m_{e, e-1}^{trt, rc}(X) - \hat{m}_{e, e-1}^{trt, rc}(X)\}
		\\
		\nonumber & \quad
		- \hat{w}^{c,rc}_{e, t, t}
		\{m_{e, t}^{c, rc}(X) - \hat{m}_{e, t}^{c, rc}(X)\}
		+
		\hat{w}^{c,rc}_{e, t, e-1}
		\{m_{e, e-1}^{c, rc}(X) - \hat{m}_{e, e-1}^{c, rc}(X)\}
		\\
		\nonumber & \quad
		+ \hat{w}^{trt,p}_{e}
		[\hat{m}_{e, t}^{trt, rc}(X) - \hat{m}_{e, e - 1}^{trt, rc}(X)
		- \{\hat{m}_{e, t, t}^{c, rc}(X) - \hat{m}_{e, t, e - 1}^{c, rc}(X)\}]
		\\
		\nonumber & \quad
		- \hat{w}^{trt,p}_{e}
		[
		m_{e, t}^{trt, rc}(X) - m_{e, e-1}^{trt, rc}(X)
		-
		\{m_{e, t, t}^{c, rc}(X) - m_{e, t, e-1}^{c, rc}(X)\}
		]
		\Big]
		\\
		\label{eq:remainder-terms-product}
		          & =
		E \Big[
			(\hat{w}^{trt,rc}_{e, t} - \hat{w}^{trt,p}_{e})
			[m_{e, t}^{trt, rc}(X) - \hat{m}_{e, t}^{trt, rc}(X)]
		\\ & \quad
			\nonumber
			-
			(\hat{w}^{trt,rc}_{e, e-1} - \hat{w}^{trt,p}_{e})
			\{m_{e, e-1}^{trt, rc}(X) - \hat{m}_{e, e-1}^{trt, rc}(X)\}
		\\
		\nonumber & \quad
			- (\hat{w}^{c,rc}_{e, t, t} - \hat{w}^{trt,p}_{e})
			\{m_{e, t}^{c, rc}(X) - \hat{m}_{e, t}^{c, rc}(X)\}
		\\ & \quad
			\nonumber
			+
			(\hat{w}^{c,rc}_{e, t, e-1} - \hat{w}^{trt,p}_{e})
			\{m_{e, e-1}^{c, rc}(X) - \hat{m}_{e, e-1}^{c, rc}(X)\}
			\Big]
	\end{align}
	The expression in \cref{eq:remainder-terms-product} has the correct form for
	the mean functions;
	that is, the true mean function minus the estimated one.
	What remains is to derive an expression for the weights.

	To this end, recall the weight definitions \cref{eq:weights-trt-rc-norm,eq:weights-nev-rc-norm,eq:weights-nye-rc-norm}.
	For notational purposes we refine the denominators of the weights as follows:
	\begin{align*}
		w^{trt,rc}_{e, t}
		:= E_{e} \cdot \mathbf{1}\{T = t\}/\rho^{trt, rc}_{e, t},
		\quad
		w^{c,rc}_{e, t}
		 & :=
		\frac{C \cdot \mathbf{1}\{T = t\}p(X)}{1 - p(X)}
		/
		\rho^{c, rc}_{e, t}
	\end{align*}
	and likewise
	\begin{align*}
		\hat{w}^{trt,rc}_{e, t}
		:= E_{e} \cdot \mathbf{1}\{T = t\}/\hat{\rho}^{trt, rc}_{e, t},
		\quad
		\hat{w}^{c,rc}_{e, t}
		 & :=
		\frac{C \cdot \mathbf{1}\{T = t\}\hat{p}(X)}{1 - \hat{p}(X)}
		/
		\hat{\rho}^{c, rc}_{e, t}
	\end{align*}
	where we defined
	$\hat{\rho}^{c, rc}_{e, t}
		:=
		P_{n}
		\left[\frac{C \cdot \mathbf{1}\{T = t\}\hat{p}(X)}{1 - \hat{p}(X)}\right]$
	for an estimator of the denominator of a generic control weight in the
	repeated cross-sections case
	\cref{eq:weights-nev-rc-norm,eq:weights-nye-rc-norm}.

	Then,
	\begin{align*}
		E[\hat{w}^{trt,rc}_{e, t}]
		=
		E[E[E_{e} \mathbf{1}\{T = t\} \mid X]/\hat{\rho}^{trt, rc}_{e, t}]
		= E[p(X) \lambda_{t}/\hat{\rho}^{trt, rc}_{e, t}]
	\end{align*}
	and
	\begin{align*}
		E[\hat{w}^{c,rc}_{e, t}]
		=
		E\left[[1 - p(X)] \lambda_{t}\frac{\hat{p}(X)}{1 - \hat{p}(X)}
			/\hat{\rho}^{c, rc}_{e, t}\right].
	\end{align*}
	Also, $\hat{w}^{trt,p}_{e}$ is similar to the panel case,
	where we define $\rho := E[E_{e}]$ and $\hat{\rho}$ its estimated
	value.
	Hence:
	\begin{align*}
		E[
		\hat{w}^{c,rc}_{e, t}
		- \hat{w}^{trt,p}_{e}
		]
		 & =
		E[
				[1 - p(X)] \lambda_{t}\frac{\hat{p}(X)}{1 - \hat{p}(X)}
				/\hat{\rho}^{c, rc}_{e, t}
				- p(X)/\hat{\rho}
			]
		\\
		 & =
		- E\left[
			[1 - p(X)]
			\left(
			\frac{p(X)}{1 - p(X)}
			/\hat{\rho}
			-
			\lambda_{t}
			\frac{\hat{p}(X)}{1 - \hat{p}(X)}
			/\hat{\rho}^{c, rc}_{e, t}
			\right)
			\right]
	\end{align*}
	and similarly:
	\begin{align*}
		E[
		\hat{w}^{c,rc}_{e, e-1}
		- \hat{w}^{trt,p}_{e}
		]
		 & =
		- E\left[
			[1 - p(X)]
			\left(
			\frac{p(X)}{1 - p(X)}
			/\hat{\rho}
			-
			\lambda_{e - 1}
			\frac{\hat{p}(X)}{1 - \hat{p}(X)}
			/\hat{\rho}^{c, rc}_{e, e - 1}
			\right)
			\right]
	\end{align*}
	The last two difference-in-weights of \cref{eq:remainder-terms-product},
	corresponding to the treated cases in periods $t$ and $e - 1$,
	are rewritten as follows:
	\begin{align*}
		E[\hat{w}^{trt,rc}_{e, t} - \hat{w}^{trt,p}_{e}]
		 & = E[
				p(X) \lambda_{t} / \hat{\rho}^{trt, rc}_{e, t}
				- p(X)/\hat{\rho}
			]
		= E[
				p(X)
				(
				\lambda_{t} / \hat{\rho}^{trt, rc}_{e, t} - 1/\hat{\rho}
				)
			],
		\\
		E[\hat{w}^{trt,rc}_{e, e-1} - \hat{w}^{trt,p}_{e}]
		 & = E[
				p(X)
				(
				\lambda_{e - 1} / \hat{\rho}^{trt, rc}_{e, e - 1} - 1/\hat{\rho}
				)
			].
	\end{align*}
	Hence, taking the absolute value of \cref{eq:remainder-terms-product}, we can
	bound it as:
	\begin{align}
		\nonumber
		          & \vert P \varphi^{dr,rc}_{num}(\cdot; \tau^{dr, rc, num}_{e, t}; \hat{\eta}^{dr, rc}_{num}) \vert
		\\
		\nonumber
		          & = \Big\vert
		E \Big[
		p(X)
		(
		\lambda_{t} / \hat{\rho}^{trt, rc}_{e, t} - 1/\hat{\rho}
		)
		[m_{e, t}^{trt, rc}(X) - \hat{m}_{e, t}^{trt, rc}(X)]
		\\
		\nonumber & \quad
		-
		p(X)
		(
		\lambda_{e - 1} / \hat{\rho}^{trt, rc}_{e, e - 1} - 1/\hat{\rho}
		)
		\{m_{e, e-1}^{trt, rc}(X) - \hat{m}_{e, e-1}^{trt, rc}(X)\}
		\\
		\nonumber
		\nonumber & \quad
		+ [1 - p(X)]
		\left(
		\frac{p(X)}{1 - p(X)}
		/\hat{\rho}
		-
		\lambda_{t}
		\frac{\hat{p}(X)}{1 - \hat{p}(X)}
		/\hat{\rho}^{c, rc}_{e, t}
		\right)
		\{m_{e, t}^{c, rc}(X) - \hat{m}_{e, t}^{c, rc}(X)\}
		\\
		\nonumber & \quad
		- [1 - p(X)]
		\left(
		\frac{p(X)}{1 - p(X)}
		/\hat{\rho}
		-
		\lambda_{e - 1}
		\frac{\hat{p}(X)}{1 - \hat{p}(X)}
		/\hat{\rho}^{c, rc}_{e, e - 1}
		\right)
		\{m_{e, e-1}^{c, rc}(X) - \hat{m}_{e, e-1}^{c, rc}(X)\}
		\Big]
		\Big\vert
		\\
		\nonumber
		          & \leq
		E \Big[
			\Big\vert
			p(X)
			(
			\lambda_{t} / \hat{\rho}^{trt, rc}_{e, t} - 1/\hat{\rho}
			)
			[m_{e, t}^{trt, rc}(X) - \hat{m}_{e, t}^{trt, rc}(X)]
			\Big\vert
		\\
		\nonumber & \quad
			+
			\Big\vert
			p(X)
			(
			\lambda_{e - 1} / \hat{\rho}^{trt, rc}_{e, e - 1} - 1/\hat{\rho}
			)
			\{m_{e, e-1}^{trt, rc}(X) - \hat{m}_{e, e-1}^{trt, rc}(X)\}
			\Big\vert
		\\
			\nonumber
		\nonumber & \quad
			+
			\Big\vert
			[1 - p(X)]
			\left(
			\frac{p(X)}{1 - p(X)}
			/\hat{\rho}
			-
			\lambda_{t}
			\frac{\hat{p}(X)}{1 - \hat{p}(X)}
			/\hat{\rho}^{c, rc}_{e, t}
			\right)
			\{m_{e, t}^{c, rc}(X) - \hat{m}_{e, t}^{c, rc}(X)\}
			\Big\vert
		\\
		\nonumber & \quad
			+
			\Big\vert
			[1 - p(X)]
			\left(
			\frac{p(X)}{1 - p(X)}
			/\hat{\rho}
			-
			\lambda_{e - 1}
			\frac{\hat{p}(X)}{1 - \hat{p}(X)}
			/\hat{\rho}^{c, rc}_{e, e - 1}
			\right)
			\{m_{e, e-1}^{c, rc}(X) - \hat{m}_{e, e-1}^{c, rc}(X)\}
			\Big\vert
			\Big]
		\\
		\nonumber
		          & \leq
		\Big\vert
		p
		(
		\lambda_{t} / \hat{\rho}^{trt, rc}_{e, t} - 1/\hat{\rho}
		)
		\Big\vert
		\Big\Vert
		m_{e, t}^{trt, rc} - \hat{m}_{e, t}^{trt, rc}
		\Big\Vert_{2, P}
		+
		\Big\vert
		p
		(
		\lambda_{e - 1} / \hat{\rho}^{trt, rc}_{e, e - 1} - 1/\hat{\rho}
		)
		\Big\vert
		\Big\Vert
		m_{e, e-1}^{trt, rc} - \hat{m}_{e, e-1}^{trt, rc}
		\Big\Vert_{2, P}
		\\
		\nonumber
		\nonumber & \quad
		+
		\Big\vert
		[1 - p]
		\left(
		\frac{p}{1 - p}
		/\hat{\rho}
		-
		\lambda_{t}
		\frac{\hat{p}}{1 - \hat{p}}
		/\hat{\rho}^{c, rc}_{e, t}
		\right)
		\Big\vert
		\Big\Vert
		m_{e, t}^{c, rc} - \hat{m}_{e, t}^{c, rc}
		\Big\Vert_{2, P}
		\\
		\nonumber & \quad
		+
		\Big\vert
		[1 - p]
		\left(
		\frac{p}{1 - p}
		/\hat{\rho}
		-
		\lambda_{e - 1}
		\frac{\hat{p}}{1 - \hat{p}}
		/\hat{\rho}^{c, rc}_{e, e - 1}
		\right)
		\Big\vert
		\Big\Vert
		m_{e, e-1}^{c, rc} - \hat{m}_{e, e-1}^{c, rc}
		\Big\Vert_{2, P}
		\\
		\label{eq:rc-num-bound-remainder}
		          & \leq
		(1 - c)
		\Big\{
		\Big\vert
		\lambda_{t} / \hat{\rho}^{trt, rc}_{e, t} - 1/\hat{\rho}
		\Big\vert
		\Big\Vert
		m_{e, t}^{trt, rc} - \hat{m}_{e, t}^{trt, rc}
		\Big\Vert_{2, P}
		\\
		\nonumber & \quad
		+
		\Big\vert
		\lambda_{e - 1} / \hat{\rho}^{trt, rc}_{e, e - 1} - 1/\hat{\rho}
		\Big\vert
		\Big\Vert
		m_{e, e-1}^{trt, rc} - \hat{m}_{e, e-1}^{trt, rc}
		\Big\Vert_{2, P}
		\Big\}
		\\
		\nonumber
		\nonumber & \quad
		+
		\frac{1}{c}
		\Big\{
		\Big\Vert
		\frac{p}{1 - p}
		/\hat{\rho}
		-
		\lambda_{t}
		\frac{\hat{p}}{1 - \hat{p}}
		/\hat{\rho}^{c, rc}_{e, t}
		\Big\Vert_{2, P}
		\Big\Vert
		m_{e, t}^{c, rc} - \hat{m}_{e, t}^{c, rc}
		\Big\Vert_{2, P}
		\\
		\nonumber & \quad
		+
		\Big\Vert
		\frac{p}{1 - p}
		/\hat{\rho}
		-
		\lambda_{e - 1}
		\frac{\hat{p}}{1 - \hat{p}}
		/\hat{\rho}^{c, rc}_{e, e - 1}
		\Big\Vert_{2, P}
		\Big\vert
		\Big\Vert
		m_{e, e-1}^{c, rc} - \hat{m}_{e, e-1}^{c, rc}
		\Big\Vert_{2, P}
		\Big\}
	\end{align}
	where we
	in the first inequality used
	Jensen's inequality and the triangle inequality for integrals,
	in the second inequality used
	Hölder's inequality,
	in the third
	Hölder's inequality again together with
	$1 - p(X) < 1/c$ and $p(X) < 1 - c$ cf. \Cref{ass:overlap}.

	\medskip\noindent
	The exact same argument applies to the denominator,
	providing a bound on $\vert P \varphi^{dr,rc}_{den}(\cdot; \tau^{dr, rc,
			den}_{e, t}; \hat{\eta}^{dr, rc}_{den}) \vert$ for
	\Cref{eq:if-dr-den-estimand-rc}.
	It has the exact same form as \cref{eq:rc-num-bound-remainder}
	but with the mean function related to the denominator instead of the
	numerator.
	Hence, similar to the panel data case, by the quotient rule,
	the unscaled remainder term satisfies:
	\begin{align*}
		\Big\vert P \varphi^{dr,rc}(\cdot; \tau^{dr, rc}; \hat{\eta^{dr, rc}}) \Big\vert
		= o_{P}(n^{-1/2})
	\end{align*}
	provided the assumptions in \Cref{prop:remainder-term-rc} hold.
\end{proofref}

\begin{remark}[Identical remainder terms]
	\label{remark:identical-remainders}
	\leavevmode

	\noindent
	\textbf{Panel Data:}
	\Cref{eq:estimand-left-panel}
	equals \cref{eq:identification-panel};
	and
	from \Cref{eq:remainder-num-rc-compare}
	we see
	\begin{align*}
		 & P \varphi^{dr, p, num}(\cdot; \tau^{dr, p, num}_{e, t},
		\hat{\eta}^{dr, p, num}_{e, t})
		\\
		 & =
		E[
				(
				\hat{w}^{trt,p}_{e}
				-
				\hat{w}^{c,p}_{e, t}
				)
				\hat{\varepsilon}^{Y, p}_{e, t}
				- \hat{w}^{trt,p}_{e} \tau^{dr,p,num}_{e, t, 1}
			]
		= P \varphi^{p, num}(\cdot; \tau^{p, num}_{e, t},
		\hat{\eta}^{p, num}_{e, t})
	\end{align*}
	i.e. the numerator component of \cref{eq:if-latt-panel}
	is identical to \cref{eq:if-dr-estimand-rc} under the expectation operator
	with estimated nuisances inserted.
	Hence the remainder terms using either influence function
	\cref{eq:if-latt-panel}
	or \cref{eq:if-dr-estimand-panel} are the same,
	and similarly for the denominator and the joint influence function.

	\medskip\noindent
	\textbf{Repeated Cross-Sections:}
	\Cref{eq:estimand-left-rc}
	equals \cref{eq:identification-rc};
	and
	\Cref{eq:remainder-num-rc-compare}
	is exactly equal to
	$P \varphi^{rc, num}(\cdot; \tau^{rc}_{e, t}, \hat{\eta}^{rc}_{e, t})$.
	Hence the remainder terms using either influence function
	\cref{eq:if-latt-rc}
	or \cref{eq:if-dr-estimand-rc} are the same,
	i.e.
	\begin{align*}
		P \varphi^{dr,rc}_{num}(\cdot; \tau^{dr, rc, num}_{e, t};
		\hat{\eta}^{dr, rc}_{num})
		=
		P \varphi^{rc, num}(\cdot; \tau^{rc}_{e, t}, \hat{\eta}^{rc, num}_{e, t})
	\end{align*}
	and similarly for the denominator and the joint influence function.
\end{remark}

\begin{remark}[Similarity of the IDiD and DiD remainder term calculations]
	\leavevmode
	\begin{itemize}[nosep]
		\item The analysis was applied to the numerator and denominator separately
		      and then joined together.
		      Hence the same remainder-term calculations also apply to the DiD
		      setting for the $ATT(g, t)$ estimand and its influence function, cf.
		      \cite{sazhao,csa}.
		\item As for the panel case, the non-normalized \cref{eq:if-latt-rc} and normalized
		      influence function turn out to have a similar remainder term.
	\end{itemize}
\end{remark}

\newpage
\section{Other calculations}
\subsection{Plug-in bias DR EIF}
\label{eq:plug-in-extra}
Here we show that even though the DR estimator $\hat{\tau}^{dr, p}_{e, t}$
solves
\begin{align*}
	P_{n}\varphi^{p}(\cdot; \hat{\tau}^{dr, p}_{e, t}, \hat{\eta}^{p}_{e, t}) = 0
\end{align*}
using the EIF $\varphi^{p}(O; \tau^{p}_{e, t}, \eta^{p}_{e, t})$
\cref{eq:if-latt-panel},
it also holds that
\begin{align*}
	P_{n}\varphi^{dr, p}(\cdot; \hat{\tau}^{dr, p}_{e, t}, \hat{\eta}^{dr, p}_{e, t}) = 0
\end{align*}
that is, the empirical mean of
$\varphi^{dr, p}(\cdot; \hat{\tau}^{dr, p}_{e, t}, \hat{\eta}^{dr, p}_{e, t})$
equals $0$.
A similar argument shows that
\begin{align*}
	P_{n}\varphi^{dr, rc}(\cdot; \hat{\tau}^{dr, rc}_{e, t}, \hat{\eta}^{dr, rc}_{e, t}) = 0
\end{align*}
as well.

To see this, first write \cref{eq:if-latt-panel} as
\begin{align*}
	 & \varphi^{p}(O; \tau^{p}_{e, t}, \eta^{p}_{e, t})
	\\
	\nonumber
	 & :=
	\frac{1}{\tau^{p, den}_{e, t}}
	[
	\{ w^{trt,p}_{e} - w^{c,p}_{e, t} \}
	\{\Delta_{t-e+1}Y_{t} - m_{e, t}^{c, p}(X)\}
	- \tau^{p}_{e, t}
	\{ w^{trt,p}_{e} - w^{c,p}_{e, t} \}
	\{\Delta_{t-e+1}D_{t} - g_{e, t}^{c, p}(X)\}
	]
	\\
	\nonumber
	 & =
	\frac{1}{\tau^{p, den}_{e, t}}
	[
		\{ w^{trt,p}_{e} - w^{c,p}_{e, t} \} \varepsilon^{Y, p}_{e, t}
		- \tau^{p}_{e, t}
		\{ w^{trt,p}_{e} - w^{c,p}_{e, t} \} \varepsilon^{D, p}_{e, t}
	],
\end{align*}
and \cref{eq:if-dr-estimand-panel} as
\begin{align*}
	\varphi^{dr, p}(O; \tau^{dr, p}_{e, t}, \eta^{dr, p}_{e, t})
	 & =
	\frac{1}{\tau^{dr, p, den}_{e, t}}
	\Big[
		w^{trt,p}_{e}
		\{
		\varepsilon^{Y, p}_{e, t}
		-
		E[ w^{trt,p}_{e} \varepsilon^{Y, p}_{e, t} ]
		\}
		-
		w^{c,p}_{e, t}
		\{
		\varepsilon^{Y, p}_{e, t}
		-
		E[ w^{c,p}_{e, t} \varepsilon^{Y, p}_{e, t} ]
		\}
	\\ & \quad
		\nonumber
		- \tau^{dr, p}_{e, t}
		(w^{trt,p}_{e}
		\{
		\varepsilon^{D, p}_{e, t}
		-
		E[ w^{trt,p}_{e} \varepsilon^{D, p}_{e, t} ]
		\}
		-
		w^{c,p}_{e, t}
		\{
		\varepsilon^{D, p}_{e, t}
		-
		E[ w^{c,p}_{e, t} \varepsilon^{D, p}_{e, t} ]
		\})
		\Big]
\end{align*}
Consider the part corresponding to the numerator of the latter expression:
\begin{align*}
	 & w^{trt,p}_{e}
	\{
	\varepsilon^{Y, p}_{e, t}
	-
	E[ w^{trt,p}_{e} \varepsilon^{Y, p}_{e, t} ]
	\}
	-
	w^{c,p}_{e, t}
	\{
	\varepsilon^{Y, p}_{e, t}
	-
	E[ w^{c,p}_{e, t} \varepsilon^{Y, p}_{e, t} ]
	\}
	\\
	 & =
	(
	w^{trt,p}_{e}
	-
	w^{c,p}_{e, t}
	)
	\varepsilon^{Y, p}_{e, t}
	-
	w^{trt,p}_{e}
	E[ w^{trt,p}_{e} \varepsilon^{Y, p}_{e, t} ]
	+
	w^{c,p}_{e, t}
	E[ w^{c,p}_{e, t} \varepsilon^{Y, p}_{e, t} ].
\end{align*}
Taking expectations yields
\begin{align*}
	 & E\left[
		(
		w^{trt,p}_{e}
		-
		w^{c,p}_{e, t}
		)
		\varepsilon^{Y, p}_{e, t}
		-
		w^{trt,p}_{e}
		E[ w^{trt,p}_{e} \varepsilon^{Y, p}_{e, t} ]
		+
		w^{c,p}_{e, t}
		E[ w^{c,p}_{e, t} \varepsilon^{Y, p}_{e, t} ]
		\right]
	\\
	 & =
	E\left[
		(
		w^{trt,p}_{e}
		-
		w^{c,p}_{e, t}
		)
		\varepsilon^{Y, p}_{e, t}
		\right]
	-
	E[
			(w^{trt,p}_{e} - w^{c,p}_{e, t}) \varepsilon^{Y, p}_{e, t}
		]
	= 0
\end{align*}
since the expected values of the normalized weights equal $1$.
The same calculation applies to the denominator. Hence the influence function
has mean zero. Likewise, its empirical mean also equals $0$:
\begin{align*}
	 & P_{n}\left[
		(
		\hat{w}^{trt,p}_{e}
		-
		\hat{w}^{c,p}_{e, t}
		)
		\hat{\varepsilon}^{Y, p}_{e, t}
		-
		\hat{w}^{trt,p}_{e}
		P_{n}[ \hat{w}^{trt,p}_{e} \hat{\varepsilon}^{Y, p}_{e, t} ]
		+
		\hat{w}^{c,p}_{e, t}
		P_{n}[ \hat{w}^{c,p}_{e, t} \hat{\varepsilon}^{Y, p}_{e, t} ]
		\right]
	\\
	 & =
	P_{n}\left[
		(
		\hat{w}^{trt,p}_{e}
		-
		\hat{w}^{c,p}_{e, t}
		)
		\hat{\varepsilon}^{Y, p}_{e, t}
		\right]
	-
	P_{n}[ (\hat{w}^{trt,p}_{e} - \hat{w}^{c,p}_{e, t})
			\hat{\varepsilon}^{Y, p}_{e, t} ]
	\\
	 & = 0
\end{align*}
since the weights also sum to $1$ in sample. The same calculation applies to
the denominator.

Therefore, even though we solve the estimating equation for the first
influence function, the second one is also equal to $0$ in sample. It follows
that for the plug-in estimator of $\tau^{dr,s}_{e, t}$, $s \in \{p, rc\}$, the only remaining
terms to handle are the empirical process term and the remainder term. The
remainder term is the same for both influence functions.

\end{document}